\documentclass[12pt]{scrartcl}
\usepackage{a4}
\usepackage{amsthm}
\usepackage{amsmath}
\usepackage{amssymb}
\usepackage{amsfonts}
\usepackage{mathrsfs}
\usepackage{dsfont}
\usepackage{latexsym}
\usepackage{color}
\usepackage{bbm,exscale}
\definecolor{Myblue}{rgb}{0,0,0.6}
\usepackage[a4paper,colorlinks,citecolor=Myblue,linkcolor=Myblue,urlcolor=Myblue,pdfpagemode=None]{hyperref}
\usepackage[square,numbers,sort&compress]{natbib}
\usepackage[all,cmtip]{xy}
\usepackage{tikz}

  \vfuzz \hfuzz
\newcommand{\D}{\text{d}}
\newcommand{\E}{\text{e}}
\newcommand{\I}{\text{i}}
\newcommand{\C}{\mathds{C}}

\newcommand{\R}{\mathds{R}}
\newcommand{\Z}{\mathds{Z}}
\def\1{\ifmmode\mathrm{1\!l}\else\mbox{\(\mathrm{1\!l}\)}\fi}
\newcommand{\one}{\mathbbm{1}}
\newcommand{\be}{\begin{equation}}
\newcommand{\ee}{\end{equation}}
\newcommand{\bes}{\begin{equation*}}
\newcommand{\ees}{\end{equation*}}

\newcommand{\sVir}{\mathsf{sVir}}
\newcommand{\MF}{\operatorname{MF}_{\operatorname{bi}}}
\newcommand{\MFW}{\operatorname{MF}_{\operatorname{bi}}(W)}
\newcommand{\MFR}{\operatorname{MF}^\text{R}_{\operatorname{bi}}}

\newcommand{\DGW}{\operatorname{DG}_{\text{bi}}(W)}

\newcommand{\tc }{\otimes_\C}
\newcommand{\tr}{\otimes_R}
\newcommand{\id}{\text{id}}
\newcommand{\KMF}{K_{0}(\operatorname{MF}_{\text{bi}}}

\newcommand{\Hom}{\operatorname{Hom}}
\newcommand{\End}{\operatorname{End}}
\newcommand{\ev}{\operatorname{ev}}
\newcommand{\tev}{\widetilde{\operatorname{ev}}}
\newcommand{\coev}{\operatorname{coev}}
\newcommand{\tcoev}{\widetilde{\operatorname{coev}}}
\newcommand{\Ga}[1]{\Gamma_{\hspace{-2pt}#1}}

\newcommand\arxiv[2]      {\href{http://arXiv.org/abs/#1}{#2}}
\newcommand\doi[2]        {\href{http://dx.doi.org/#1}{#2}}
\newcommand\httpurl[2]    {\href{http://#1}{#2}}

\allowdisplaybreaks

\deffootnote[1em]{1em}{1em}
{\textsuperscript{\thefootnotemark}}

\theoremstyle{definition}
\newtheorem{definition}{Definition}
\newtheorem{proposition}[definition]{Proposition}
\newtheorem{theorem}[definition]{Theorem}
\newtheorem{lemma}[definition]{Lemma}
\newtheorem{remark}[definition]{Remark}

\numberwithin{equation}{section}
\numberwithin{definition}{section}
\numberwithin{figure}{section}

\usetikzlibrary{arrows,calc,decorations.pathreplacing,decorations.markings,shapes.geometric,shadows}
\tikzset{
    string/.style={draw=#1, postaction={decorate}, decoration={markings,mark=at position .51 with {\arrow[draw=#1]{>}}}},
    costring/.style={draw=#1, postaction={decorate}, decoration={markings,mark=at position .51 with {\arrow[draw=#1]{<}}}},
    ostring/.style={draw=#1, postaction={decorate}, decoration={markings,mark=at position .47 with {\arrow[draw=#1]{>}}}},
    ustring/.style={draw=#1, postaction={decorate}, decoration={markings,mark=at position .56 with {\arrow[draw=#1]{>}}}},
    oostring/.style={draw=#1, postaction={decorate}, decoration={markings,mark=at position .43 with {\arrow[draw=#1]{>}}}},
    uustring/.style={draw=#1, postaction={decorate}, decoration={markings,mark=at position .59 with {\arrow[draw=#1]{>}}}},
    directed/.style={string=blue!50!black}, 
    odirected/.style={ostring=blue!50!black}, 
    udirected/.style={ustring=blue!50!black}, 
    oodirected/.style={oostring=blue!50!black}, 
    uudirected/.style={uustring=blue!50!black},     
    redirected/.style={costring= blue!50!black},
}

\usetikzlibrary{fadings,decorations.pathreplacing}

\newcommand\pgfmathsinandcos[3]{%
  \pgfmathsetmacro#1{sin(#3)}%
  \pgfmathsetmacro#2{cos(#3)}%
}
\newcommand\LongitudePlane[3][current plane]{%
  \pgfmathsinandcos\sinEl\cosEl{#2} 
  \pgfmathsinandcos\sint\cost{#3} 
  \tikzset{#1/.estyle={cm={\cost,\sint*\sinEl,0,\cosEl,(0,0)}}}
}
\newcommand\LatitudePlane[3][current plane]{%
  \pgfmathsinandcos\sinEl\cosEl{#2} 
  \pgfmathsinandcos\sint\cost{#3} 
  \pgfmathsetmacro\yshift{\cosEl*\sint}
  \tikzset{#1/.estyle={cm={\cost,0,0,\cost*\sinEl,(0,\yshift)}}} %
}
\newcommand\DrawLongitudeCircle[2][1]{
  \LongitudePlane{\angEl}{#2}
  \tikzset{current plane/.prefix style={scale=#1}}
  \pgfmathsetmacro\angVis{atan(sin(#2)*cos(\angEl)/sin(\angEl))} %
  \draw[redirected,current plane,color=blue!50!black, very thick] (\angVis:1) arc (\angVis:\angVis+180:1);
  \draw[current plane,dashed,color=blue!50!black, very thick] (\angVis-180:1) arc (\angVis-180:\angVis:1);
}
\newcommand\DrawLatitudeCircle[2][1]{
  \LatitudePlane{\angEl}{#2}
  \tikzset{current plane/.prefix style={scale=#1}}
  \pgfmathsetmacro\sinVis{sin(#2)/cos(#2)*sin(\angEl)/cos(\angEl)}
  \pgfmathsetmacro\angVis{asin(min(1,max(\sinVis,-1)))}
  \draw[directed,current plane, color=blue!50!black] (\angVis:1) arc (\angVis:-\angVis-180:1);
  \draw[current plane,dashed, color=blue!50!black] (180-\angVis:1) arc (180-\angVis:\angVis:1);
}
\newcommand\DrawLatitudeCircleU[2][1]{
  \LatitudePlane{\angEl}{#2}
  \tikzset{current plane/.prefix style={scale=#1}}
  \pgfmathsetmacro\sinVis{sin(#2)/cos(#2)*sin(\angEl)/cos(\angEl)}
  \pgfmathsetmacro\angVis{asin(min(1,max(\sinVis,-1)))}
  \draw[redirected,current plane, color=blue!50!black] (\angVis:1) arc (\angVis:-\angVis-180:1);
  \draw[current plane,dashed, color=blue!50!black] (180-\angVis:1) arc (180-\angVis:\angVis:1);
}

\newcommand\void[1]{}

\begin{document}

\title{Rigidity and defect actions in Landau-Ginzburg models}
\author{Nils Carqueville$^*$ \quad Ingo Runkel$^\dagger$
\\[0.5cm]
 \normalsize{\tt \href{mailto:nils.carqueville@physik.uni-muenchen.de}{nils.carqueville@physik.uni-muenchen.de}} \quad
  \normalsize{\tt \href{mailto:ingo.runkel@uni-hamburg.de}{ingo.runkel@uni-hamburg.de}}\\[0.1cm]
  {\normalsize\slshape $^*$Arnold Sommerfeld Center for Theoretical Physics, }\\[-0.1cm]
  {\normalsize\slshape LMU M\"unchen, Theresienstra\ss e~37, D-80333 M\"unchen}\\[-0.1cm]
  {\normalsize\slshape $^*$Excellence Cluster Universe, Boltzmannstra\ss e~2, D-85748 Garching}\\[0.1cm]
  {\normalsize\slshape $^\dagger$Department Mathematik, Universit\"{a}t Hamburg, }\\[-0.1cm]
  {\normalsize\slshape Bundesstra\ss e 55, D-20146 Hamburg}\\[-0.1cm]
}
\date{}
\maketitle

\vspace{-11.8cm}
\hfill {\scriptsize Hamburger Beitr\"age zur Mathematik 383}

\vspace{-1.0cm}

\hfill {\scriptsize ZMP-HH/10-16}

\vspace{-1.0cm}

\hfill {\scriptsize LMU-ASC 48/10}

\vspace{12cm}

\begin{abstract}
Studying two-dimensional field theories in the presence of defect lines naturally gives rise to monoidal categories: their objects are the different (topological) defect conditions, their morphisms are junction fields, and their tensor product describes the fusion of defects. These categories should be equipped with a duality operation corresponding to reversing the orientation of the defect line, providing a rigid and pivotal structure. We make this structure explicit in topological Landau-Ginzburg models with potential $x^d$, where defects are described by matrix factorisations of $x^d-y^d$. The duality allows to compute an action of defects on bulk fields, which we compare to the corresponding $\mathcal N=2$ conformal field theories. We find that the two actions differ by phases.
\end{abstract}

\newpage

\tableofcontents

\section{Introduction and summary}\label{introduction}

Defect lines are one-dimensional interfaces that separate different regions on the worldsheet in two-dimensional field theory. As such they, together with the fields that may be inserted at their junctions, are entities of the theory in their own right, and hence a complete study of field theories must also feature defects. Furthermore, defects may be used as a valuable tool to understand relations between possibly distinct theories. Both the ``internal'' and ``external'' view on defects can lead to new insights.
    
In topologically B-twisted Landau-Ginzburg models with potential~$W$, defects are described by matrix factorisations of $W\tc 1-1\tc W$~\cite{br0707.0922}. This may be understood via the folding trick~\cite{Wong:1994pa} and the fact~\cite{KontsevichU,kl0210,bhls0305,l0312} that boundary conditions in such theories are modelled by matrix factorisations of the potential.
Just as in any other topological field theory it is natural to consider the category, denoted $\operatorname{MF}(W)$ in the present case, of boundary conditions, whose morphisms describe boundary condition changing operators (with associative operator product expansion). Similarly, defects in topological Landau-Ginzburg models are the objects of a category $\MFW$ whose morphisms are topological junction fields in between possibly different defect lines. Besides being of interest on their own, defects in Landau-Ginzburg models also occur in the description of boundary conditions in the three-dimensional Rozansky-Witten model~\cite{krs0810.5415}. 

While sharing similar properties in many regards, topological boundary conditions and defects also differ in fundamental ways. One important aspect is that there is  a natural ``multiplication operation'' for topological defects, but not for boundary conditions. Indeed, by definition (see e.\,g.~\cite[sec.~3]{Runkel:2008gr}) the location of a topological defect on a two-dimensional worldsheet can be varied without affecting the value of the correlator assigned by the field theory to the worldsheet, as long as the defect line is not moved across field insertions or other defect lines. 
Hence one may consider the well-defined limit of moving two topological defects~$X$ and~$Y$ arbitrarily close to each other. This is the fused defect, denoted by $X\otimes Y$. If there are topological junction fields on the defects before fusion, then via this process they translate into one single field between the fused defects. Thus fusion is defined on the category of topological defects, and one may expect that it gives rise to a monoidal structure. That this is indeed the case for topological Landau-Ginzburg models was shown in~\cite{cr0909.4381} (building on~\cite{yoshinoTP,add0401,Khovanov:2004,br0707.0922}). 

Landau-Ginzburg models with $\mathcal N=2$ supersymmetry are closely related to superconformal field theories: physically one expects that the latter are infrared fixed points under renormalisation group flow of the former \cite{kms1989,m1989,vw1989,howewest}. The correspondence is much clearer if one restricts to the topologically twisted sector on both sides. In this case it has been successfully tested for numerous models by matching various substructures in the bulk, boundary, and defect sectors, see e.\,g.~\cite{add0401,h0401,bg0503,err0508, br0707.0922, cr0909.4381}.

If a two-dimensional conformal field theory is rational (by which we mean that the underlying vertex operator algebra satisfies the finiteness conditions of~\cite{Huang2005}), one has a concrete description of all topological defect lines which are compatible with the rational symmetry. Namely, they correspond to bimodules over a certain algebra in the category of representations of the associated vertex operator algebra. The fusion of defect lines is just given by the tensor product of these bimodules \cite{tft1,Frohlich:2006ch}.

In the present paper we shall study another property which one expects to find in the monoidal category describing defect lines, namely that of \textsl{rigidity} and that of a \textsl{pivotal structure}. Roughly, a rigid monoidal category is one with a good notion of dual objects, and a pivotal structure provides an isomorphism between an object and its double-dual which is compatible with tensor products. The basic example of a category that has these properties is that of finite-dimensional vector spaces. 

Both structures are present in the afore-mentioned defect category of rational conformal field theory~\cite{Frohlich:2006ch}. Hence one may think that the CFT/LG correspondence suggests an equivalence of such categories; however, we will find that the pivotal structure on $\MFW$ agrees with the one of the conformal field theory side only up to phases. 

Before we motivate in more detail why one should expect a rigid and pivotal structure from the physical picture, we briefly state the mathematical results proved in this paper.
\begin{itemize}
\item The category $\MFW$ of finite-rank matrix bi-factorisations in one variable is a pivotal rigid monoidal category (theorems~\ref{MFleftdual} and \ref{thm-pivotal}), and we work out this structure in explicit detail.
\item The duality operation provides an involutive ring anti-homomorphism $C$ on the Grothendieck ring $\KMF(W))$, as well as a surjective algebra homomorphism $\mathcal D_r : \KMF(W)) \otimes_{\Z} \C
\rightarrow \End^0(\End_{\MFW}(I))$ to the grade preserving linear maps on the endomorphisms of the tensor unit (lemma~\ref{C-on-K0-def} and propositions \ref{prop:map-on-K0} and \ref{Grothmap}).
\end{itemize}
Furthermore, we comment on how one might establish rigidity in the many-variable case (remark~\ref{manyduals}).

\medskip

Let us now expand on the physical motivation. 
We only consider  two-dimensional field theories defined on oriented surfaces, whose defect lines also carry an orientation. Reversing the defect orientation while retaining all other independent properties thus produces another defect~$X^\vee$, which we refer to as the dual of~$X$. A slightly different way to think about this is that one may consider ``bending'' a topological defect, e.\,g.~like this: 
\be\label{XdX}
\begin{tikzpicture}[very thick,scale=0.8,color=blue!50!black, baseline]
\draw (1,0) .. controls +(0,1) and +(0,1) .. (0,0);
\draw[->] (0,0) -- (0,-0.1);
\draw[>-] (1,-0.1) -- (1,0);
\draw (0,-0.8) -- (0,0);
\draw (1,-0.8) -- (1,0)
node[very near end,right] {{{\small$X$}}};
\end{tikzpicture}
\ee
To make a connection to what we can treat algebraically, let us reinterpret  this picture as describing a particular field inserted at the junction of the fusion of~$X$ with its dual and the invisible defect: 
\be
\begin{tikzpicture}[very thick,scale=0.8,color=blue!50!black, baseline]
\draw (1,0) .. controls +(0,1) and +(0,1) .. (0,0);
\draw[->] (0,0) -- (0,-0.1);
\draw[>-] (1,-0.1) -- (1,0);
\draw (0,-0.8) -- (0,0);
\draw (1,-0.8) -- (1,0)
node[very near end,right] {{{\small$X$}}};
\end{tikzpicture}
\equiv
\begin{tikzpicture}[very thick,scale=0.8,color=blue!50!black, baseline]
\draw (1,0) .. controls +(0,1) and +(0,1) .. (0,0);
\draw[->] (0,-0.1) -- (0,0);
\draw[->] (1,-0.1) -- (1,0);
\draw (0,-0.8) -- (0,0)
node[very near end,left] {{{\small$X^\vee$}}};
\draw (1,-0.8) -- (1,0)
node[very near end,right] {{{\small$X$}}};
\draw[dashed] (0.5,0.75) -- (0.5,1.5)
node[midway,right] {{{\small$I$}}};
\fill (0.5,0.75) circle (2pt); 
\end{tikzpicture} \, . 
\ee
The invisible defect~$I$ by definition acts as the identity under fusion, i.\,e.~there are isomorphisms $\lambda_{Y}:I\otimes Y\to Y$ and $\rho_{Y}:Y\otimes I \to Y$ for all defects~$Y$. Thus its presence can never change the value of correlators, and because of the triviality of the invisible defect it also must be dual to itself, 
\be
I^\vee \cong I \, , 
\ee
since an orientation that cannot be seen is irrelevant.

Always reading diagrams from bottom to top, we may now identify~\eqref{XdX} and its $180^\circ$-rotated version with junction fields and therefore morphisms in the defect category: 
\be\label{evcoev}
\begin{tikzpicture}[very thick,scale=0.8,color=blue!50!black, baseline]
\draw (1,0) .. controls +(0,1) and +(0,1) .. (0,0);
\draw[->] (0,-0.1) -- (0,0);
\draw[->] (1,-0.1) -- (1,0);
\draw (0,-0.8) -- (0,0)
node[very near end,left] {{{\small$X^\vee$}}};
\draw (1,-0.8) -- (1,0)
node[very near end,right] {{{\small$X$}}};
\draw[dashed] (0.5,0.75) -- (0.5,1.5)
node[midway,right] {{{\small$I$}}};
\fill (0.5,0.75) circle (2pt); 
\end{tikzpicture}
: X^\vee \otimes X \longrightarrow I \, , \quad 
\begin{tikzpicture}[very thick,scale=0.8,color=blue!50!black, baseline,rotate=180]
\draw (1,0) .. controls +(0,1) and +(0,1) .. (0,0);
\draw[<-] (0,-0.1) -- (0,0);
\draw[<-] (1,-0.1) -- (1,0);
\draw (0,-0.8) -- (0,0)
node[very near end,right] {{{\small$X^\vee$}}};
\draw (1,-0.8) -- (1,0)
node[very near end,left] {{{\small$X$}}};
\draw[dashed] (0.5,0.75) -- (0.5,1.5)
node[midway,right] {{{\small$I$}}};
\fill (0.5,0.75) circle (2pt); 
\end{tikzpicture}
: I \longrightarrow X \otimes X^\vee \, . 
\ee
These are the \textsl{evaluation} and \textsl{coevaluation} maps which are at the heart of the general duality structure of definition~\ref{defleftrigid}, and whose concrete realisation in Landau-Ginzburg models with only one chiral superfield will be given in~\eqref{evX} and~\eqref{coevXf} below.

Another intuitively natural property of topological defects is that one should be able to ``straighten them out'' as their precise location does not matter. By this we mean that locally on a worldsheet we should have the identities
\be\label{preZorro}
\begin{tikzpicture}[very thick,scale=0.8,color=blue!50!black, baseline]
\draw (1,0) .. controls +(0,1) and +(0,1) .. (0,0);
\draw (0,0) .. controls +(0,-1) and +(0,-1) .. (-1,0);
\draw (1,-0.6) -- (1,0);
\draw[->] (1,-1.3) -- (1,-0.55)
node[very near end,right] {{{\small$X$}}};
\draw (-1,0.65) -- (-1,0);
\draw[>-] (-1,0.55) -- (-1,1.3)
node[very near start,left] {{{\small$X$}}};
\end{tikzpicture}
=\,
\begin{tikzpicture}[very thick,scale=0.8,color=blue!50!black, baseline]
\draw (0,0) -- (0,1.3);
\draw[->] (0,-1.3) -- (0,0.2)
node[very near end,right] {{{\small$X$}}};
\end{tikzpicture} \, , \qquad
\begin{tikzpicture}[very thick,scale=0.8,color=blue!50!black, baseline]
\draw (1,0) .. controls +(0,-1) and +(0,-1) .. (0,0);
\draw (0,0) .. controls +(0,1) and +(0,1) .. (-1,0);
\draw[->] (1,0) -- (1,0.6);
\draw (1,1.3) -- (1,0.55)
node[very near end,right] {{{\small$X^\vee$}}};
\draw[>-] (-1,-0.65) -- (-1,0);
\draw (-1,-0.55) -- (-1,-1.3)
node[very near start,left] {{{\small$X^\vee$}}};
\end{tikzpicture}
=\,
\begin{tikzpicture}[very thick,scale=0.8,color=blue!50!black, baseline]
\draw[->] (0,-1.3) -- (0,0);
\draw (0,1.3) -- (0,0)
node[very near end,right] {{{\small$X^\vee$}}};
\end{tikzpicture}
\ee
where we have chosen not to display the invisible defect. The existence of morphisms~\eqref{evcoev} subject to the above relations is precisely what it means for the defect category to be rigid. This is the subject of theorem~\ref{MFleftdual} and remark~\ref{manyduals} for the case of $\MFW$.  

Since we think about passing to the dual defect as orientation reversal, one should expect that the map $(\,\cdot\,)^{\vee\vee}$ which sends a defect~$X$ to its double dual $X^{\vee\vee}$ is the identity. The more precise statement, which we prove as theorem~\ref{thm-pivotal}, is that there is a natural isomorphism between the identity functor and $(\,\cdot\,)^{\vee\vee}$ which is compatible with the monoidal structure. This result will be crucial for applications to concrete models. 

\medskip

Once the duality structures described so far are established, they can be used to study more concrete situations, for instance the action of defects on bulk fields. For this, consider an insertion of a bulk field~$\varphi$ somewhere on the worldsheet. Then one may ask the question of what happens to this field if one wraps a topological defect~$X$ around it and subsequently collapses~$X$ to coincide with the insertion point of~$\varphi$. This process should map~$\varphi$ to a new bulk field $\varphi_{X}$ inserted at the same point:
\be\label{defectaction}
\varphi
\longmapsto
\varphi_{X} \equiv\!
\begin{tikzpicture}[very thick,scale=0.6,color=blue!50!black, baseline]
\draw[line width=1pt] 
(0,0) node[line width=0pt] (Y) {{\small $\varphi$}};
\draw[directed] (0,-1) .. controls +(1.4,0) and +(1.4,0) .. (0,1)
node[midway,right] {{{\small$X$}}};
\draw (0,-1) .. controls +(-1.4,0) and +(-1.4,0) .. (0,1);
\end{tikzpicture} .
\ee
To formulate this in the language of rigid monoidal categories, all we have to do is to reinterpret the above picture in terms of the defect (junction) fields that we have already introduced. As a first step, we note that any bulk field may also be viewed as a defect field living on the defect~$I$ (which is invisible, after all):
\be
\varphi
\equiv\,
\begin{tikzpicture}[very thick,scale=0.8,color=blue!50!black, baseline]
\draw[line width=1pt] 
(0,0) node[inner sep=2pt,draw] (f) {{\small$\varphi$}};
\draw[dashed] (f) -- (0,1.3)
node[near end,right] {{{\small$I$}}};
\draw[dashed] (f) -- (0,-1.3)
node[near end,right] {{{\small$I$}}};
\end{tikzpicture} \, .
\ee
Consequently we may interpret the action~\eqref{defectaction} on bulk fields as a linear map $\mathcal D_r(X)$ on the endomorphisms of~$I$: 
\be\label{Dl-on-phi}
\mathcal D_r(X): \;
\begin{tikzpicture}[very thick,scale=0.8,color=blue!50!black, baseline]
\draw[line width=1pt] 
(0,0) node[inner sep=2pt,draw] (f) {{\small$\varphi$}};
\draw[dashed] (f) -- (0,1.3)
node[near end,right] {{{\small$I$}}};
\draw[dashed] (f) -- (0,-1.3)
node[near end,right] {{{\small$I$}}};
\end{tikzpicture}
\longmapsto\,
\mathcal D_{r}(X)(\varphi) =
\begin{tikzpicture}[very thick,scale=0.8,color=blue!50!black, baseline]
\draw[line width=1pt] 
(-1.5,1.3) node[inner sep=2pt,draw] (lam) {{\small$\rho_{X}\vphantom{rho_{X^\vee}}$}}
(-1.5,-1.3) node[inner sep=2pt,draw] (lamin) {{\small$\rho^{-1}_{X}$}}
(0,0) node[inner sep=2pt,draw] (f) {{\small$\varphi$}}; 
\draw[odirected] (lam) .. controls +(0,1.5) and +(0,1.5) .. (1.5,1.3);
\draw[udirected] (1.5,-1.3) .. controls +(0,-1.5) and +(0,-1.5) .. (lamin);
\draw[directed] (lamin) -- (lam)
node[midway,left] {{{\small$X$}}};
\draw (1.5,-1.3) -- (1.5,1.3);
\draw[dashed] (f) .. controls +(0,1) and +(0.5,-1) .. (lam);
\draw[dashed] (f) .. controls +(0,-1) and +(0.5,1) .. (lamin);
\draw[dashed] (0,-2.5) -- (0,-3.5)
node[near end,right] {{{\small$I$}}};
\draw[dashed] (0,2.47) -- (0,3.5)
node[near end,right] {{{\small$I$}}};
\end{tikzpicture} \; .
\ee
The right-hand side is now solely expressed in terms of the known morphisms $\varphi, \rho_{X}, \rho^{-1}_{X}$ and~\eqref{evcoev} in the defect category, and hence one can explicitly compute this map on bulk fields using the rigid monoidal structure. A special case is the action of~$X$ on the identity field, which is called the (right) \textsl{quantum dimension}
\be
\dim_r(X) = 
\begin{tikzpicture}[very thick,scale=0.6,color=blue!50!black, baseline]
\draw[directed] (0,-1) .. controls +(1.4,0) and +(1.4,0) .. (0,1)
node[midway,right] {{{\small$X$}}};
\draw (0,-1) .. controls +(-1.4,0) and +(-1.4,0) .. (0,1);
\end{tikzpicture} .
\ee
For the opposite defect orientation one obtains the (possibly different) map $\mathcal D_l$ and the left quantum dimension $\dim_l(X)$.

In section~\ref{defectbulk} we will perform this analysis of defect actions on bulk fields for a certain class of Landau-Ginzburg models
and compare the result to the analogous computation in the corresponding conformal field theories. These turn out not to agree, but they differ only by phases, and moreover these phases cancel in compositions $\mathcal{D}_l(X) \circ \mathcal{D}_r(X)$ for elementary defects $X$ (where by elementary we mean that all weight zero fields on the defect are multiples of the identity field).

\medskip

In any rational conformal field theory the defect maps~$\mathcal D_{r}$ induce bijective ring homomorphisms from the Grothendieck ring of topological defects preserving the rational symmetry to endomorphisms of the space of bulk fields that intertwine the action of the rational symmetry \cite{Fuchs:2007vk}.

On the other hand, there also exists the notion of the Grothendieck group~$\KMF(W))$ for topological defects in Landau-Ginzburg models, and we will see that it again has a ring structure via the tensor product. But since here the defect category is only triangulated and not abelian (in the non-semisimple case), the elements of the Grothendieck ring are only defined ``up to defect condensation'', see subsection~\ref{triacomp} for the precise definition. Nevertheless, despite this difference we will show in proposition~\ref{Grothmap} that when restricted to all known defects in the models that we consider, the map
\be
\KMF(W)) \otimes_\Z \C \longrightarrow \End^0(\End_{\MFW}(I))
\ee
induced by~$\mathcal D_r$ is an algebra isomorphism (we recall that the endomorphisms of the invisible defect~$I$ are precisely the bulk fields). 

In fact, the observation that the assignment of defect operators to defect conditions factors through the Grothendieck rings necessitates that the defect operators differ on the Landau-Ginzburg and conformal field theory side. As an example, a non-zero object in $\MFW$ can be zero in $\KMF(W))$, while the analogous statement is never true on the rational conformal field theory side.

\bigskip

The present paper is organised as follows. In section \ref{rigidity} we review the definition of rigid monoidal categories and pivotal structures, and show in explicit detail that matrix bi-factorisations in one variable have such structures. In section~\ref{defectbulk} these results are applied to the study of defect operators, and we compare the action of defects in topological Landau-Ginzburg models with potential $W(x) = x^d$ and A-series $\mathcal N=2$ minimal conformal field theories. Section~\ref{discuss} contains a brief discussion into the direction of duality on a higher categorial level, and some technical details are relegated to an appendix.

\section{Right and left duals for matrix bi-factorisations}\label{rigidity}

In this section we study the category of matrix bi-factorisations of one-variable potentials in detail. We explicitly show that this category is endowed with left and right dualities, and that in addition it is pivotal. The results of this section will be used in the next section where we will analyse the action of defects on bulk fields in Landau-Ginzburg models and establish that the dualities are compatible with the triangulated structure of matrix bi-factorisations. 

\subsection{Preliminaries}\label{prelim}

We will now recall the basic definition of matrix bi-factorisations and their monoidal structure. More details can be found in~\cite{cr0909.4381}. Let $R=\C[x_{1},\ldots,x_{N}]$ and $W\in R$ be a potential with an isolated singularity at the origin, i.\,e.~$\dim_{\C}(R/(\partial_{1} W,\ldots,\partial_{N}W))<\infty$. We call an $R$-bimodule free if the corresponding left $(R\tc R)$-module is free. 

A \textsl{matrix bi-factorisation (of possibly infinite rank)} of~$W$ is a tuple
\be
(X_{0},X_{1},d_{0}^X,d_{1}^X)
\ee
where~$X_{i}$ are free $R$-bimodules (of possibly infinite rank), and $d_{0}^X:X_{0}\rightarrow X_{1}$, $d_{1}^X:X_{1}\rightarrow X_{0}$ are bimodule maps such that
\be
(d_{1}^X \circ d_{0}^X)(m_{0}) = W.m_{0} - m_{0}.W \, , \quad (d_{0}^X \circ d_{1}^X)(m_{1}) = W.m_{1} - m_{1}.W
\ee
for all $m_{i}\in X_{i}$. We often represent $X$ by a matrix which we denote by the same symbol, $X \equiv (\begin{smallmatrix}0&d^X_1\\ d^X_0&0\end{smallmatrix})$.

Matrix bi-factorisations of~$W$ form the objects of a differential $\Z_{2}$-graded category $\operatorname{DG}^\infty_{\text{bi}}(W)$; its even morphisms $\phi\equiv(\begin{smallmatrix}\phi_0&0\\ 0&\phi_1\end{smallmatrix})$ from $X$ to $Y$ are pairs of bimodule maps $\phi_0:X_0\rightarrow Y_0$, $\phi_1:X_1\rightarrow Y_1$, and odd morphisms $\psi\equiv(\begin{smallmatrix}0&\psi_1\\ \psi_0&0\end{smallmatrix})$ are pairs of bimodule maps $\psi_0:X_0\rightarrow Y_1$, $\psi_1:X_1\rightarrow Y_0$. The composition in $\operatorname{DG}^\infty_{\text{bi}}(W)$ is given by matrix multiplication, and its differential~$d$ sends a homogeneous element~$\varphi\in\Hom_{\operatorname{DG}^\infty_{\text{bi}}(W)}(X,Y)$ to $d(\varphi)=Y\varphi-(-1)^{|\varphi|}\varphi X$. 

\begin{remark}\label{ony-check-f0=g0}
If $\varphi, \psi : X\rightarrow Y$ are $d$-closed even morphisms in $\operatorname{DG}^\infty_{\text{bi}}(W)$, to establish $\varphi=\psi$ it is enough to show either $\varphi_0=\psi_0$ or $\varphi_1=\psi_1$. The other equality then follows because the maps $d^X_i$, $d^Y_i$ are injective.
\end{remark}

The \textsl{category of matrix bi-factorisations (of possibly infinite rank)} of~$W$ is defined to be the homotopy category
\be
\operatorname{MF}^\infty_{\text{bi}}(W) = H^0_{d}(\operatorname{DG}^\infty_{\text{bi}}(W)) \, , 
\ee
i.\,e.~$\operatorname{MF}^\infty_{\text{bi}}(W)$ also has matrix bi-factorisations as objects, and its morphism spaces are given by the zeroth $d$-cohomology of the morphism spaces of $\operatorname{DG}^\infty_{\text{bi}}(W)$. 

Mostly we will be dealing with the full subcategory $\MFW$ of $\operatorname{MF}^\infty_{\text{bi}}(W)$ whose objects are isomorphic to matrix bi-factorisations~$X$ of finite rank. We note that instead of defining $\MFW$ as above one can of course also work exclusively with left modules and equivalently define a category $\operatorname{MF}(W\tc 1 - 1 \tc W)$. However, since our motivation is to describe topological defects, on both sides of which Landau-Ginzburg models are defined, we prefer the bimodule language of $\MFW$ to the ``folded'' boundary conditions of $\operatorname{MF}(W\tc 1 - 1 \tc W)$. 

\medskip

In order to keep the following exposition of the monoidal structure of $\MFW$ simple, let us from now on assume that $R=\C[x]$. For the general case we refer to~\cite{cr0909.4381}. To explicitly describe the monoidal structure we first have to introduce some notation to calculate with free bimodules. Every free $R$-bimodule~$M$ is isomorphic to $R\tc\check M\tc R$ for some complex vector space~$\check M$. For two vector spaces $\check M, \check N$ we consider linear maps $\phi=\sum_{m,n}\phi_{mn}a^mb^n\in\Hom_{\C}(\check M, \check N[a,b])$ where~$a$ and~$b$ are formal variables. From~$\phi$ we obtain an $R$-bimodule map~$\hat\phi$ from~$M$ to~$N$ by setting $\hat\phi(r\tc v\tc r')= \sum_{m,n} r x^m\tc \phi_{mn}(v)\tc x^n r'$. This gives us an isomorphism $\Hom_{\C}(\check M, \check N[a,b])\cong \Hom_{R\text{-mod-}R}(M,N)$. Its inverse will be denoted by~$(\check~)$, i.\,e.~for a bimodule map $\psi:M\rightarrow N$ we have $\psi=[\check\psi(a,b)]\hat~$. 

We can now recall the monoidal structure of $\MFW$ 
from~\cite{cr0909.4381} (see also \cite{yoshinoTP, add0401, kr0405232, br0707.0922}) where the general definition may be found as well. The tensor product on objects is given by
\begin{align}
X\otimes Y = & \left(\vphantom{{\textstyle \begin{pmatrix}d^X_0\otimes_R\id_{Y_0} & -\id_{X_1}\otimes_R d^Y_1\\\id_{X_0}\otimes_R d^Y_0 & d	^X_1\otimes_R \id_{Y_1}\end{pmatrix}}}X_0\otimes_R Y_0 \oplus X_1\otimes_R Y_1, X_1\otimes_R Y_0 \oplus X_0\otimes_R Y_1,\right. \nonumber \\
&\left. {\textstyle \begin{pmatrix}d^X_0\otimes_R\id_{Y_0} & -\id_{X_1}\otimes_R d^Y_1\\\id_{X_0}\otimes_R d^Y_0 & d^X_1\otimes_R \id_{Y_1}\end{pmatrix}}, {\textstyle \begin{pmatrix}d^X_1\otimes_R\id_{Y_0} & \id_{X_0}\otimes_R d^Y_1\\-\id_{X_1}\otimes_R d^Y_0 & d^X_0\otimes_R \id_{Y_1}\end{pmatrix}} \right) , \label{XtensorY}
\end{align}
and its action on morphisms is spelled out in appendix~\ref{uglymatrices}, where we also write down the explicit associator isomorphism $\alpha_{X,Y,Z}:(X\otimes Y)\otimes Z\rightarrow X\otimes(Y\otimes Z)$. The unit object $I = (R \tc R, R\tc R, d^I_0, d^I_1)$ is given by
\be
I = \begin{pmatrix}
0 & [a-b]\hat~\\
\big[ \tfrac{ W(a)-W(b) }{a-b} \big]\hat~ & 0
\end{pmatrix}
\ee
and its left and right unit isomorphisms are
\begin{subequations}\label{lambdarho}
\begin{align}
\lambda_X &= \begin{pmatrix}
\mu\otimes_R\id_{X_0}&0&0&0\\
0&0&0&\mu\otimes_R\id_{X_1}
\end{pmatrix} : I\otimes X \longrightarrow X \, , \\
\rho_X & = \begin{pmatrix}
\id_{X_0}\otimes_R\mu&0&0&0\\
0&0&\id_{X_1}\otimes_R\mu&0
\end{pmatrix} : X\otimes I \longrightarrow X
\end{align}
\end{subequations}
where $\mu:R\tc R\rightarrow R$ is the multiplication map, $\mu(r\tc r')=rr'$; their inverses in $\MFW$ are given in appendix~\ref{uglymatrices}.

Finally, we note that one easily computes $\End_{\MFW}(I)\cong R/(\partial W)$, which corresponds to the fact that defect fields living on the invisible defect are precisely bulk fields. Their identification with endomorphisms of the unit object~$I$ will be relevant when we discuss the action of defects on bulk fields in section~\ref{defectbulk}.

\subsection[Right duals in ${\rm MF_{bi}}(W)$]{Right duals in $\boldsymbol{\MFW}$}\label{leftduals}

We now begin the study of duals in $\MFW$. However, before we can identify the relevant structure, it is necessary to present some elementary constructions on the level of ordinary bimodules. 

\subsubsection{Duals of free bimodules}

Let~$R$ and~$S$ be commutative $\C$-algebras with augmentation maps $\varepsilon_{R}:R\rightarrow\C$ and $\varepsilon_{S}:S\rightarrow\C$. In our application we will have $R=S=\C[x]$ and $\varepsilon_{R}(x^k)=\delta_{k,0}$, 
but for the moment we keep our discussion more general to keep track of the left and right actions more easily. 

The \textsl{dual} of a free $R$-$S$-bimodule~$M$ is the $S$-$R$-bimodule $M^\vee$ defined as $\Hom_{R\text{-mod-}S}(M,R\tc S)$ with bimodule action $(s.\varphi.r)(m)=\varphi(r.m.s)$ for $r\in R, s\in S, m\in M, \varphi\in M^\vee$. If $M$ is not free, $M^\vee$ may well be empty, e.\,g.~for $R = \C[x]$ as an $R$-bimodule over itself one has $\Hom_{R\text{-mod-}R}(R,R\tc R)=0$.
     
Furthermore, for a map $f: M\rightarrow N$ of bimodules, we have the dual map
\be
f^\vee : N^\vee \longrightarrow M^\vee \, , \quad \psi \longmapsto \psi \circ f \, . 
\ee
In the case $R=S=\C[x]$ we can write $f=[\check f(a,b)]\hat~$ using the notation introduced in subsection~\ref{prelim}; for $\check f(a,b) = \sum_{m,n} f_{mn} a^m b^n$, this gives
\be
f^\vee = \Big[\sum_{m,n} f_{mn}^* b^m a^n\Big]^{\!\wedge} =: [\check f^*(b,a)]\hat~ \, .
\ee
The bimodule~$M$ comes together with the natural morphism
\be\label{Delta}
\delta_{M} : M \longrightarrow M^{\vee\vee} \, , \quad (\delta_{M}(m))(\varphi)=\sigma_{R,S}(\varphi(m)) \, ,
\ee
where $\sigma_{R,S} : R \tc S \rightarrow S \tc R$ is the linear map exchanging tensor factors; this is needed because elements of $M^{\vee\vee}$ are $S$-$R$-bimodule maps $M^\vee \rightarrow S \tc R$, while $\varphi(m) \in R \tc S$. Setting $\widetilde m=\delta_{M^\vee}(\varphi)\in M^{\vee\vee\vee}$, we can compute
\begin{align}
(\delta^\vee_{M}(\widetilde m))(m) & = \widetilde m(\delta_{M}(m)) = (\delta_{M^\vee}(\varphi))(\delta_{M}(m)) = \sigma_{S,R}((\delta_{M}(m))(\varphi)) 
\nonumber \\
& = \sigma_{S,R}(\sigma_{R,S}(\varphi(m))) = \varphi(m) \, ;
\end{align}
in other words, $\delta^\vee_M \circ \delta_{M^\vee} = \id_{M^\vee}$. If $M$ is finitely generated, this implies that
the map~$\delta_{M}$ enjoys the property
\be\label{dMdM}
\delta_{M}^\vee = \delta_{M^\vee}^{-1} :  M^{\vee\vee\vee} \longrightarrow M^\vee
\, . 
\ee

Any free $R$-$S$-bimodule is isomorphic to one of the form $M=R\tc \check M\tc S$ where~$\check M$ is a complex vector space, and we have a natural isomorphism $M^\vee\cong S\tc \check M^* \tc R$, see appendix~\ref{u1-action-duals-bim}; in the following we will not write out this isomorphism and identify $M^\vee \equiv S\tc \check M^* \tc R$. Then we have the bimodule map
\be
e_{M} : M^\vee \otimes_{R} M \longrightarrow S \tc S
\ee
defined via
\be
(s\tc \psi \tc r)\otimes_{R} (r'\tc m \tc s') \longmapsto \psi(m)\, \varepsilon_{R}(rr') \, s\tc s' \, .
\ee
If~$M$ is finitely generated, i.\,e.~if~$\check M$ is finite-dimensional, then we also have the bimodule map
\be
c_{M}: R\tc R \longrightarrow M\otimes_{S} M^\vee \, , \quad r\tc r' \longmapsto \sum_{i} r.e_{i} \tc 1\otimes_{S} 1\tc e_{i}^*.r' \, ,
\ee
where $\{ e_{i} \}$ is a basis of~$\check M$ and $\{ e_{i}^* \}$ is the dual basis of~$\check M^*$. The maps $e_M$ and $c_M$ will be used in the construction of the duality morphisms for $\MFW$ below.

\subsubsection{Right duals in monoidal categories}

Before we turn to duals in $\MFW$, we shall recall the notion of duality in a general monoidal category. 
\begin{definition}\label{defleftrigid}
A monoidal category $(\mathcal M, \otimes, I, \alpha, \lambda, \rho)$ is equipped with \textsl{right duality} (or is \textsl{right rigid}) if an object $A^\vee$ is assigned to each object $A\in\mathcal M$ together with morphisms $\ev_{A}: A^\vee\otimes A \rightarrow I$ and $\coev_{A}: I \rightarrow A \otimes A^\vee$ such that
\begin{subequations}\label{Zorro}
\begin{align}
& \rho_{A} \circ (\id_{A} \otimes \ev_{A})\circ \alpha_{A,A^\vee,A} \circ (\coev_{A} \otimes \id_{A}) \circ \lambda^{-1}_{A}= \id_{A} \, , \\
& \lambda_{A^\vee}\circ (\ev_{A} \otimes \id_{A^\vee})\circ \alpha^{-1}_{A^\vee,A,A^\vee} \circ (\id_{A ^\vee} \otimes \coev_{A}) \circ \rho_{A^\vee}^{-1}= \id_{A ^\vee} \, . \label{Zorrob}
\end{align}
\end{subequations}
\end{definition}
Let us introduce a convenient and standard graphical notation to express identities like the one above. Reading every diagram from bottom to top, we can picture the \textsl{evaluation} and \textsl{coevaluation} maps as follows: 
\be
\ev_{A} = 
\begin{tikzpicture}[very thick,scale=1.0,color=blue!50!black, baseline=.6cm]
\draw[line width=0pt] 
(2.5,1.6) node[line width=0pt] (I) {{\small$I$}}
(3,0) node[line width=0pt] (D) {{\small $A\vphantom{A^\vee}$}}
(2,0) node[line width=0pt] (s) {\small{$A^\vee$}}; 
\draw[directed] (D) .. controls +(0,1) and +(0,1) .. (s);
\draw[dashed] (2.5,0.81) -- (I);
\end{tikzpicture}
\equiv
\begin{tikzpicture}[very thick,scale=1.0,color=blue!50!black, baseline=.6cm]
\draw[line width=0pt] 
(3,0) node[line width=0pt] (D) {{\small$A\vphantom{A^\vee}$}}
(2,0) node[line width=0pt] (s) {{\small$A^\vee$}}; 
\draw[directed] (D) .. controls +(0,1) and +(0,1) .. (s);
\end{tikzpicture} , 
\quad
\coev_{A} = 
\begin{tikzpicture}[very thick,scale=1.0,color=blue!50!black, baseline=-.6cm,rotate=180]
\draw[line width=0pt] 
(2.5,1.6) node[line width=0pt] (I) {{\small$I$}}
(3,0) node[line width=0pt] (D) {{\small$A\vphantom{A^\vee}$}}
(2,0) node[line width=0pt] (s) {{\small$A^\vee$}}; 
\draw[redirected] (D) .. controls +(0,1) and +(0,1) .. (s);
\draw[dashed] (2.5,0.81) -- (I);
\end{tikzpicture}
\equiv
\begin{tikzpicture}[very thick,scale=1.0,color=blue!50!black, baseline=-.6cm,rotate=180]
\draw[line width=0pt] 
(3,0) node[line width=0pt] (D) {{\small$A\vphantom{A^\vee}$}}
(2,0) node[line width=0pt] (s) {{\small$A^\vee$}}; 
\draw[redirected] (D) .. controls +(0,1) and +(0,1) .. (s);
\end{tikzpicture} \, .
\ee
In this language, the conditions~\eqref{Zorro} can be rephrased as the statement that the ``Zorro moves''
\be\label{sternstern}
\begin{tikzpicture}[very thick,scale=1.0,color=blue!50!black, baseline=0cm]
\draw[line width=0] 
(-1,1.25) node[line width=0pt] (A) {{\small $A$}}
(1,-1.25) node[line width=0pt] (A2) {{\small $A$}}; 
\draw[directed] (0,0) .. controls +(0,-1) and +(0,-1) .. (-1,0);
\draw[directed] (1,0) .. controls +(0,1) and +(0,1) .. (0,0);
\draw (-1,0) -- (A); 
\draw (1,0) -- (A2); 
\end{tikzpicture}
=
\begin{tikzpicture}[very thick,scale=1.0,color=blue!50!black, baseline=0cm]
\draw[line width=0] 
(0,1.25) node[line width=0pt] (A) {{\small $A$}}
(0,-1.25) node[line width=0pt] (A2) {{\small $A$}}; 
\draw (A2) -- (A); 
\end{tikzpicture}
\, , \qquad
\begin{tikzpicture}[very thick,scale=1.0,color=blue!50!black, baseline=0cm]
\draw[line width=0] 
(1,1.25) node[line width=0pt] (A) {{\small $A^\vee$}}
(-1,-1.25) node[line width=0pt] (A2) {{\small $A^\vee$}}; 
\draw[directed] (0,0) .. controls +(0,1) and +(0,1) .. (-1,0);
\draw[directed] (1,0) .. controls +(0,-1) and +(0,-1) .. (0,0);
\draw (-1,0) -- (A2); 
\draw (1,0) -- (A); 
\end{tikzpicture}
=
\begin{tikzpicture}[very thick,scale=1.0,color=blue!50!black, baseline=0cm]
\draw[line width=0] 
(0,1.25) node[line width=0pt] (A) {{\small $A^\vee$}}
(0,-1.25) node[line width=0pt] (A2) {{\small $A^\vee$}}; 
\draw (A2) -- (A); 
\end{tikzpicture}
\ee
hold true. We note that here and below we do not explicitly depict the isomorphisms $\alpha,\lambda,\rho$ and their inverses in such diagrams. Thus identities like the ones above may be thought of as true after passing to a strict model of the monoidal category~\cite{catforwormath}, or one mentally adds the missing parts, e.\,g.
\be
\begin{tikzpicture}[very thick,scale=0.8,color=blue!50!black, baseline=0cm]
\draw[directed] (0,0) .. controls +(0,-1) and +(0,-1) .. (-1,0);
\draw[directed] (1,0) .. controls +(0,1) and +(0,1) .. (0,0);
\draw (-1,0) -- (-1,2.4); 
\draw (1,0) -- (1,-2.4); 
\end{tikzpicture}
\;\;\equiv
\begin{tikzpicture}[very thick,scale=0.8,color=blue!50!black, baseline=0cm,line/.style={&gt;=latex}]
\draw[line width=1pt] 
(-1,1.8) node[inner sep=2pt,draw] (r) {{\small$\rho_{A}$}} 
(1,-1.8) node[inner sep=2pt,draw] (l) {{\small$\lambda^{-1}_{A}$}};  
\draw[directed] (0,0) .. controls +(0,-1) and +(0,-1) .. (-1,0);
\draw[directed] (1,0) .. controls +(0,1) and +(0,1) .. (0,0);
\draw (-1,0) -- (r); 
\draw (-1,2.4) -- (r); 
\draw (1,0) -- (l); 
\draw (1,-2.4) -- (l); 
\draw[dashed] (r) .. controls +(0.25,-0.75) and +(0,0.75) .. (0.5,0.75);
\draw[dashed] (l) .. controls +(-0.25,0.75) and +(0,-0.75) .. (-0.5,-0.75);

\draw[line width=1pt] 
(0,0) node[rectangle, very thick, blue!50!black,draw,fill=white] 
(alpha) {{\small$\;\; \alpha_{A,A^\vee,A} \;\; $}} ;
\end{tikzpicture}
\, . 
\ee

As an example of a right duality one may think of the category of finite-dimensional complex vector spaces~$V$ together with the standard evaluation and coevaluation maps: 
\be
\ev_{V}: e_{i}^*\tc e_{j} \longmapsto \delta_{i,j} \, , \quad \coev_{V}: 1 \longmapsto \sum_{i} e_{i} \tc e^*_{i} \, ,
\ee
where~$\{e_{i}\}$ is an arbitrary basis of $V$. In this case one easily verifies that the Zorro moves~\eqref{sternstern} hold, which in general abstract the existence of a perfect pairing between~$V$ and~$V^*$ in the case of vector spaces. 

\begin{remark}
\label{left-rigid-all-equiv}
Let $\mathcal M$ be a right rigid monoidal category with duality given by $(X^\vee, \ev_X, \coev_X)$ for each $X \in \mathcal M$.
Suppose that $(X',\ev'_X,\coev'_X)$ is another right rigid structure on $\mathcal M$. If we replace $\coev$ by $\coev'$ in \eqref{Zorrob}, then
\be
 \phi_X = 
\lambda_{X^\vee}\circ (\ev_{X} \otimes \id_{X'})\circ \alpha^{-1}_{X^\vee,X,X'} \circ (\id_{X^\vee} \otimes \coev'_{X}) \circ \rho_{X^\vee}^{-1} : X^\vee \longrightarrow X'
\ee
gives a family of isomorphisms, natural in $X$. It follows from the Zorro moves that
\be
 \ev_X = \ev'_X \circ (\phi_X \otimes \id_X)
 \, , \quad
 \coev_X = (\id_X \otimes \phi_X^{-1}) \circ \coev'_X  \, .
\ee
In this sense, all right rigid structures on $\mathcal M$ are equivalent.
\end{remark}

\subsubsection{Right duals of matrix bi-factorisations}\label{left-dual-MFbi}

We shall now explicitly identify a right duality structure in the category of matrix bi-factorisations for the one-variable case by giving a contravariant functor $(\,\cdot\,)^\vee: \MFW\rightarrow\MFW$ and appropriate evaluation and coevaluation maps. The multi-variable case will be discussed in remark~\ref{manyduals}. 

On objects the functor $(\,\cdot\,)^\vee$ acts as
\be\label{Xdual}
X = (X_{0}, X_{1}, d^X_{0}, d^X_{1}) \longmapsto X^\vee = (X_{1}^\vee, X_{0}^\vee, (d_{0}^X)^\vee, -(d_{1}^X)^\vee) \, , 
\ee
(see~\eqref{star-dual} and~\eqref{general-dual} for the many-variable
case) and it sends a morphism $\varphi \equiv (\begin{smallmatrix}\varphi_{0}&0\\0&\varphi_{1}\end{smallmatrix}):X \rightarrow Y$ to 
\be\label{dual-morphism-phi}
  \varphi^\vee \equiv (\begin{smallmatrix}\varphi_{1}^\vee&0\\0&\varphi_{0}^\vee\end{smallmatrix}):Y^\vee\rightarrow X^\vee\, .
\ee   
We note that with this definition one has $I^\vee=I$, cf.~the discussion of section~\ref{introduction}. 

We will explicitly give the evaluation map $\ev_{X}:X^\vee \otimes X \rightarrow I$ only for objects~$X$ that have twisted differentials $\check d^X_{i}(a,b)$ with entries of polynomial degree less than $\text{deg}(W)$. This is sufficient, since any matrix bi-factorisation is isomorphic to such an object (which in turn follows as $\MFW$ has a split-generator with this property~\cite{d0904.4713} and because of~\cite[lem.~2.4]{kst0511155}), and the evaluation map can be transported using this isomorphism (see the proof of lemma~\ref{phihop} below for a similar argument in the case of the coevaluation).

For~$X$ as above the evaluation map is given by
\be\label{evX}
\ev_{X} = \begin{pmatrix}A_{X}&0&0&0\\0&0&B_{X}&C_{X}\end{pmatrix}
\ee
where we define
\begin{align}
A_{X} & = - \left[ \ev_{\check X_{1}} \circ (\id_{\check X_{1}^*} \tc \mathcal F \tc \id_{\check X_{1}}) \circ (\id_{\check X_{1}^*} \tc \id_R \tc \check d^X_{0}(x,b)) \right]^{\!\wedge} \, , \label{AX} \\
B_{X} & = \phantom{-}  \left[ \ev_{\check X_{0}} \circ \tfrac{(\id_{\check X_{1}^*}\tc \mathcal F \tc \id_{\check X_{1}})\circ \{ \id_{\check X_{1}^*}\tc \id_R \tc (\check d^X_{1}(x,a) \check d^X_{0}(x,b)) \}}{a-b}  \right]^{\!\wedge} \, , \\
C_{X} & = - e_{X_{1}} \, , \label{CX}\\
\mathcal F & = \frac{1}{2\pi\I}\oint \frac{(a-b-x)\,\D x}{x(W(x)-W(b))}  \, . \label{Fcal}
\end{align}
The formal variable $x$ in \eqref{AX} and \eqref{CX} acts by multiplication with $x$ on the middle factor in $\check X_{1}^* \tc R \tc \check X_i$. The integration contour in $\mathcal F$ is oriented counter-clockwise and taken to encircle all poles. In other words, $\mathcal F(x^k)$ computes the coefficient of $x^{-1}$ in the expansion of the formal sum \mbox{$(a-b-x)x^{-1+k} \sum_{n=0}^\infty W(b)^n/W(x)^{n+1}$}, and hence $\mathcal F$ gives a map $\C[a,x,b]\rightarrow\C[a,b]$. One may verify by direct computation (as we do in appendix~\ref{evmorph}) that $\ev_{X}$ is well-defined and indeed a morphism in $\MFW$. 

To present the coevaluation map $\coev_{X}: I \rightarrow X\otimes X^\vee$ for $X\in\MFW$, let $\vartheta : X \rightarrow Z$ be an isomorphism to a finitely generated object $Z$. Then we define $\coev_{X} = (\vartheta^{-1}\otimes\vartheta^\vee)\circ \coev_{Z}$ with
\be\label{coevXf}
\coev_{Z} = 
\begin{pmatrix}
\left[\tfrac{\check d^{Z}_{1}(a,x) \tc \id_R \tc \id_{\check Z_{1}^*} - \check d^{Z}_{1}(b,x)\tc \id_R \tc \id_{\check Z_{1}^*}}{a-b}\right]^{\!\wedge} \circ c_{Z_{1}} & 0  \\
\left[\tfrac{\check d^{Z}_{0}(a,x)\tc \id_R \tc \id_{\check Z_{0}^*} - \check d^{Z}_{0}(b,x)\tc \id_R \tc \id_{\check Z_{0}^*}}{a-b}\right]^{\!\wedge} \circ c_{Z_{0}} & 0 \\
0 & c_{Z_{1}} \\
0 & c_{Z_{0}}
\end{pmatrix} .
\ee
Again, one verifies by direct computation that this  is a morphism in $\MFW$. 

\begin{lemma}\label{phihop}
$\coev_{X}$ is independent of the choice of isomorphism~$\vartheta$. Furthermore, for any morphism $\varphi:X\rightarrow Y$ one has
\be\label{Umove}
(\varphi\otimes\id_{X^\vee})\circ \coev_{X} = (\id_{Y}\otimes \varphi^\vee) \circ \coev_{Y} \, . 
\ee
\end{lemma}
\begin{proof}
We first show that~\eqref{Umove} holds for a finitely generated~$X$. Indeed, one readily verifies that
\be
(\varphi\otimes\id_{X^\vee})\circ \coev_{X} - (\id_{Y}\otimes\varphi^\vee)\circ\coev_{Y} = \psi\circ I + (Y\otimes X^\vee)\circ\psi
\ee
for
\be\label{psi-homotopy}
\psi = 
\begin{pmatrix}
0 & \left[\tfrac{\check \varphi_{1}(a,x)\otimes_{R} \id_{X^\vee_{1}} - \check \varphi_{1}(b,x)\otimes_{R} \id_{X^\vee_{1}}}{a-b}\right]^{\!\wedge} \circ c_{X_{1}}  \\
0 & \left[\tfrac{\check \varphi_{0}(a,x)\otimes_{R} \id_{X^\vee_{0}} - \check \varphi_{0}(b,x)\otimes_{R} \id_{X^\vee_{0}}}{a-b}\right]^{\!\wedge} \circ c_{X_{0}}  \\
0 & 0 \\
0 & 0
\end{pmatrix} : I \longrightarrow Y\otimes X^\vee \, , 
\ee
and hence~\eqref{Umove} is true in $\MFW$. 

Now let $X\in\MFW$, and let $\vartheta:X\to X_{f}$ and $\vartheta':X\to X'_{f}$ be two isomorphisms to finitely generated matrix bi-factorisations. Then
\begin{align}
\coev_{X} & = (\vartheta^{-1}\otimes \vartheta^\vee)\circ \coev_{X_{f}} = (\vartheta^{-1}\otimes (\vartheta'^\vee\circ(\vartheta'^{-1})^\vee\circ \vartheta^\vee)) \circ \coev_{X_{f}} \nonumber \\
& = ((\vartheta^{-1} \circ\vartheta\circ \vartheta'^{-1})\otimes \vartheta'^\vee)\circ \coev_{X_{f}'} = (\vartheta'^{-1} \otimes \vartheta'^\vee)\circ \coev_{X_{f}'} \, ,
\end{align}
where we used~\eqref{Umove} for~$X_{f}$ and~$X_{f}'$. Thus $\coev_{X}$ is independent of the choice of isomorphism. 

Finally, we prove that~\eqref{Umove} holds for arbitrary $X,Y\in\MFW$. Let $\vartheta:X\to X_{f}, \eta:Y\to Y_{f}$ be isomorphisms to finitely generated matrix bi-factorisations and define $\Phi=\eta\circ\varphi\circ\vartheta^{-1}:X_{f}\to Y_{f}$. From this it follows that $\eta^{-1}\circ\Phi=\varphi\circ\vartheta^{-1}$ and $\vartheta^\vee \circ \Phi^\vee = \varphi^\vee \circ \eta^\vee$, so we find
\begin{align}
(\varphi\otimes\id_{X^\vee})\circ \coev_{X} & = ((\varphi\circ\vartheta^{-1})\otimes \vartheta^\vee)\circ \coev_{X_{f}} = ((\eta^{-1}\circ \Phi)\otimes \vartheta^\vee)\circ \coev_{X_{f}} \nonumber \\
& = (\eta^{-1} \otimes(\vartheta^\vee \circ \Phi^\vee)) \circ \coev_{Y_{f}} = (\eta^{-1} \otimes (\varphi^\vee \circ \eta^\vee))\circ \coev_{Y_{f}} \nonumber \\
& = (\id_{Y} \otimes \varphi^\vee) \circ \coev_{Y} \, ,
\end{align}
which concludes the proof. 
\end{proof}

Now that we have introduced all the ingredients, we can show that the functor $(\,\cdot\,)^\vee$ and the morphisms $\ev_{X}, \coev_{X}$ endow $\MFW$ with a right duality. The following result is proved in appendix~\ref{appZorro}. 
\begin{theorem}\label{MFleftdual}
For all $X\in\MFW$ we have
\begin{align}
& \rho_{X} \circ (\id_{X} \otimes \ev_{X})\circ \alpha_{X,X^\vee,X} \circ (\coev_X \otimes \id_{X}) \circ \lambda^{-1}_{X}= \id_{X} \, , \label{Z1}\\
& \lambda_{X^\vee}\circ (\ev_{X} \otimes \id_{X^\vee})\circ \alpha^{-1}_{X^\vee,X,X^\vee} 
\circ (\id_{X ^\vee} \otimes \coev_{X}) \circ \rho_{X^\vee}^{-1}= \id_{X^\vee} \, , \label{lambdaZorro}
\end{align}
i.\,e.~the Zorro moves hold true, and $\MFW$ is right rigid. 
\end{theorem}
In fact, the proof shows that the Zorro moves even hold in $\DGW$.

In pictorial language, the identity $(\varphi\otimes\id_{X^\vee})\circ \coev_{X} = (\id_{Y}\otimes \varphi^\vee) \circ \coev_{Y}$ of lemma~\ref{phihop} reads
\be\label{stern}
\begin{tikzpicture}[very thick,scale=1.0,color=blue!50!black, baseline]
\draw[line width=1pt] 
(0,1) node[line width=0pt] (Y) {{\small $Y\vphantom{X^\vee}$}}
(1,1) node[line width=0pt] (X) {{\small $X^\vee$}}
(0,0) node[inner sep=2pt,draw] (f) {{\small$\,\varphi\vphantom{\varphi^\vee}\,$}}; 
\draw[uudirected] (1,0) .. controls +(0,-1) and +(0,-1) .. (f);
\draw (X) -- (1,0);
\draw (Y) -- (f);
\end{tikzpicture}
=\;
\begin{tikzpicture}[very thick,scale=1.0,color=blue!50!black, baseline]
\draw[line width=1pt] 
(0,1) node[line width=0pt] (Y) {{\small $Y\vphantom{X^\vee}$}}
(1,1) node[line width=0pt] (X) {{\small $X^\vee$}}
(1,0) node[inner sep=2pt,draw] (f) {{\small$\varphi^\vee$}}; 
\draw[oodirected] (f) .. controls +(0,-1) and +(0,-1) .. (0,0);
\draw (X) -- (f);
\draw (Y) -- (0,0);
\end{tikzpicture} \, .
\ee
Using both Zorro moves we can readily derive the analogous expression for the evaluation map: by appending curved lines to the right and left of equation~\eqref{stern} it follows that
\be \label{dual-morph-and-ev}
\begin{tikzpicture}[very thick,scale=1.0,color=blue!50!black, baseline]
\draw[line width=1pt] 
(0,0) node[inner sep=2pt,draw] (f) {{\small$\,\varphi\vphantom{\varphi^\vee}\,$}}; 
\draw (f) -- (0,1);
\draw (0,-1) -- (f);
\end{tikzpicture}
=
\begin{tikzpicture}[very thick,scale=1.0,color=blue!50!black, baseline]
\draw[line width=1pt] 
(0,0) node[inner sep=2pt,draw] (f) {{\small$\,\varphi\vphantom{\varphi^\vee}\,$}}; 
\draw[uudirected] (0.75,0) .. controls +(0,-1) and +(0,-1) .. (f);
\draw[directed] (1.5,0) .. controls +(0,1) and +(0,1) .. (0.75,0);
\draw (0,1) -- (f);
\draw (1.5,0) -- (1.5,-1);
\end{tikzpicture}
=
\begin{tikzpicture}[very thick,scale=1.0,color=blue!50!black, baseline]
\draw[line width=1pt] 
(1,0) node[inner sep=2pt,draw] (f) {{\small$\varphi^\vee$}}; 
\draw[oodirected] (f) .. controls +(0,-1) and +(0,-1) .. (0.25,0);
\draw[uudirected] (1.75,0) .. controls +(0,1) and +(0,1) .. (f);
\draw (0.25,1) -- (0.25,0);
\draw (1.75,-1) -- (1.75,0);
\end{tikzpicture}
\quad\Rightarrow\quad
\begin{tikzpicture}[very thick,scale=1.0,color=blue!50!black, baseline]
\draw[line width=1pt] 
(0,0) node[inner sep=2pt,draw] (f) {{\small$\,\varphi\vphantom{\varphi^\vee}\,$}}; 
\draw[oodirected] (f) .. controls +(0,1) and +(0,1) .. (-0.75,0);
\draw (-0.75,0) -- (-0.75,-1);
\draw (f) -- (0,-1);
\end{tikzpicture}
=
\begin{tikzpicture}[very thick,scale=1.0,color=blue!50!black, baseline]
\draw[line width=1pt] 
(1,0) node[inner sep=2pt,draw] (f) {{\small$\varphi^\vee$}}; 
\draw[oodirected] (f) .. controls +(0,-1) and +(0,-1) .. (0.25,0);
\draw[uudirected] (1.75,0) .. controls +(0,1) and +(0,1) .. (f);
\draw[directed] (0.25,0) .. controls +(0,1) and +(0,1) .. (-0.5,0);
\draw (1.75,-1) -- (1.75,0);
\draw (-0.5,0) -- (-0.5,-1);
\end{tikzpicture}
=
\begin{tikzpicture}[very thick,scale=1.0,color=blue!50!black, baseline]
\draw[line width=1pt] 
(1,0) node[inner sep=2pt,draw] (f) {{\small$\varphi^\vee$}}; 
\draw[uudirected] (1.75,0) .. controls +(0,1) and +(0,1) .. (f);
\draw (1.75,-1) -- (1.75,0);
\draw (f) -- (1,-1);
\end{tikzpicture} \, .
\ee
Thus we have found: 
\begin{lemma}\label{ev-and-f-dual}
For any morphism $\varphi:X\rightarrow Y$ in $\MFW$ one has
\be
\ev_{Y}\circ (\id_{Y^\vee}\otimes \varphi)  = \ev_{X}\circ (\varphi^\vee \otimes\id_{X})\, . 
\ee
\end{lemma}

Another simple application of~\eqref{sternstern} and the above lemma is to show that our definition of~$\varphi^\vee$ in~\eqref{dual-morphism-phi} agrees with the canonical definition of a dual morphism in a rigid category, 
\begin{align}\label{dualphi}
\begin{tikzpicture}[very thick,scale=1.0,color=blue!50!black, baseline]
\draw[line width=1pt] 
(0,0) node[inner sep=2pt,draw] (f) {{\small$\varphi^\vee$}}; 
\draw (f) -- (0,1);
\draw (0,-1) -- (f);
\end{tikzpicture}
=
\begin{tikzpicture}[very thick,scale=1.0,color=blue!50!black, baseline]
\draw[line width=1pt] 
(0,0) node[inner sep=2pt,draw] (f) {{\small$\varphi^\vee$}}; 
\draw[uudirected] (1,0) .. controls +(0,1) and +(0,1) .. (f);
\draw[directed] (2,0) .. controls +(0,-1) and +(0,-1) .. (1,0);
\draw (0,-1) -- (f);
\draw (2,0) -- (2,1);
\end{tikzpicture}
=
\begin{tikzpicture}[very thick,scale=1.0,color=blue!50!black, baseline]
\draw[line width=1pt] 
(1,0) node[inner sep=2pt,draw] (f) {{\small$\,\varphi\vphantom{\varphi^\vee}\,$}}; 
\draw[oodirected] (f) .. controls +(0,1) and +(0,1) .. (0,0);
\draw[uudirected] (2,0) .. controls +(0,-1) and +(0,-1) .. (f);
\draw (0,-1) -- (0,0);
\draw (2,1) -- (2,0);
\end{tikzpicture} \, .
\end{align}
In diagram-free language, this reads
\be \label{phi-dual-generadef}
\varphi^\vee = 
\lambda_{X^\vee}\circ (\ev_{Y}\otimes \id_{X^\vee}) \circ \alpha^{-1}_{Y^\vee,X,X^\vee} \circ (\id_{Y^\vee}\otimes(\varphi\otimes \id_{X^\vee})) \circ (\id_{Y^\vee}\otimes \coev_{X}) \circ \rho^{-1}_{Y^\vee} \, .
\ee
We note that if the dual of a morphism is defined as above, then the identities~\eqref{stern} and~\eqref{dual-morph-and-ev} immediately follow by applying Zorro moves.

\begin{lemma}\label{evlr}
We have $\ev_{I}=\lambda_{I}=\rho_{I}$ and $\coev_{I}=\lambda^{-1}_{I}=\rho^{-1}_{I}$ in $\MFW$. 
\end{lemma}
\begin{proof}
By direct computation one finds $\lambda_{I}\circ \coev_{I} = \id_{I}$ in $\DGW$, and therefore $\coev_{I}=\lambda^{-1}_{I}=\rho^{-1}_{I}$ in $\MFW$. 

The $(1,1)$-entry of the $(2\times 2)$-matrix $\ev_{I}\circ \lambda^{-1}_{I}$ is given by $A_{I}\circ [1\tc \id_{\check I_{0}}]\hat~¬†$ which is equal to
\be
- \left[ \frac{1}{2\pi\I} \oint \frac{(a-b-x)\check d^I_{0}(x,b) \D x}{x(W(x)-W(b))} \right]^{\!\wedge} = - \left[ \frac{1}{2\pi\I} \oint \frac{(a-b-x) \D x}{x(x-b)} \right]^{\!\wedge} = \id_{I_{0}} \, .
\ee
By remark~\ref{ony-check-f0=g0}, this determines the $(2,2)$-entry to be $\id_{I_{1}}$, and we have $\ev_{I}\circ \lambda^{-1}_{I} = \id_{I}$ in $\DGW$ and thus $\ev_{I}=\lambda_{I}=\rho_{I}$ in $\MFW$. 
\end{proof}

\begin{remark}\label{manyduals}
Let us explain the relation between the duality structure discussed here and the one relevant for the category $\operatorname{MF}(W)$ of matrix factorisations (describing boundary conditions, not defects). This will allow us to argue that $\MFW$ is expected to be right rigid also in the multi-variable case. 
\begin{enumerate}
\item 
For an object $Q=(\begin{smallmatrix}0&q_{1}\\q_{0}&0\end{smallmatrix})\in\operatorname{MF}(W)$, its dual is given by $Q^*=(\begin{smallmatrix}0&-q^*_{0}\\q^*_{1}&0\end{smallmatrix})\in\operatorname{MF}(-W)$. This is the natural choice in the sense that we have isomorphisms of complexes
\be\label{HomTensor}
\Hom_{\operatorname{DG}(W)}(Q,P) \cong (P_{0} \oplus P_{1}) \otimes_{R} (Q_{0}^* \oplus Q_{1}^*)
\ee
where the differential on the right-hand side is the matrix factorisation $P \otimes_{R} Q^*$ of 
zero~\cite{kr0405232}. 
Furthermore, there are isomorphisms
\begin{align}
\Hom_{\MFW}(P^*\tc Q,I) & \cong \Hom_{\operatorname{MF}(W)}(Q,P) \, , \nonumber \\
 \Hom_{\MFW}(I,Q\tc P^*) & \cong \Hom_{\operatorname{MF}(W)}(P,Q) \, ,
\end{align}
see e.\,g.~\cite{brr0909.0696, dm1004.0687}. However, for a potential~$W$ in~$N$ variables we define duals as follows in $\MFW$: we set
\be\label{star-dual}
X^\star = 
\begin{pmatrix}
0 & -[(\check d^X_{0})^*(b_{1},\ldots,b_{N},a_{1},\ldots,a_{N})]\hat~ \\
[(\check d^X_{1})^*(b_{1},\ldots,b_{N},a_{1},\ldots,a_{N})]\hat~ & 0
\end{pmatrix} 
\ee
and
\be\label{general-dual}
X^\vee = T^{N}X^\star \, , \quad
\begin{pmatrix}
\varphi_{0} & 0 \\
0 & \varphi_{1}
\end{pmatrix}^\vee
=
T^N\begin{pmatrix}
\varphi_{0}^\vee & 0 \\
0 & \varphi_{1}^\vee
\end{pmatrix}
\ee
for objects~$X$ and morphisms~$\varphi$ in $\MFW$, where~$T$ is the shift functor (cf.~section~\ref{triacomp}). 
We note that the definition of $X^\vee$ coincides with~\eqref{Xdual} in the one-variable case. The crucial fact, proved e.\,g.~by generalising the method of~\cite{err0508} or the homological perturbation lemma analysis of~\cite{dm1004.0687},\footnote{We thank Daniel Murfet for a helpful discussion on this point.} is that only with this definition do we have
the natural (in $X$ and $Y$) isomorphisms
\begin{subequations}\label{isos}
\begin{align}
\Hom_{\MFW}(Y^\vee \otimes X, I) & \cong \Hom_{\MFW}(X,Y) \, , \\
\Hom_{\MFW}(I,X\otimes Y^\vee) & \cong \Hom_{\MFW}(Y,X)
\end{align}
\end{subequations}
in $\MFW$. Using $X^\star$ instead of $X^\vee$ gives rise to quasi-isomorphisms of non-zero degree in $\DGW$ if~$N$ is odd, which hence do not induce isomorphisms in $\MFW$, and the physical condition $I^\vee\cong I$ is only satisfied for the correct dual $(\,\cdot\,)^\vee$. 

\item
It is expected that one can use the isomorphisms~\eqref{isos} to prove that $\MFW$ is right rigid also in the general multi-variable case. Indeed, natural candidates for the evaluation and coevaluation maps can be constructed as the preimages of the identity (for $X=Y$) under the isomorphisms~\eqref{isos}. Naturality in $X$ and $Y$ of the maps~\eqref{isos} then implies that for any morphism $\varphi : X\rightarrow Y$ we have $\ev_Y \circ (\id_{Y^\vee} \otimes \varphi) = \ev_X \circ (\varphi^\vee \otimes \id_X)$, i.\,e.~the statement of lemma~\ref{ev-and-f-dual} holds.

Nonetheless, it would have to be checked separately if the Zorro moves are satisfied, and to do this explicitly is 
(in principle straightforward yet) 
rather involved for general~$W$. In this paper we are concerned with the one-variable case and we leave the multi-variable expressions for the evaluation and coevaluation maps to future work.

\item
\label{I-dual-several-var}
In the one-variable case we saw that the unit object is \textsl{equal} to its dual. In the many-variable case this is no longer true for our choice~\eqref{general-dual} of duals. However, it is straightforward to construct an isomorphism $\gamma:I^\vee\rightarrow I$; one finds that~$\gamma$ is given by a symmetric permutation matrix (with some negative entries). 

From part (ii) above we expect that there exists a right rigid structure $(X^\vee,\ev_{X},\coev_{X})$ on $\MFW$; let us assume that this is the case. By remark~\ref{left-rigid-all-equiv} all such structures are equivalent, and hence one could choose another right rigid structure on $\MFW$ which coincides with $(X^\vee,\ev_{X},\coev_{X})$ for all $X\neq I$, but for which the duality maps of~$I$ are defined by
\be\label{evI-coevI-are-lambdaI}
 \ev_I = \lambda_I \circ (\gamma \otimes \id_I)
 \, , \quad
 \coev_I = (\id_I \otimes \gamma^{-1}) \circ \lambda_I^{-1}  \, .
\ee
One easily verfies that $\ev_{I}$ and $\coev_{I}$ as above satisfy the Zorro moves. We note that the statement of lemma~\ref{evlr} can be rephrased as saying that with $\ev_X$ and $\coev_X$ as given in \eqref{evX} and \eqref{coevXf}, equation~\eqref{evI-coevI-are-lambdaI} holds with $\gamma = \id_I$.
\end{enumerate}
\end{remark}

\subsubsection{R-charge}\label{R-charge-short}

Instead of $\MFW$ one may also consider the \textsl{category of graded matrix bi-factorisations} $\MFR(W)$, see e.\,g.~\cite{hw0404196} and appendix~\ref{R-charge-app}. Its objects are matrix bi-factorisations~$X$ together with invertible even bimodule maps $U^X(\alpha): X_{0}\oplus X_{1}\to X_{0}\oplus X_{1}$ for all $\alpha\in\C$ subject to a group law (see appendix~\ref{R-charge-app}) and such that
\be
U^{X}(\alpha) \circ [\check X(\E^{\I q_{x}\alpha}a,\E^{\I q_{x}\alpha}b)]\hat~ \circ U^{X}(\alpha)^{-1} = \E^{\I\alpha} X
\ee
for all $\alpha\in\C$. Here we take $W$ to be homogeneous of polynomial degree $d$ and $q_{x}=2/d$ is the charge assigned to~$x$ in~$R$. A morphism~$\varphi\in\Hom_{\MFR(W)}(X,Y)$ is the same as a morphism in $\MFW$; it has \textsl{R-charge}~$p$ if
\be
U^{Y}(\alpha)\circ [\varphi(\E^{\I q_{x}\alpha}a, \E^{\I q_{x}\alpha}b)]\hat~ \circ U^X(\alpha)^{-1} = \E^{\I p \alpha} \varphi \, . 
\ee
It is shown in~\cite[sec.~2.3]{cr0909.4381} that with 
$U^I(\alpha) =
(\begin{smallmatrix}
1 & 0 \\
0 & \E^{\I\alpha(q_x-1)}
\end{smallmatrix})$, the isomorphisms $\alpha_{X}, \lambda_{X}, \rho_{X}$ and their inverses have R-charge zero. 

The dual of a graded matrix bi-factorisation $(X,U^X(\alpha))$ is
\be
(X^\vee, \E^{\I \alpha (q_x-1)} (U^X(\alpha)^{-1})^\vee) \, .
\ee
With this definition we have $I=I^\vee$ also as graded matrix bi-factorisation, and one can check that both $\ev_{X}$ and $\coev_{X}$ have R-charge zero. More details can be found in appendix~\ref{R-charge-app}.

\subsection{Pivotal structure}\label{pivotal}

The notion of a pivotal structure\footnote{For a more detailed discussion of pivotal structures one may e.\,g.~refer to \cite{freyd-yetter} (in the strict case), \cite[sec.~3.1]{maltsiniotis} (where the name ``sovereign'' is used), or \cite{mu0804.3587}.} will be needed when we derive the properties of defect operators in the next section. We will first state the general definition and then show that $\MFW$ has a natural pivotal structure (and also left duals).

\subsubsection{Definition and properties of pivotal structures}
  
Let $\mathcal M$ be a right rigid monoidal category as in definition \ref{defleftrigid}. We obtain a contravariant functor $(\,\cdot\,)^\vee : \mathcal M \rightarrow \mathcal M$ which acts as $X \mapsto X^\vee$ on objects and as \eqref{phi-dual-generadef} on morphisms. This functor can be equipped with a natural \textsl{monoidal structure} $\big( (\,\cdot\,)^\vee, \nu^2, \nu^0 \big)$, where $\nu^0 : I \rightarrow I^\vee$ is an isomorphism and $\nu^2$ is a natural family of isomorphisms
\be\label{nu2-iso}
\nu^2_{X,Y} : X^\vee \otimes Y^\vee \longrightarrow (Y\otimes X)^\vee \, .
\ee
Both~$\nu^0$ and~$\nu^2$ are given in terms of the right rigid structure, namely, $\nu^0 = \lambda_{I^\vee} \circ \coev_I$ and
\begin{align}\label{gamma}
\nu^2_{X,Y} & = 
\begin{tikzpicture}[very thick,scale=0.8,color=blue!50!black, baseline]
\draw[line width=1pt] 
(-2,-1) node[line width=0pt] (X) {{\small $X^\vee$}}
(-1,-1) node[line width=0pt] (Y) {{\small $Y^\vee$}}
(2,2) node[line width=0pt] (XY) {{\small $(Y\otimes X)^\vee$}}; 
\draw[directed] (0,0) .. controls +(0,1) and +(0,1) .. (-1,0);
\draw[directed] (1,0) .. controls +(0,2) and +(0,2) .. (-2,0);
\draw[directed] (2,0) .. controls +(0,-1) and +(0,-1) .. (0.5,0);
\draw (-1,0) -- (Y);
\draw (-2,0) -- (X);
\draw[dotted] (0,0) -- (1,0);
\draw (2,0) -- (XY);
\end{tikzpicture}
 \nonumber \\
& = \lambda_{(Y\otimes X)^\vee} \circ (\ev_{X} \otimes \id_{(Y\otimes X)^\vee}) \circ \alpha^{-1}_{X^\vee,X,(Y\otimes X)^\vee} \circ  (\id_{X^\vee} \otimes (\lambda_{X}\otimes \id_{(Y\otimes X)^\vee})) \nonumber \\
& \qquad \circ  ( \id_{X^\vee} \otimes ((\ev_{Y} \otimes \id_{X})\otimes \id_{(Y\otimes X)^\vee})) \nonumber \\
& \qquad \circ (\id_{X^\vee} \otimes ( (\alpha^{-1}_{Y^\vee,Y,X}\otimes \id_{(Y\otimes X)^\vee}) \circ \alpha^{-1}_{Y^\vee,Y\otimes X, (Y\otimes X)^\vee} )) \nonumber \\
& \qquad \circ (\id_{X^\vee} \otimes (\id_{Y^\vee} \otimes \coev_{Y\otimes X})) \circ (\id_{X^\vee} \otimes \rho^{-1}_{Y^\vee}) \, .
\end{align}
The isomorphisms $\nu^2_{X,Y}$ and $\nu^0$ have to satisfy the coherence conditions of a monoidal functor: using repeated Zorro moves and \eqref{dualphi} one verifies that the three diagrams
\be\label{cross}
\xymatrix{%
(X^\vee\otimes Y^\vee)\otimes Z^\vee \ar[d]_{\alpha_{X^\vee,Y^\vee,Z^\vee}} \ar[rr]^-{\nu^2_{X,Y}\otimes\id_{Z^\vee}} && (Y\otimes X)^\vee\otimes Z^\vee \ar[rr]^-{\nu^2_{Y\otimes X,Y}} && (Z\otimes (Y\otimes X))^\vee\ar[d]^{(\alpha^{-1}_{Z,Y,X})^\vee} \\
X^\vee\otimes (Y^\vee\otimes Z^\vee) \ar[rr]^-{\id_{X^\vee}\otimes \nu^2_{Y,Z}} && X^\vee\otimes (Z\otimes Y)^\vee \ar[rr]^-{\nu^2_{X,Z\otimes Y}} && ((Z\otimes Y)\otimes X)^\vee \\
}%
\ee
and
\begin{equation}
\parbox[c]{\textwidth}{
\xymatrix{%
I \otimes X^\vee \ar[r]^-{\lambda_{X^\vee}} \ar[d]_{\nu^0 \otimes \id_{X^\vee}} & X^\vee \ar[d]^{(\rho_X)^\vee} \\
I^\vee \otimes X^\vee \ar[r]^-{\nu^2_{I,X}} & (X\otimes I)^\vee
}}
\quad , \quad
\parbox[c]{\textwidth}{
\xymatrix{%
X^\vee \otimes I \ar[r]^-{\rho_{X^\vee}} \ar[d]_{\id_{X^\vee} \otimes \nu^0} & X^\vee \ar[d]^{(\lambda_X)^\vee} \\
X^\vee \otimes I^\vee \ar[r]^-{\nu^2_{X,I}} & (I\otimes X)^\vee
}}
\end{equation}
commute. 

We will need the covariant monoidal functor 
$\big( (\,\cdot\,)^{\vee\vee}, \omega^2, \omega^0 \big)$ whose isomorphism data are given by
\begin{align}\label{vv-monoidal}
\omega^0 & = \big( (\nu^0)^{-1} \big)^\vee \circ \nu^0 : I \longrightarrow I^{\vee\vee} \, , \nonumber \\
\omega^2_{X,Y} & = \big( (\nu^2_{Y,X})^{-1} \big)^\vee \circ \nu^2_{X^\vee,Y^\vee} : X^{\vee\vee} \otimes Y^{\vee\vee} \longrightarrow (X \otimes Y)^{\vee\vee} \, .
\end{align}
It follows from a straightforward calculation using the Zorro moves and from definition \eqref{phi-dual-generadef} of the
action of $(\,\cdot\,)^\vee$ on morphisms that these morphisms satisfy the following equalities:
\begin{align} \label{omega-IX-XI}
  \omega_{I,Y}^2 &= (\lambda_Y^{-1})^{\vee\vee} \circ \lambda_{Y^{\vee\vee}} 
  \circ \big( (\omega^0)^{-1} \otimes \id_{Y^{\vee\vee}} \big) \, ,
\nonumber\\
  \omega_{X,I}^2 &= (\rho_X^{-1})^{\vee\vee} \circ \rho_{X^{\vee\vee}} 
  \circ \big(   \id_{X^{\vee\vee}} \otimes (\omega^0)^{-1} \big) \, .
\end{align}

Let $\mathcal M$ and $\mathcal M'$ be monoidal categories, and let $F \equiv (F,F^2,F^0)$ and $G \equiv (G,G^2,G^0)$ be covariant monoidal functors $\mathcal M \rightarrow \mathcal M'$. We recall that a \textsl{monoidal natural transformation} is a natural transformation $\eta : F \Rightarrow G$ such that
\be\label{mnt}
\parbox[c]{\textwidth}{
\xymatrix{%
F(X)\otimes' F(Y) \ar[rr]^-{F^2_{X,Y}} \ar[d]_{\eta_{X}\otimes' \eta_{Y}} && F(X\otimes Y) \ar[d]^{\eta_{X\otimes Y}} \\
G(X)\otimes' G(Y) \ar[rr]^-{G^2_{X,Y}} && G(X\otimes Y)
}}
\quad\text{and}\quad
\parbox[c]{\textwidth}{
\xymatrix{%
I \ar[d]_{F^0} \ar[dr]^{G^0} & \\
F(I) \ar[r]^{\eta_{I}} & G(I)
}}
\ee
commute. 

\begin{definition}
A \textsl{pivotal structure} on a right rigid monoidal category is a monoidal natural isomorphism $\operatorname{Id}\Rightarrow (\,\cdot\,)^{\vee\vee}$. We will call a right rigid monoidal category with pivotal structure a \textsl{pivotal category}.
\end{definition}

Let $t_X : X \rightarrow X^{\vee\vee}$ be such a pivotal structure. It is proved e.\,g.~in \cite[prop.~A.1 (journal version)]{s0309131} that $t_X$ automatically satisfies the identity
\be \label{tXvtinv}
   t_{X^\vee}^{-1} = (t_X)^\vee : X^{\vee\vee\vee} \longrightarrow X^\vee \, .
\ee
Also note that any two pivotal structures $t_X$ and $s_X$ on a given right rigid monoidal category $\mathcal M$ differ by a monoidal natural isomorphism of the identity functor, namely $s_X^{-1} \circ t_X$. In other words, if $\mathcal M$ allows for a pivotal structure, the set of all pivotal structures on $\mathcal M$ forms a torsor over the group of monoidal natural isomorphisms of the identity functor.

We will later need to compare different pivotal categories. To prepare the definition, suppose we are given two right rigid monoidal categories $\mathcal C$ and $\mathcal D$, and a monoidal functor $(F,F^2,F^0) : \mathcal C \to \mathcal D$. By remark \ref{left-rigid-all-equiv}, we can construct a natural isomorphism $\psi : F\circ(\,\cdot\,)^\vee \Rightarrow (\,\cdot\,)^\vee\circ F$. Namely,
\be
\psi_X = 
\lambda_{F(X^\vee)}\circ (\ev^F_{X} \otimes \id_{F(X)^\vee}) \circ \alpha^{-1}_{F(X^\vee),F(X),F(X)^\vee} \circ (\id_{F(X)^\vee} \otimes \coev^{\mathcal D}_{F(X)}) \circ \rho_{F(X^\vee)}^{-1}
\ee
where $\ev^F_{X}$ is given by
\be
 \ev^F_{X} = \Big( \!
 \xymatrix{%
 F(X^\vee) \otimes F(X)
 \ar[rr]^-{F^2_{X^\vee,X}} && 
 F(X^\vee \otimes X) 
 \ar[rr]^-{F(\ev^{\mathcal{C}}_X)} && 
 F(I)
 \ar[rr]^-{(F^0)^{-1}} &&
 I \Big)
 } .
\ee

\begin{definition}\label{pivotal-functor-def}
Let $\mathcal C$ and $\mathcal D$ be pivotal categories with pivotal structures $t^{\mathcal C}$ and $t^{\mathcal D}$. A monoidal functor $(F,F^2,F^0) : \mathcal C \to \mathcal D$ is called \textsl{pivotal} if
\be
 t^{\mathcal{D}}_{F(X)} = (\psi_{X}^{-1})^\vee \circ \psi_{X^\vee} \circ F(t^{\mathcal C}_X)  
  :  F(X) \longrightarrow F(X)^{\vee\vee}
\ee
for all $X \in \mathcal C$. Two pivotal categories are \textsl{pivotally equivalent} if there are pivotal monoidal functors $F : \mathcal C \to \mathcal D$ and $G: \mathcal D \to \mathcal C$, such that $F \circ G$ and $G \circ F$ are naturally isomorphic to the identity functor.
\end{definition}

\begin{remark}\label{strict-pivotal}
\begin{enumerate}
\item Let $F : \mathcal C \to \mathcal D$ be a pivotal monoidal functor between two pivotal categories which is essentially surjective. This implies that $F$ is an equivalence of monoidal categories. However, it does \textsl{not} imply that $\mathcal C$ and $\mathcal D$ are pivotally equivalent, because there may not exist an inverse functor $G : \mathcal D \to \mathcal C$ which is also pivotal.
\item We say a pivotal category is \textsl{strictly pivotal} if $(\,\cdot\,)^{\vee\vee} = \mathrm{Id}$ and $t_X = \id_X$ for all $X$. Given a pivotal category $\mathcal C$, in~\cite{JoyalStreet2} a strictly pivotal category $\mathcal C_\text{str}$ is constructed such that there is a pivotal functor $F : \mathcal C \to \mathcal C_\text{str}$, which is essentially surjective. However, $\mathcal C$ is in general not pivotally equivalent to $\mathcal C_\text{str}$.

\end{enumerate}
\end{remark}

\subsubsection{Left duals}\label{rightduals}

Let $\mathcal M$ be a right rigid monoidal category and for $X \in \mathcal M$ let $t_X : X \rightarrow X^{\vee\vee}$ be a collection of isomorphisms (which need not be natural). Using the maps $t_X$ we can 
define \textsl{left duals}, i.\,e.~``tilded'' evaluation and coevaluation maps: 
\begin{align}
\tev_{X} & = \ev_{X^\vee} \circ (t_{X}\otimes \id_{X^\vee}): X \otimes X^\vee \longrightarrow I \, , \nonumber \\
\tcoev_{X} & = ( \id_{X^\vee}\otimes t_{X}^{-1} ) \circ \coev_{X^\vee} : I \longrightarrow X^\vee \otimes X \, . 
\label{tev-tcoev-def}
\end{align}
Pictorially we present these as
\be
\tev_{X} = 
\begin{tikzpicture}[very thick,scale=1.0,color=blue!50!black, baseline=.6cm]
\draw[line width=0pt] 
(2.5,1.6) node[line width=0pt] (I) {{\small$I$}}
(3,0) node[line width=0pt] (D) {{\small $X^\vee$}}
(2,0) node[line width=0pt] (s) {\small{$X\vphantom{A^\vee}$}}; 
\draw[redirected] (D) .. controls +(0,1) and +(0,1) .. (s);
\draw[dashed] (2.5,0.81) -- (I);
\end{tikzpicture}
\equiv
\begin{tikzpicture}[very thick,scale=1.0,color=blue!50!black, baseline=.6cm]
\draw[line width=0pt] 
(3,0) node[line width=0pt] (D) {{\small $X^\vee$}}
(2,0) node[line width=0pt] (s) {\small{$X\vphantom{A^\vee}$}}; 
\draw[redirected] (D) .. controls +(0,1) and +(0,1) .. (s);
\end{tikzpicture} , 
\quad
\tcoev_{X} = 
\begin{tikzpicture}[very thick,scale=1.0,color=blue!50!black, baseline=-.6cm,rotate=180]
\draw[line width=0pt] 
(2.5,1.6) node[line width=0pt] (I) {{\small$I$}}
(3,0) node[line width=0pt] (D) {{\small$X^\vee$}}
(2,0) node[line width=0pt] (s) {{\small$X\vphantom{A^\vee}$}}; 
\draw[directed] (D) .. controls +(0,1) and +(0,1) .. (s);
\draw[dashed] (2.5,0.81) -- (I);
\end{tikzpicture}
\equiv
\begin{tikzpicture}[very thick,scale=1.0,color=blue!50!black, baseline=-.6cm,rotate=180]
\draw[line width=0pt] 
(3,0) node[line width=0pt] (D) {{\small$X^\vee$}}
(2,0) node[line width=0pt] (s) {{\small$X\vphantom{A^\vee}$}}; 
\draw[directed] (D) .. controls +(0,1) and +(0,1) .. (s);
\end{tikzpicture}
\ee
and by construction $\tev_{X}, \tcoev_{X}$ satisfy the Zorro moves
\be
\begin{tikzpicture}[very thick,scale=1.0,color=blue!50!black, baseline=0cm]
\draw[line width=0] 
(1,1.25) node[line width=0pt] (A) {{\small $X$}}
(-1,-1.25) node[line width=0pt] (A2) {{\small $X$}}; 
\draw[redirected] (0,0) .. controls +(0,1) and +(0,1) .. (-1,0);
\draw[redirected] (1,0) .. controls +(0,-1) and +(0,-1) .. (0,0);
\draw (-1,0) -- (A2); 
\draw (1,0) -- (A); 
\end{tikzpicture}
=
\begin{tikzpicture}[very thick,scale=1.0,color=blue!50!black, baseline=0cm]
\draw[line width=0] 
(0,1.25) node[line width=0pt] (A) {{\small $X$}}
(0,-1.25) node[line width=0pt] (A2) {{\small $X$}}; 
\draw (A2) -- (A); 
\end{tikzpicture}
\, , \qquad
\begin{tikzpicture}[very thick,scale=1.0,color=blue!50!black, baseline=0cm]
\draw[line width=0] 
(-1,1.25) node[line width=0pt] (A) {{\small $X^\vee$}}
(1,-1.25) node[line width=0pt] (A2) {{\small $X^\vee$}}; 
\draw[redirected] (0,0) .. controls +(0,-1) and +(0,-1) .. (-1,0);
\draw[redirected] (1,0) .. controls +(0,1) and +(0,1) .. (0,0);
\draw (-1,0) -- (A); 
\draw (1,0) -- (A2); 
\end{tikzpicture}
=
\begin{tikzpicture}[very thick,scale=1.0,color=blue!50!black, baseline=0cm]
\draw[line width=0] 
(0,1.25) node[line width=0pt] (A) {{\small $X^\vee$}}
(0,-1.25) node[line width=0pt] (A2) {{\small $X^\vee$}}; 
\draw (A2) -- (A); 
\end{tikzpicture}
\ee
since $\ev_{X}, \coev_{X}$ satisfy~\eqref{sternstern}. 

Left and right dualities defined as above satisfy the following standard identities which we will need in the next section; for the convenience of the reader, we have included a proof in appendix~\ref{left-from-right-lemma}.

\begin{lemma} \label{isos-nat-monoid}
Let $\mathcal M$ and $t_X$ be as above.
\begin{enumerate}
\item
$t$ is a natural isomorphism $\operatorname{Id}\Rightarrow (\,\cdot\,)^{\vee\vee}$ iff for all $X,Y \in \mathcal M$ and all $\varphi : X \rightarrow Y$ we have
\be\label{sovereign}
\begin{tikzpicture}[very thick,scale=1.0,color=blue!50!black, baseline]
\draw[line width=1pt] 
(2,-1) node[line width=0pt] (Y) {{\small $Y^\vee$}}
(0,1) node[line width=0pt] (X) {{\small $X^\vee$}}
(1,0) node[inner sep=2pt,draw] (f) {{\small$\varphi$}}; 
\draw[odirected] (f) .. controls +(0,1) and +(0,1) .. (2,0);
\draw[udirected] (0,0) .. controls +(0,-1) and +(0,-1) .. (f);
\draw (Y) -- (2,0);
\draw (X) -- (0,0);
\end{tikzpicture}
=
\begin{tikzpicture}[very thick,scale=1.0,color=blue!50!black, baseline]
\draw[line width=1pt] 
(0,-1) node[line width=0pt] (Y) {{\small $Y^\vee$}}
(2,1) node[line width=0pt] (X) {{\small $X^\vee$}}
(1,0) node[inner sep=2pt,draw] (f) {{\small$\varphi$}}; 
\draw[odirected] (f) .. controls +(0,1) and +(0,1) .. (0,0);
\draw[udirected] (2,0) .. controls +(0,-1) and +(0,-1) .. (f);
\draw (Y) -- (0,0);
\draw (X) -- (2,0);
\end{tikzpicture}
\, .
\ee
\item 
If~$t$ is a natural isomorphism, then~$t$ is monoidal iff for all $X,Y \in \mathcal M$ one has
\be\label{bends}
\begin{tikzpicture}[very thick,scale=0.8,color=blue!50!black, baseline]
\draw[line width=1pt] 
(-2,-1) node[line width=0pt] (X) {{\small $Y^\vee$}}
(-1,-1) node[line width=0pt] (Y) {{\small $X^\vee$}}
(2,2) node[line width=0pt] (XY) {{\small $(X\otimes Y)^\vee$}}; 
\draw[directed] (0,0) .. controls +(0,1) and +(0,1) .. (-1,0);
\draw[directed] (1,0) .. controls +(0,2) and +(0,2) .. (-2,0);
\draw[directed] (2,0) .. controls +(0,-1) and +(0,-1) .. (0.5,0);
\draw (-1,0) -- (Y);
\draw (-2,0) -- (X);
\draw[dotted] (0,0) -- (1,0);
\draw (2,0) -- (XY);
\end{tikzpicture}
=
\begin{tikzpicture}[very thick,scale=0.8,color=blue!50!black, baseline]
\draw[line width=1pt] 
(2,-1) node[line width=0pt] (Y) {{\small $Y^\vee$}}
(3,-1) node[line width=0pt] (X) {{\small $X^\vee$}}
(-1,2) node[line width=0pt] (XY) {{\small $(X\otimes Y)^\vee$}}; 
\draw[directed] (1,0) .. controls +(0,1) and +(0,1) .. (2,0);
\draw[directed] (0,0) .. controls +(0,2) and +(0,2) .. (3,0);
\draw[directed] (-1,0) .. controls +(0,-1) and +(0,-1) .. (0.5,0);
\draw (2,0) -- (Y);
\draw (3,0) -- (X);
\draw[dotted] (0,0) -- (1,0);
\draw (-1,0) -- (XY);
\end{tikzpicture}
\, .
\ee
\end{enumerate}
\end{lemma}

\subsubsection{Pivotal structure for matrix factorisations}

After the general discussion, we now turn to $\MFW$. For any $X\in\MFW$, we can use bimodule maps $\delta_{X_{i}}$ from~\eqref{Delta} to obtain the following isomorphisms in $\MFW$, which we also denote by $\delta$,
\be\label{delta-for-MFbi}
\delta_{X} = \begin{pmatrix}\delta_{X_{0}}&0\\0&\delta_{X_{1}}\end{pmatrix} : X \longrightarrow X^{\vee\vee} \, .
\ee

\begin{theorem}\label{thm-pivotal}
$\delta:\operatorname{Id} \Rightarrow (\,\cdot\,)^{\vee\vee}$ endows $\MFW$ with a pivotal structure. 
\end{theorem}
\begin{proof}
We first show that~$\delta$ really is a natural transformation. This means that for any $\varphi\in\Hom_{\MFW}(X,Y)$ we must have $\varphi^{\vee\vee}\circ \delta_{X}=\delta_{Y}\circ \varphi$. Writing out this condition as a matrix equation, it immediately follows from the corresponding identity on the level of bimodule maps. Indeed, for $m\in X_{i}$ and $w\in Y_{i}^\vee$ we have
\begin{align}
((\varphi_{i}^{\vee\vee}\circ\delta_{X_{i}})(m))(w) 
& = (\delta_{X_{i}}(m))(\varphi_{i}^\vee(w)) = \sigma_{R,R}((\varphi_{i}^\vee(w))(m))
\nonumber  \\
& = \sigma_{R,R}(w(\varphi_{i}(m))) = (\delta_{Y_{i}}(\varphi_{i}(m)))(w) 
\nonumber  \\
& = ((\delta_{Y_{i}}\circ \varphi)(m))(w) \, . 
\end{align}

Now we shall prove that~$\delta$ is also monoidal. By lemma~\ref{isos-nat-monoid}(ii), doing so is equivalent to establishing that
\be
L :=
\begin{tikzpicture}[very thick,scale=0.8,color=blue!50!black, baseline]
\draw[line width=1pt] 
(-2,-1) node[line width=0pt] (X) {{\small $X^\vee$}}
(-1,-1) node[line width=0pt] (Y) {{\small $Y^\vee$}}
(2,2) node[line width=0pt] (XY) {{\small $(Y\otimes X)^\vee$}}; 
\draw[directed] (0,0) .. controls +(0,1) and +(0,1) .. (-1,0);
\draw[directed] (1,0) .. controls +(0,2) and +(0,2) .. (-2,0);
\draw[directed] (2,0) .. controls +(0,-1) and +(0,-1) .. (0.5,0);
\draw (-1,0) -- (Y);
\draw (-2,0) -- (X);
\draw[dotted] (0,0) -- (1,0);
\draw (2,0) -- (XY);
\end{tikzpicture}
=
\begin{tikzpicture}[very thick,scale=0.8,color=blue!50!black, baseline]
\draw[line width=1pt] 
(2,-1) node[line width=0pt] (Y) {{\small $X^\vee$}}
(3,-1) node[line width=0pt] (X) {{\small $Y^\vee$}}
(-1,2) node[line width=0pt] (XY) {{\small $(Y\otimes X)^\vee$}}; 
\draw[directed] (1,0) .. controls +(0,1) and +(0,1) .. (2,0);
\draw[directed] (0,0) .. controls +(0,2) and +(0,2) .. (3,0);
\draw[directed] (-1,0) .. controls +(0,-1) and +(0,-1) .. (0.5,0);
\draw (2,0) -- (Y);
\draw (3,0) -- (X);
\draw[dotted] (0,0) -- (1,0);
\draw (-1,0) -- (XY);
\end{tikzpicture}
=: R
\ee
for all $X,Y\in\MFW$. Note that $L=\nu_{X,Y}^2$ as in~\eqref{gamma}. 

Written in terms of their matrix representatives, $L$ and $R$ are of the form
\be
\bordermatrix{
   & \text{{\scriptsize {$(X^\vee)_{0} \otimes (Y^\vee)_{0}$}}} & \text{{\scriptsize {$(X^\vee)_{1} \otimes (Y^\vee)_{1}$}}} & \text{{\scriptsize {$(X^\vee)_{1} \otimes (Y^\vee)_{0}$}}} & \text{{\scriptsize {$(X^\vee)_{0} \otimes (Y^\vee)_{1}$}}} \cr
\text{{\scriptsize {$(Y_{1} \otimes X_{0})^\vee$}}}   & * & * & 0 & 0 \cr
\text{{\scriptsize {$(Y_{0} \otimes X_{1})^\vee$}}}   & * & * & 0 & 0 \cr
\text{{\scriptsize {$(Y_{0} \otimes X_{0})^\vee$}}}   & 0 & 0 & * & * \cr
\text{{\scriptsize {$(Y_{1} \otimes X_{1})^\vee$}}}   & 0 & 0 & * & *
}
\label{L-long-expression}
\ee
where the column and row labels indicate the entries' source and target, respectively, as dictated by our convention~\eqref{XtensorY} for tensor products (here and below, ``$\otimes$'', when applied to $R$-bimodules, stands for ``$\otimes_R$''). Recall also from \eqref{Xdual} that $(X^\vee)_0 = X_1^\vee$, etc.
A straightforward but tedious calculation (aided by the fact that certain entries of $\lambda,\rho$ and $\ev,\tev$ are zero) shows that both~$L$ and~$R$ have only one non-vanishing entry in each $\Z_{2}$-degree, e.\,g.~$L_{0}=(\begin{smallmatrix}l_{11}&0\\0&0\end{smallmatrix})$ with
\begin{align}
l_{11} & = \Big( \Big\{ \mu \circ A_{X} \circ  
\left( \id_{X^\vee_{1}} \otimes (\mu \circ A_{Y}) \otimes \id_{X_{0}}\right) \!\Big\} \otimes \id_{(Y_{1}\otimes X_{0})^\vee} \Big) \nonumber \\
& \qquad \circ \Bigg( \id_{X^\vee_{1}} \otimes \id_{Y^\vee_{1}} \otimes \left\{ \left[\tfrac{\check d^{Y\otimes X}_{1}(a,x) - \check d^{Y\otimes X}_{1}(b,x)}{a-b}\tc \id_R \tc \id_{(\check Y_1 \otimes R \otimes \check X_0)^*} \right]^{\!\wedge} \right\} \nonumber \\
&  \qquad \qquad \circ \pi^{1010} \circ c_{Y\otimes X} \Bigg) \circ \left(\id_{X^\vee_{1}} \otimes [\id_{\check Y^*_{1}} \tc 1]\hat~ \right)
\, ,
\end{align}
where $\pi^{1010}$ projects to $Y_{1}\otimes X_{0}\otimes (Y_{1}\otimes X_{0})^\vee$, and similarly for $R_{0}=(\begin{smallmatrix}0&0\\r_{21}&0\end{smallmatrix})$. These expressions can be simplified (in a way very similar to~\eqref{puh}) and one finds
\be
l_{11} = f \circ \left(d_{0}^{X^\vee} \otimes \id_{Y^\vee_{1}} \right) \, , \quad
r_{21} = g \circ \left( \id_{X^\vee_{1}} \otimes d_{0}^{Y^\vee} \right)
\ee
with maps 
\begin{align}
f & : X^\vee_{0} \otimes Y^\vee_{1} \longrightarrow R \otimes \check X^*_{0} \otimes R^* \otimes \check Y^*_{1} \otimes R \subset \left( (Y\otimes X)^\vee \right)_{0} \, , \\
g & : X^\vee_{1} \otimes Y^\vee_{0} \longrightarrow R \otimes \check X^*_{1} \otimes R^* \otimes \check Y^*_{0} \otimes R \subset \left( (Y\otimes X)^\vee \right)_{0} 
\end{align}
given by
\be
f = 
\begin{tikzpicture}[very thick,scale=0.8,color=blue!50!black, baseline=(F.base),line/.style={&gt;=latex}]
\draw[line width=1pt] 
(1,0) node[inner sep=2pt] (Rl) {{\small$R$}} 
(2,0) node[inner sep=2pt] (X) {{\small$\check X^*_{0}$}}
(3,0) node[inner sep=2pt] (Rm) {{\small$R$}} 
(4,0) node[inner sep=2pt] (Y) {{\small$\check Y^*_{1}$}}
(5,0) node[inner sep=2pt] (Rr) {{\small$R$}} 

(2.5, 2.5) node[inner sep=2pt,draw] (F) {{\small$F$}} 

(1,5) node[inner sep=2pt] (Rlu) {{\small$R$}} 
(2,5) node[inner sep=2pt] (Xu) {{\small$\check X^*_{0}$}}
(3,5) node[inner sep=2pt] (Rmu) {{\small$R^*$}} 
(4,5) node[inner sep=2pt] (Yu) {{\small$\check Y^*_{1}$}}
(5,5) node[inner sep=2pt] (Rru) {{\small$R$}};

\draw[dotted] (Rl) -- (Rlu); 
\draw[dotted] (Rr) -- (Rru); 
\draw (Y) -- (Yu); 
\draw (X) -- (Xu); 

\draw[dotted] (Rmu) -- (3, 2.4); 
\draw[dotted] (3, 2.4) .. controls +(0,-1) and +(0,-1) .. (F);

\draw[dotted] (Rm) .. controls +(0,1) and +(0,-1) .. (F);
\draw[dotted] (F) .. controls +(0,1) and +(0,-1) .. (1,4);
\end{tikzpicture} 
\, , \quad
g = 
\begin{tikzpicture}[very thick,scale=0.8,color=blue!50!black, baseline=(F.base),line/.style={&gt;=latex}]
\draw[line width=1pt] 
(1,0) node[inner sep=2pt] (Rl) {{\small$R$}} 
(2,0) node[inner sep=2pt] (X) {{\small$\check X^*_{1}$}}
(3,0) node[inner sep=2pt] (Rm) {{\small$R$}} 
(4,0) node[inner sep=2pt] (Y) {{\small$\check Y^*_{0}$}}
(5,0) node[inner sep=2pt] (Rr) {{\small$R$}} 

(3.5, 2.5) node[inner sep=2pt,draw] (F) {{\small$F$}} 

(1,5) node[inner sep=2pt] (Rlu) {{\small$R$}} 
(2,5) node[inner sep=2pt] (Xu) {{\small$\check X^*_{1}$}}
(3,5) node[inner sep=2pt] (Rmu) {{\small$R^*$}} 
(4,5) node[inner sep=2pt] (Yu) {{\small$\check Y^*_{0}$}}
(5,5) node[inner sep=2pt] (Rru) {{\small$R$}};

\draw[dotted] (Rl) -- (Rlu); 
\draw[dotted] (Rr) -- (Rru); 
\draw (Y) -- (Yu); 
\draw (X) -- (Xu); 

\draw[dotted] (Rmu) -- (3, 2.4); 
\draw[dotted] (3, 2.4) .. controls +(0,-1) and +(0,-1) .. (F);

\draw[dotted] (Rm) .. controls +(0,1) and +(0,-1) .. (F);
\draw[dotted] (F) .. controls +(0,1) and +(0,-1) .. (5,4);
\end{tikzpicture} 
\ee
where $F: \C[x] \to \C[y]$, $x^k \mapsto \frac{1}{2\pi\I} \oint \frac{x^k \D x}{W(x)-W(y)}$ for appropriate variables~$x$ and~$y$, and we implicitly insert the multiplication~$\mu$ where the dotted $R$-lines meet. 
(The map $F$ arises from $\mathcal F$ in \eqref{Fcal} after composing with $\mu$ in \eqref{L-long-expression}, which cancels the $a-b$ part.)

We now claim that $L-R$ is null-homotopic, so that $L=R$ in $\MFW$. To show this we will use the fact that two morphisms $\phi,\psi:Z\rightarrow Z'$ in $\MFW$ are the same if (but not only if)
$\phi_{0}$ and $\psi_{0}$ induce the same maps $\operatorname{coker}(d^Z_{1})\rightarrow \operatorname{coker}(d^{Z'}_{1})$, see e.\,g.~\cite{eisen1980}.

As a first step we note that
\be
R'_{0} = \begin{pmatrix}
0 & 0 \\
g'  \circ ( \id_{X^\vee_{1}} \otimes d_{0}^{Y^\vee} ) & 0 
\end{pmatrix}
\, , \quad 
g' = 
\begin{tikzpicture}[very thick,scale=0.8,color=blue!50!black, baseline=(F.base),line/.style={&gt;=latex}]
\draw[line width=1pt] 
(1,0) node[inner sep=2pt] (Rl) {{\small$R$}} 
(2,0) node[inner sep=2pt] (Y) {{\small$\check X^*_{1}$}}
(3,0) node[inner sep=2pt] (Rm) {{\small$R$}} 
(4,0) node[inner sep=2pt] (X) {{\small$\check Y^*_{0}$}}
(5,0) node[inner sep=2pt] (Rr) {{\small$R$}} 

(2.5, 2.5) node[inner sep=2pt,draw] (F) {{\small$F$}} 

(1,5) node[inner sep=2pt] (Rlu) {{\small$R$}} 
(2,5) node[inner sep=2pt] (Yu) {{\small$\check X^*_{1}$}}
(3,5) node[inner sep=2pt] (Rmu) {{\small$R^*$}} 
(4,5) node[inner sep=2pt] (Xu) {{\small$\check Y^*_{0}$}}
(5,5) node[inner sep=2pt] (Rru) {{\small$R$}};

\draw[dotted] (Rl) -- (Rlu); 
\draw[dotted] (Rr) -- (Rru); 
\draw (X) -- (Xu); 
\draw (Y) -- (Yu); 

\draw[dotted] (Rmu) -- (3, 2.4); 
\draw[dotted] (3, 2.4) .. controls +(0,-1) and +(0,-1) .. (F);

\draw[dotted] (Rm) .. controls +(0,1) and +(0,-1) .. (F);
\draw[dotted] (F) .. controls +(0,1) and +(0,-1) .. (1,4);
\end{tikzpicture} 
\, ,
\ee
induces the same map on cokernels as $R_{0}$. This follows from the identity $\pi\circ f=\pi\circ g$ where~$\pi$ denotes the projection to the cokernel. This last identity, in turn, follows since $F(x^k) = c_0 \cdot 1 + c_1 \cdot W(y) + c_2 \cdot W(y)^2 + \ldots$ for some $k$-dependent numbers $c_i \in \C$, and after composing with $\pi$ the left action of $W$ on a bimodule equals the right action of $W$. As a result there is a morphism~$R'$ with $R=R'$ in $\MFW$ whose even component is $R'_{0}$. 

Secondly, consider the odd map
\be
h = 
\bordermatrix{
   & \text{{\scriptsize {$(X^\vee)_{0} \otimes (Y^\vee)_{0}$}}} & \text{{\scriptsize {$(X^\vee)_{1} \otimes (Y^\vee)_{1}$}}} & \text{{\scriptsize {$(X^\vee)_{1} \otimes (Y^\vee)_{0}$}}} & \text{{\scriptsize {$(X^\vee)_{0} \otimes (Y^\vee)_{1}$}}} \cr
\text{{\scriptsize {$(Y_{1} \otimes X_{0})^\vee$}}}   & 0 & 0& f & 0 \cr
\text{{\scriptsize {$(Y_{0} \otimes X_{1})^\vee$}}}   & 0 & 0 & 0 & -g' \cr
\text{{\scriptsize {$(Y_{0} \otimes X_{0})^\vee$}}}   & 0 & e & 0 & 0 \cr
\text{{\scriptsize {$(Y_{1} \otimes X_{1})^\vee$}}}   & 0 & 0 & 0 & 0
}
\label{L-long-expression}
\ee
with
\be
e = 
\begin{tikzpicture}[very thick,scale=0.8,color=blue!50!black, baseline=(F.base),line/.style={&gt;=latex}]
\draw[line width=1pt] 
(1,0) node[inner sep=2pt] (Rl) {{\small$R$}} 
(2,0) node[inner sep=2pt] (Y) {{\small$\check X^*_{0}$}}
(3,0) node[inner sep=2pt] (Rm) {{\small$R$}} 
(4,0) node[inner sep=2pt] (X) {{\small$\check Y^*_{0}$}}
(5,0) node[inner sep=2pt] (Rr) {{\small$R$}} 

(2.5, 2.5) node[inner sep=2pt,draw] (F) {{\small$F$}} 

(1,5) node[inner sep=2pt] (Rlu) {{\small$R$}} 
(2,5) node[inner sep=2pt] (Yu) {{\small$\check X^*_{0}$}}
(3,5) node[inner sep=2pt] (Rmu) {{\small$R^*$}} 
(4,5) node[inner sep=2pt] (Xu) {{\small$\check Y^*_{0}$}}
(5,5) node[inner sep=2pt] (Rru) {{\small$R$}};

\draw[dotted] (Rl) -- (Rlu); 
\draw[dotted] (Rr) -- (Rru); 
\draw (X) -- (Xu); 
\draw (Y) -- (Yu); 

\draw[dotted] (Rmu) -- (3, 2.4); 
\draw[dotted] (3, 2.4) .. controls +(0,-1) and +(0,-1) .. (F);

\draw[dotted] (Rm) .. controls +(0,1) and +(0,-1) .. (F);
\draw[dotted] (F) .. controls +(0,1) and +(0,-1) .. (1,4);
\end{tikzpicture}  
\, .
\ee
This provides a homotopy between $L$ and $R'$. Indeed, one can directly verify that $L-R'$ and $d^{(Y\otimes X)^\vee} \circ h + h \circ d^{X^\vee\otimes Y^\vee}$ induce the same map $\operatorname{coker}(d^{X^\vee\otimes Y^\vee}) \rightarrow \operatorname{coker}(d^{(Y\otimes X)^\vee})$. 
\end{proof}

Setting $t=\delta$ in subsection \ref{rightduals}, it follows that matrix bi-factorisations also have left duals. There is an analogous result to lemma~\ref{evlr} (and also a result analogous to remark \ref{manyduals}(iii)): 

\begin{lemma}\label{tevlr}
We have $\tev_{I}=\lambda_{I}=\rho_{I}$ and $\tcoev_{I}=\lambda^{-1}_{I}=\rho^{-1}_{I}$ in $\MFW$. 
\end{lemma}

\section{Defect action on bulk fields}\label{defectbulk}

As described in the introduction we can use the duality structure of $\MFW$ to study the action of defects on bulk fields. To do so, we first study the general situation of a pivotal monoidal category and then specialise to $\MFW$. We also compare the results obtained this way to the description in terms of the associated rational conformal field theory. 

\subsection[Action on $\mathrm{End}(I)$ for pivotal categories]{Action on $\boldsymbol{\mathrm{End}(I)}$ for pivotal categories}

Let $\mathcal M$ be a right rigid monoidal category with pivotal structure $\delta$. Then we also have left duals on $\mathcal M$ as in subsection \ref{rightduals}. Given an object $X \in \mathcal M$, one can define the maps
\be
 \mathcal D_{l}(X)\,,\,\mathcal D_{r}(X) : \End(I) \longrightarrow \End(I)
\ee 
as follows. For $\varphi : I \rightarrow I$ we set
\begin{subequations}\label{DlDr-def-general}
\begin{align}
\mathcal D_{l}(X)(\varphi) & = \ev_{X} \circ (\id_{X^\vee} \otimes (\lambda_{X} \circ (\varphi\otimes\id_{X}) \circ  \lambda_{X}^{-1})))\circ \tcoev_{X} \, , \\
\mathcal D_{r}(X)(\varphi) & = \tev_{X} \circ (((\rho_{X} \circ (\id_{X}\otimes\varphi) \circ  \rho_{X}^{-1}) \otimes \id_{X^\vee}) \circ \coev_{X} \, .
\end{align}
\end{subequations}
In pictorial notation, this amounts to
\be\label{defect-loops}
\mathcal D_{l}(X)(\varphi) = 
\begin{tikzpicture}[very thick,scale=0.8,color=blue!50!black, baseline]
\draw[line width=1pt] 
(1.5,1.3) node[inner sep=2pt,draw] (lam) {{\small$\lambda_{X}$}}
(1.5,-1.3) node[inner sep=2pt,draw] (lamin) {{\small$\lambda^{-1}_{X}$}}
(0,0) node[inner sep=2pt,draw] (f) {{\small$\varphi$}}; 
\draw[odirected] (lam) .. controls +(0,1.5) and +(0,1.5) .. (-1.5,1.3);
\draw[udirected] (-1.5,-1.3) .. controls +(0,-1.5) and +(0,-1.5) .. (lamin);
\draw (lamin) -- (lam)
node[midway,right] {{{\small$X$}}};
\draw (-1.5,-1.3) -- (-1.5,1.3);
\draw[dashed] (f) .. controls +(0,1) and +(-0.5,-1) .. (lam);
\draw[dashed] (f) .. controls +(0,-1) and +(-0.5,1) .. (lamin);
\draw[dashed] (0,-2.5) -- (0,-3.5)
node[near end,right] {{{\small$I$}}};
\draw[dashed] (0,2.47) -- (0,3.5)
node[near end,right] {{{\small$I$}}};
\end{tikzpicture}
\; , \quad
\mathcal D_{r}(X)(\varphi) =
\begin{tikzpicture}[very thick,scale=0.8,color=blue!50!black, baseline]
\draw[line width=1pt] 
(-1.5,1.3) node[inner sep=2pt,draw] (lam) {{\small$\rho_{X}\vphantom{rho_{X^\vee}}$}}
(-1.5,-1.3) node[inner sep=2pt,draw] (lamin) {{\small$\rho^{-1}_{X}$}}
(0,0) node[inner sep=2pt,draw] (f) {{\small$\varphi$}}; 
\draw[odirected] (lam) .. controls +(0,1.5) and +(0,1.5) .. (1.5,1.3);
\draw[udirected] (1.5,-1.3) .. controls +(0,-1.5) and +(0,-1.5) .. (lamin);
\draw (lamin) -- (lam)
node[midway,left] {{{\small$X$}}};
\draw (1.5,-1.3) -- (1.5,1.3);
\draw[dashed] (f) .. controls +(0,1) and +(0.5,-1) .. (lam);
\draw[dashed] (f) .. controls +(0,-1) and +(0.5,1) .. (lamin);
\draw[dashed] (0,-2.5) -- (0,-3.5)
node[near end,right] {{{\small$I$}}};
\draw[dashed] (0,2.47) -- (0,3.5)
node[near end,right] {{{\small$I$}}};
\end{tikzpicture} 
\; ,
\ee
as in~\eqref{Dl-on-phi}. 

\begin{lemma}\label{Dprop-general}
For all $X,Y \in \mathcal M$ we have: 
\begin{enumerate}
\item $\mathcal D_{l}(I)=\id=\mathcal D_{r}(I)$,
\item if $X \cong Y$ then $\mathcal D_{l}(X)=\mathcal D_{l}(Y)$ and $\mathcal D_{l}(X)=\mathcal D_{l}(Y)$,
\item $\mathcal D_{l}(X\otimes Y) = \mathcal D_{l}(Y)\circ \mathcal D_{l}(X)$, 
\item $\mathcal D_{r}(X\otimes Y) = \mathcal D_{r}(X)\circ \mathcal D_{r}(Y)$, 
\item $\mathcal D_{l}(X^\vee) = \mathcal D_{r}(X)$.
\end{enumerate}
\end{lemma}
\begin{proof}
For part (i) we note that since $\delta$ is pivotal, $\delta_I = \omega^0$ as given in subsection~\ref{pivotal}. Substituting the definition
\eqref{tev-tcoev-def} of $\tev_{X}$ and $\tcoev_{X}$ (with $t_X = \delta_X$), after a short calculation one arrives at the assertion. 

Part (ii) is a consequence of \eqref{stern} and \eqref{dual-morph-and-ev} as well as the naturality of $\lambda, \rho$ 
and~$\delta$.

(iii) \& (iv): It follows from~\eqref{bends} that
\be
\begin{tikzpicture}[very thick,scale=0.8,color=blue!50!black, baseline]
\draw[directed] (0,1) .. controls +(0,1) and +(0,1) .. (-1,1);
\draw[directed] (1,1) .. controls +(0,2) and +(0,2) .. (-2,1);
\draw[directed] (2,1) .. controls +(0,-1) and +(0,-1) .. (0.5,1);
\draw[redirected] (0,-1) .. controls +(0,-1) and +(0,-1) .. (-1,-1);
\draw[redirected] (1,-1) .. controls +(0,-2) and +(0,-2) .. (-2,-1);
\draw[redirected] (2,-1) .. controls +(0,1) and +(0,1) .. (0.5,-1);
\draw (-1,1) -- (-1,-1);
\draw (-2,1) -- (-2,-1);
\draw[dotted] (0,1) -- (1,1);
\draw[dotted] (0,-1) -- (1,-1);
\draw[line width=1pt] 
(2,3.3) node[line width=0pt] (XYo) {{\small $(X\otimes Y)^\vee$}}
(2,-3.3) node[line width=0pt] (XYu) {{\small $(X\otimes Y)^\vee$}};
\draw (2,1) -- (XYo);
\draw (2,-1) -- (XYu);
\end{tikzpicture}
\!=\!
\begin{tikzpicture}[very thick,scale=0.8,color=blue!50!black, baseline]
\draw[directed] (-3,1) .. controls +(0,1) and +(0,1) .. (-2,1);
\draw[directed] (-4,1) .. controls +(0,2) and +(0,2) .. (-1,1);
\draw[directed] (-5,1) .. controls +(0,-1) and +(0,-1) .. (-3.5,1);
\draw[redirected] (0,-1) .. controls +(0,-1) and +(0,-1) .. (-1,-1);
\draw[redirected] (1,-1) .. controls +(0,-2) and +(0,-2) .. (-2,-1);
\draw[redirected] (2,-1) .. controls +(0,1) and +(0,1) .. (0.5,-1);
\draw (-1,1) -- (-1,-1);
\draw (-2,1) -- (-2,-1);
\draw[dotted] (0,-1) -- (1,-1);
\draw[dotted] (-4,1) -- (-3,1);
\draw[line width=1pt] 
(-5,3.3) node[line width=0pt] (XYo) {{\small $(X\otimes Y)^\vee$}}
(2,-3.3) node[line width=0pt] (XYu) {{\small $(X\otimes Y)^\vee$}};
\draw (2,-1) -- (XYu);
\draw (-5,1) -- (XYo);
\end{tikzpicture}
\!=\!
\begin{tikzpicture}[very thick,scale=0.8,color=blue!50!black,baseline]
\draw[line width=1pt] 
(0,3.3) node[line width=0pt] (XYo) {{\small $(X\otimes Y)^\vee$}}
(0,-3.3) node[line width=0pt] (XYu) {{\small $(X\otimes Y)^\vee$}};
\draw (XYu) -- (XYo);
\end{tikzpicture}
\ee
where we used two Zorro moves in the second step. With this we can compute
\be
\mathcal D_{l}(X\otimes Y)(\varphi)  =\,
\begin{tikzpicture}[very thick,scale=0.8,color=blue!50!black, baseline]
\draw[line width=1pt] 
(0,0) node[inner sep=2pt,draw] (f) {{\small$\varphi$}}; 
\draw[directed] (1,0) .. controls +(0,1.5) and +(0,1.5) .. (-1,0);
\draw[directed] (-1,0) .. controls +(0,-1.5) and +(0,-1.5) .. (1,0);
\end{tikzpicture}
\,=\,
\begin{tikzpicture}[very thick,scale=0.8,color=blue!50!black, baseline]
\draw[line width=1pt] 
(5,0) node[inner sep=2pt,draw] (f) {{\small$\varphi$}}; 
\draw[directed] (2,1) .. controls +(0,1) and +(0,1) .. (1,1);
\draw[directed] (3,1) .. controls +(0,2) and +(0,2) .. (0,1);
\draw[directed] (1,-1) .. controls +(0,-1) and +(0,-1) .. (2,-1);
\draw[directed] (0,-1) .. controls +(0,-2) and +(0,-2) .. (3,-1);
\draw[directed] (4,1) .. controls +(0,-1) and +(0,-1) .. (2.5,1);
\draw[redirected] (4,-1) .. controls +(0,1) and +(0,1) .. (2.5,-1);
\draw[directed] (6,1) .. controls +(0,2) and +(0,2) .. (4,1);
\draw[redirected] (6,-1) .. controls +(0,-2) and +(0,-2) .. (4,-1);
\draw (0,1) -- (0,-1);
\draw (1,1) -- (1,-1);
\draw (6,1) -- (6,-1);
\draw[line width=1pt] 
(3,-0.78) node[line width=0pt] (Xd) {{\scriptsize $Y$}}
(2,-0.78) node[line width=0pt] (Yd) {{\scriptsize $X$}};
\draw[line width=1pt] 
(3,0.78) node[line width=0pt] (X) {{\scriptsize $Y$}}
(2,0.78) node[line width=0pt] (Y) {{\scriptsize $X$}};
\draw[dotted] (2,1) -- (3,1);
\draw[dotted] (2,-1) -- (3,-1);
\end{tikzpicture}
\,=\,
\begin{tikzpicture}[very thick,scale=0.8,color=blue!50!black, baseline]
\draw[line width=1pt] 
(1.5,0) node[inner sep=2pt,draw] (f) {{\small$\varphi$}}; 
\draw[directed] (2,1) .. controls +(0,1) and +(0,1) .. (1,1);
\draw[directed] (3,1) .. controls +(0,2) and +(0,2) .. (0,1);
\draw[directed] (1,-1) .. controls +(0,-1) and +(0,-1) .. (2,-1);
\draw[directed] (0,-1) .. controls +(0,-2) and +(0,-2) .. (3,-1);
\draw (0,1) -- (0,-1);
\draw (1,1) -- (1,-1);
\draw (2,1) -- (2,-1);
\draw (3,1) -- (3,-1);
\end{tikzpicture}
\ee
which is equal to $(\mathcal D_{l}(Y)\circ \mathcal D_{l}(X))(\varphi)$. 
$\mathcal D_{r}(X\otimes Y) = \mathcal D_{r}(X)\circ \mathcal D_{r}(Y)$ is proven in a similar way. 

(v): We have
\begin{align}
\mathcal D_{l}(X^\vee)(\varphi) & = \ev_{X^\vee} \circ \left( \id_{X^{\vee\vee}} \otimes \left[ \lambda_{X^\vee} \circ (\varphi \otimes\id_{X^\vee}) \circ\lambda^{-1}_{X^\vee} \right] \right) \nonumber \\
&\qquad \circ (\id_{X^{\vee\vee}}\otimes \delta_{X^\vee}^{-1}) \circ \coev_{X^{\vee\vee}} \nonumber \\
& \overset{(1)}{=} \ev_{X^\vee} \circ \left( \id_{X^{\vee\vee}} \otimes \left[ \lambda_{X^\vee} \circ (\varphi \otimes\id_{X^\vee}) \circ\lambda^{-1}_{X^\vee} \right] \right) \circ (\delta_{X} \otimes \id_{X^\vee}) \circ \coev_{X} \nonumber \\
& \overset{(2)}{=} \ev_{X^\vee} \circ (\delta_{X} \otimes \id_{X^\vee}) \circ \left( 
\left[ \rho_{X} \circ (\id_{X}\otimes\varphi) \circ  \rho_{X}^{-1}\right] \otimes \id_{X^\vee}
\right) \circ \coev_{X} \nonumber \\
& = \mathcal D_{r}(X)(\varphi) \, ,
\end{align}
where step (1) amounts to \eqref{tXvtinv}, that is
$(\delta_{X^\vee})^{-1} = \delta_{X}^\vee$, as well as~\eqref{Umove}, and in step (2) the identities $(\id_X \otimes \lambda_Y) \circ \alpha_{X,I,Y} =  \rho_X \otimes \id_Y : (X \otimes I)  \otimes Y \rightarrow X\otimes Y$ and their inverses are used. 
\end{proof}

In the defect picture, the above identities have immediate physical interpretations. For example, (i) simply expresses the fact that the action of the invisible defect leaves bulk fields invariant, independent of the orientation of the invisible defect. Similarly, (v) must be true since wrapping a defect~$X$ counterclockwise around a field is the same as wrapping the orientation reversed defect~$X^\vee$ clockwise. Properties~(iii) and~(iv) imply in particular that quantum dimensions behave multiplicatively under fusion, and one can use these properties to put constraints on the fusion decomposition. 

\subsection[Action on ${\rm End}(I)$ for ${\rm MF_{bi}}(W)$]{Action on $\boldsymbol{{\rm End}(I)}$ for $\boldsymbol{\MFW}$}
\label{triacomp}

Let us now consider the category $\MFW$. In the case of one variable, we can use the explicit expressions for $\ev_{X}, \coev_{X}, \lambda^{\pm 1}_{X}, \rho^{\pm 1}_{X}$ from the previous section to make the defect action~\eqref{DlDr-def-general} on bulk fields $\varphi \in \End_{\MFW}(I)$ concrete:
\begin{subequations}\label{DlDr}
\begin{align}
\mathcal D_{l}(X)(\varphi) & =  \left[ \frac{1}{2\pi\I} \oint \frac{\operatorname{tr}\left(\check d^X_{0}(x,b) \check d^X_{1}(x,a) \check\varphi_{0}(x)\right)\D x}{(W(x)-W(b))(b-a)}  \right]^{\!\wedge} \cdot \id \, , \label{Dl} \\
\mathcal D_{r}(X)(\varphi) & = \left[ \frac{1}{2\pi\I} \oint \frac{\operatorname{tr}\left(\check d^X_{0}(a,x) \check d^X_{1}(b,x) \check\varphi_{0}(x)\right)\D x}{(W(x)-W(a))(b-a)}  \right]^{\!\wedge} \cdot \id \, .\label{Dr}
\end{align}
\end{subequations}
The details of this calculation can be found in appendix~\ref{defacform}. 

\medskip

As we will recall momentarily, matrix (bi-)factorisations form a triangulated category~\cite{Neeman}, and the goal of this subsection is mainly to study the compatibility of the rigidity of $\MFW$ with its triangulated structure. The \textsl{distinguished triangles} of $\MFW$ are isomorphic to sequences of the form
\be\label{dt}
\xymatrix{X \ar[r]^-{\varphi} & Y\ar[r]^-{\zeta_{\varphi}} & \text{C}(\varphi) \ar[r]^-{\xi_{\varphi}} & TX} ,
\ee
see e.\,g.~\cite{o0302304}. Here the \textsl{cone} $\text{C}(\varphi)$ of a morphism $\varphi\in\Hom_{\MFW}(X,Y)$ is given by
\be\label{cone}
\text{C}(\varphi) =  \begin{pmatrix}0&0&-d^X_{0}&0\\ 0&0&\varphi_0&d^Y_{1}\\ -d^X_{1}&0&0&0\\ \varphi_1&d^Y_{0}&0&0\end{pmatrix} ,
\ee
where the matrix gives a bimodule endomorphism of $X_1 \oplus Y_0 \oplus X_0 \oplus Y_1$. The \textsl{shift functor}~$T$ acts as
\be
\begin{pmatrix}0&d^X_{1}\\d^X_{0}&0\end{pmatrix} \longmapsto \begin{pmatrix}0&-d^X_{0}\\-d^X_{1}&0\end{pmatrix} \, , \quad
\begin{pmatrix}\varphi_{0}&0\\0&\varphi_{1}\end{pmatrix} \longmapsto \begin{pmatrix}\varphi_{1}&0\\0&\varphi_{0}\end{pmatrix}
\ee
on objects and morphisms, respectively, and the two universal morphisms in~\eqref{dt} are
\be\label{triangle-morphisms}
\zeta_{\varphi}=\begin{pmatrix}
0&0\\
\id&0\\
0&0\\
0&\id
\end{pmatrix}
, \quad
\xi_{\varphi}=\begin{pmatrix}
\id&0&0&0\\
0&0&\id&0
\end{pmatrix} .
\ee
The \textsl{Grothendieck group} $\KMF(W))$ is the free abelian group of isomorphism classes of objects in $\MFW$ modulo the relations $[X]-[Y]+[Z]=0$ for all distinguished triangles $X\rightarrow Y\rightarrow Z\rightarrow TX$. 

The following lemma says that the tensor product of $\MFW$ induces a well-defined product on $\KMF(W))$, thus endowing it with a ring structure, and that the functor $(\,\cdot\,)^\vee$ induces a well-defined map on $\KMF(W))$. 

\begin{lemma} \label{K0-compatible}
Let $X\stackrel{\varphi}{\rightarrow} Y\to \text{C}(\varphi)\to TX$ be a distinguished triangle in $\MFW$. Then
\be\label{dual-triangle-relation}
[X^\vee]-[Y^\vee]+[\text{C}(\varphi)^\vee] = 0
\ee
in $\KMF(W))$, and
\begin{align}
& \xymatrix{Z\otimes X \ar[r]^-{\id\otimes\varphi} & Z\otimes Y\ar[r]^-{\id\otimes\zeta_{\varphi}} & Z\otimes \text{C}(\varphi) \ar[r]^-{\id\otimes\xi_{\varphi}} & Z\otimes TX} , \label{Dtria} \\
& \xymatrix{X\otimes Z \ar[r]^-{\varphi\otimes\id} &Y \otimes Z\ar[r]^-{\zeta_{\varphi}\otimes\id} & \text{C}(\varphi) \otimes Z \ar[r]^-{\xi_{\varphi}\otimes\id} & TX\otimes Z} \label{triaD}
\end{align}
are also distinguished triangles for all $Z\in\MFW$.
\end{lemma}
\begin{proof}
$Y^\vee\stackrel{\varphi^\vee}{\rightarrow} X^\vee\to \text{C}(\varphi^\vee)\to TY^\vee$ is a distinguished triangle, so we have the relation $[Y^\vee]-[X^\vee]+[\text{C}(\varphi^\vee)]=0$. But since $\big((\begin{smallmatrix} 0&-\id \\ \id&0\end{smallmatrix}), (\begin{smallmatrix} 0&\id \\ \id&0\end{smallmatrix})\big)$ is an isomorphism from $\text{C}(\varphi^\vee)$ to $T(\text{C}(\varphi)^\vee)$, the identity~\eqref{dual-triangle-relation} follows. 

To show that~\eqref{Dtria} is a distinguished triangle we observe that there is an isomorphism of triangles
\be\label{triaiso}
\xymatrix{%
Z\otimes X \ar[d]^{\id}\ar[r]^-{\id\otimes\varphi} & Z\otimes Y \ar[d]^{\id} \ar[r]^-{\zeta_{\id\otimes\varphi}} & \text{C}(\id\otimes\varphi) \ar[d]^{\Phi}\ar[r]^-{\xi_{\id\otimes\varphi}} & T(Z\otimes X)\ar[d]^{\Psi} \\
Z\otimes X \ar[r]^-{\id\otimes\varphi} & Z\otimes Y\ar[r]^-{\id\otimes\zeta_{\varphi}} & Z\otimes \text{C}(\varphi) \ar[r]^-{\id\otimes\xi_{\varphi}} & Z\otimes TX} 
\ee
where the maps~$\Phi$ and~$\Psi$ are given by
\be
\Phi = \left(
\begin{pmatrix}
0&\id&0&0\\
0&0&\id&0\\
-\id&0&0&0\\
0&0&0&\id
\end{pmatrix} , 
\begin{pmatrix}
0&-\id&0&0\\
0&0&\id&0\\
\id&0&0&0\\
0&0&0&\id
\end{pmatrix}
\right)
,\quad
\Psi =
\begin{pmatrix}
0&-\id&0&0\\
\id&0&0&0\\
0&0&0&\id\\
0&0&-\id&0
\end{pmatrix} .
\ee
That the squares in~\eqref{triaiso} commute easily follows from matrix multiplication. 
Checking that~\eqref{triaD} is also a distinguished triangle works analogously. 
\end{proof}

\begin{lemma}\label{C-on-K0-def}
The map $C: [X] \mapsto [X^\vee]$ defines an involutive ring anti-homomorphism on $\KMF(W))$.
\end{lemma}

\begin{proof}
By lemma~\ref{K0-compatible}, $\KMF(W))$ is a ring and $C$ is a well-defined map.
The isomorphism \eqref{nu2-iso} shows that $C([X])\cdot C([Y]) = [X^\vee \otimes Y^\vee] = [(Y \otimes X)^\vee] = C([Y]\cdot[X])$.
\end{proof}

\begin{lemma}\label{Dlr-and-triangles}
\begin{enumerate}
\item If $X \rightarrow Y \rightarrow Z \rightarrow TX$ is a distinguished triangle in $\MFW$, then 
$\mathcal D_{l}(X)-\mathcal D_{l}(Y)+\mathcal D_{l}(Z)=0=\mathcal D_{r}(X)-\mathcal D_{r}(Y)+\mathcal D_{r}(Z)$.
\item $\mathcal D_{l}(X) = -\mathcal D_{l}(TX)$, $\mathcal D_{r}(X) = -\mathcal D_{r}(TX)$.
\end{enumerate}
\end{lemma}
\begin{proof}
(i): This follows immediately from the explicit expressions~\eqref{DlDr} and~\eqref{cone} as the (off-diagonal) morphism dependent part of the cone cannot contribute to the trace in~\eqref{DlDr}. 

(ii): By the axioms of triangulated categories, $X\to Y\to\text{C}(\varphi)\to TX$ is a distinguished triangle iff $Y\to\text{C}(\varphi)\to TX\to TY$ is distinguished, and one easily checks that $\text{C}(\varphi)\cong 0$ for an isomorphism~$\varphi$. Hence if we set $X=Y$ and $\varphi=\id_{X}$, then it follows from part~(i) that $\mathcal D_{l/r}(X) = -\mathcal D_{l/r}(TX)$. 
\end{proof}

The above lemma shows that the maps $\mathcal D_{l/r}$ descend to $K_0(\MFW)$. From their explicit form one also sees that the operators $\mathcal D_{l/r}$ are degree preserving. In other words:  
\begin{proposition}\label{prop:map-on-K0}
In the one-variable case the maps $\mathcal D_{l/r}$ induce ring (anti-) homomorphisms
\be\label{DDllrr}
K_0(\MFW) \longrightarrow \End^0(\End_{\MFW}(I)) \, .
\ee
\end{proposition}
We expect this to remain true in the case of many variables.

\begin{remark}\label{remark-defect-action}
There are alternative methods to compute the action of defects on bulk fields. Instead of using rigidity as in~\eqref{defect-loops} one may also employ the folding trick. Indeed, it suggests that an action $\mathcal D_{X}$ of a defect~$X$ on a bulk field~$\varphi$ is obtained as the one-point-correlator of~$\varphi$, viewed as a field in the folded theory, in the presence of the boundary condition $B_{X}$ corresponding to~$X$. More precisely, one expects
\be\label{folding-defect-action}
\left\langle \mathcal D_{X}(\varphi) \, \psi \right\rangle
=
\left\langle\;\,
\begin{tikzpicture}[very thick,scale=0.8,color=blue!50!black, inner sep=7mm, baseline] 
\node[circle,circular glow={fill= blue!65!white},fill=white,draw= blue!50!black,thick] at (0,0) {}; 
\draw[line width=0pt] (0,0) node[inner sep=2pt] (f) {{\small$\varphi\otimes\psi$}}; 
\draw[line width=0pt] (0,-0.85) node[inner sep=2pt] (f) {{\small$B_{X}$}};
\end{tikzpicture}
\;\,\right\rangle
\ee
to hold for all bulk fields $\varphi,\psi$, where the left-hand side is a bulk correlator, and the right-hand side is a one-point-correlator of a bulk field in the presence of a boundary condition. 

In the case of topologically B-twisted Landau-Ginzburg models, such correlators can be computed with the residue formulas of~\cite{v1991} and~\cite{kl0305, hl0404184} (see also~\cite{s0904.1339, m0912.1629, dm1004.0687}): The two-point-correlator in the bulk is given by
\be\label{top-bulk-metric}
\langle \varphi \, \psi \rangle = \frac{1}{(2\pi\I)^N} \oint_{ \{ |\partial_{i}W| =1\} } \frac{\varphi\psi\, \D x_{1}\wedge\ldots\wedge\D x_{N}}{\partial_{1}W\ldots \partial_{N}W} \, ,
\ee
and the one-point-correlator of a bulk field~$\varphi$ in the presence of a boundary condition described by a matrix factorisation~$Q$ is
\be\label{top-bulk-bound}
\langle \varphi \rangle_{Q} = \frac{1}{(2\pi\I)^N} \oint_{ \{ |\partial_{i}W| =1\} } \frac{\varphi \, \operatorname{str}\left( \partial_{1} Q \ldots \partial_{N} Q \right) \D x_{1}\wedge\ldots\wedge\D x_{N}}{\partial_{1}W\ldots \partial_{N}W} \, .
\ee
Thus we can read off from~\eqref{folding-defect-action} that the defect action $\mathcal D_{X}$ is given by
\be\label{folded-defect-action}
\varphi \longmapsto \frac{1}{(2\pi\I)^N} \oint_{ \{ |\partial_{x_{i}}W| =1\} } \frac{\varphi \,\operatorname{str}\left( \partial_{x_{1}} X \partial_{y_{1}} X \ldots \partial_{y_{N}} X \right) \D x_{1}\wedge\ldots\wedge\D x_{N}}{\partial_{x_{1}}W\ldots \partial_{x_{N}}W} 
\ee
where the right-hand side is an element of the Jacobi ring of~$W$ in the $y$-variables. In the one-variable case one can check that the above $\mathcal D_{X}$ precisely coincides with our defect operator $\mathcal D_{l}(X)$ in~\eqref{Dl} for all the classes of examples that we will discuss below. 

Another way to arrive at~\eqref{folded-defect-action} is to use the theory of differential graded categories as follows. The space of bulk fields is isomorphic to the Hochschild homology $HH_{\bullet}(\operatorname{DG}(W))$ \cite[thm.~5.7]{d0904.4713}, on which the action of $X\in\MFW$ induces a map \cite[thm.~3.4]{s0710.1937} in terms of the canonical pairing on $HH_{\bullet}(\operatorname{DG}(W))$ and the Chern character of~$X$. Using \cite[cor.~4.1.3]{pv1002.2116} to make this explicit for Landau-Ginzburg models, one recovers~\eqref{folded-defect-action}. 
\end{remark}

\subsection{Examples}\label{examples}

All explicitly known matrix bi-factorisations of $W=x^d$ are isomorphic to direct sums of two distinct families of objects in $\MF(x^d)$ \cite{o0302304, add0401}. One of these is formed by the factorised matrix bi-factorisations
\be\label{Fact}
F_{i,j} = \begin{pmatrix} 0&  \hat a^i \\ \hat a^{d-i} & 0\end{pmatrix} \tc 
\begin{pmatrix} 0&  \hat b^j \\ \hat b^{d-j} & 0\end{pmatrix} \, , \quad i,j\in\{ 1,\ldots,d-1 \} \, . 
\ee
The other family is made up of elements labelled by all subsets $S\subset \{0,\ldots,d-1\}$ and given by
\be\label{PS}
P_{S} = \begin{pmatrix}0& \hat p_{S} \\ \hat p_{\{0,\ldots,d-1\}\backslash S} & 0\end{pmatrix}
\ee
where $p_{S}(a,b) = \prod_{i\in S}(a-\eta^i b)$ and $\eta=\E^{2\pi\I/d}$. In this subsection we study the action of such defects on bulk fields. We shall find agreement (up to phases) with results obtained by a dual description in terms of rational conformal field theory, and that the maps~\eqref{DDllrr} are bijective when restricted to the subring of $\KMF(x^d))$ generated by the isomorphism classes of~\eqref{Fact} and~\eqref{PS}. 

It was shown in~\cite[sec.~3.3.2]{bg0503} that $[F_{i,j}]=0$ in $\KMF(x^d))$, and a direct computation using the explicit expressions~\eqref{DlDr} shows that also
\be\label{DlrF}
\mathcal D_{l/r}(F_{i,j}) = 0 \, ,
\ee
as has to be the case by lemma~\ref{Dlr-and-triangles}.

Turning to the rank-one matrix bi-factorisations~$P_{S}$, we will now compute $\mathcal D_{l}(P_{S})$. Let us identify~$x^{i}$ with~$(\begin{smallmatrix}\hat a^{i} & 0 \\ 0 & \hat a^{i}\end{smallmatrix})\in\End_{\MF(W)}(I)\cong R/(\partial W)$. Substituting~\eqref{PS} into~\eqref{Dl} we find
\begin{align}\label{DlPS}
\mathcal D_{l}(P_{S})(x^{i}) & =  \left[ \frac{1}{2\pi\I} \oint \frac{x^i \prod_{l\notin S} (x-\eta^l b) \prod_{l\in S} (x-\eta^l a) \D x}{(x^d-b^d)(b-a)}  \right]^{\!\wedge} \nonumber \\
& =  \sum_{k=0}^{d-1}  \left[ (\eta^k b)^i \frac{\prod_{l\notin S} (\eta^k b-\eta^l b) \prod_{l\in S} (\eta^k b-\eta^l a) }{\prod_{m\neq k} (\eta^k b - \eta^m b) (b-a)}  \right]^{\!\wedge} \nonumber \\
& =  \sum_{k\in S} \left[ \eta^{ki} b^i \frac{\prod_{l\notin S} (\eta^k b-\eta^l b) \prod_{l\in S, l\neq k} (\eta^k b-\eta^l a) \eta^k (b-a) }{\prod_{m\neq k} (\eta^k b - \eta^m b) (b-a)}  \right]^{\!\wedge} \nonumber\\
& = \sum_{k\in S}\eta^{(i+1)k} x^i \, ,
\end{align}
where we used that $\hat a = \hat b$ on $\End_{\MFW}(I)$ in the last step. Similarly one obtains
\be\label{DrPS}
\mathcal D_{r}(P_{S})(x^{i}) = \sum_{k\in S} \eta^{-(i+1)k} x^i \, .
\ee

In the one-variable case with potential $W(x)=x^d$ the bulk two-point-correlator~\eqref{top-bulk-metric} simplifies to $\langle x^i \, x^j \rangle = \delta_{i+j,d-2}$. Hence it follows from~\eqref{DlPS} and~\eqref{DrPS} that $\mathcal D_{l}$ and $\mathcal D_{r}$ are adjoint in the following sense. 
\begin{proposition}
Let $X\in\MF(x^d)$ be isomorphic to a direct sum of objects of the form~\eqref{Fact} and~\eqref{PS}. Then
\be\label{adjoint-defect-in-2pf}
\langle \mathcal D_{l}(X)(\varphi) \, \psi \rangle = \langle \varphi \, \mathcal D_{r}(X)(\psi) \rangle \, . 
\ee
\end{proposition}
This result has a physical interpretation. Let us consider a worldsheet that is the Riemann sphere and that has two field insertions around one of which a topological defect is wrapped counterclockwise. As the defect is topological, the associated correlator has the same value if the defect is moved around the sphere to wrap the second field insertion: 
\be\label{beauty}
\left\langle
\begin{tikzpicture}[baseline]
\def\R{1.85}
\def\angEl{45}
\filldraw[ball color= white!70!blue,draw=white] (0,0) circle (\R);
\DrawLatitudeCircleU[\R,rotate=130,very thick, blue]{65}
\fill (-0.95,-0.83) circle (1pt) node[above] {{\small$\varphi$}}; 
\fill (0.95,-0.83) circle (1pt) node[above] {{\small$\psi$}}; 
\end{tikzpicture}
\right\rangle
=
\left\langle
\begin{tikzpicture}[baseline]
\def\R{1.85}
\def\angEl{45}
\filldraw[ball color= white!70!blue,draw=white] (0,0) circle (\R);
\DrawLongitudeCircle[\R]{80}
\fill (-0.95,-0.83) circle (1pt) node[above] {{\small$\varphi$}}; 
\fill (0.95,-0.83) circle (1pt) node[above] {{\small$\psi$}}; 
\end{tikzpicture}
\right\rangle
=
\left\langle
\begin{tikzpicture}[baseline]
\def\R{1.85}
\def\angEl{45}
\filldraw[ball color= white!70!blue,draw=white] (0,0) circle (\R);
\DrawLatitudeCircle[\R,rotate=-130, very thick, blue]{65}
\fill (-0.95,-0.83) circle (1pt) node[above] {{\small$\varphi$}}; 
\fill (0.95,-0.83) circle (1pt) node[above] {{\small$\psi$}}; 
\end{tikzpicture}\right\rangle .
\ee
Such a relation is expected to hold in any category of matrix bi-factorisations. Indeed, if we replace $\mathcal D_{l}(X)$ by $\mathcal D_{X}$ as in~\eqref{folded-defect-action} and $\mathcal D_{r}(X)$ by $\mathcal D_{X^\vee}$, then one easily checks that equation~\eqref{adjoint-defect-in-2pf} holds in general. 

We close this subsection by proving that the maps~\eqref{DDllrr} are bijective on all explicitly known matrix bi-factorisations of~$x^d$. 
\begin{proposition}\label{Grothmap}
The linear maps
\be
\mathcal D_{l/r} : \KMF(x^d)) \otimes_\Z \C \longrightarrow \End^0(\End_{\MF(x^d)}(I))
\ee
are surjective algebra (anti-)homomorphisms. Furthermore, when restricted to the subalgebra generated by elements of type~\eqref{Fact} and~\eqref{PS} they are isomorphisms. 
\end{proposition}
\begin{proof}
To see that~$\mathcal D_{l}$ is surjective (the case of~$\mathcal D_{r}$ works analogously) we will show that $\{\mathcal D_{l}(P_{\{ k \}})\}$ with $k\in\{ 0,\ldots, d-2\}$ is a basis for $\End^0(\End_{\MF(x^d)}(I))$. As $\End_{\MF(x^d)}(I)\cong R/(\partial x^d)$ is $(d-1)$-dimensional, an arbitrary element of $\End^0(\End_{\MF(x^d)}(I))$ is of the form $\text{diag}(\alpha_{0},\ldots,\alpha_{d-2})$ with~$\alpha_{l}$ in~$\C$. Let  us write such an element as $\sum_{k=0}^{d-2} \alpha_{l}\chi^l$. Then by~\eqref{DlPS}, $\mathcal D_{l}(P_{\{ k \}})$ is identified with $\sum_{l=0}^{d-2} \eta^{k(l+1)}\chi^l$. 

If we define numbers $\beta_{k}=\frac{1}{d}\sum_{i=0}^{d-1} \alpha_{i} \eta^{-k(i+1)}$, then any $\sum_{k=0}^{d-2} \alpha_{l}\chi^l$ is given by $\sum_{k=0}^{d-1} \beta_{k} \mathcal D_{l}(P_{\{ k \}}) = \sum_{k=0}^{d-1} \beta_{k} \sum_{l=0}^{d-2} \eta^{k(l+1)} \chi^l = \sum_{k=0}^{d-2} \alpha_{l}\chi^l$,
where we set $\alpha_{d-1}=0$. Thus~$\mathcal D_{l}$ and~$\mathcal D_{r}$ are surjective. 

As a first step to prove the second part of the proposition, we observe that it follows immediately from~\eqref{DlPS} and~\eqref{DrPS} that
\be\label{DSSs}
\mathcal D_{l/r}(P_{S}) + \mathcal D_{l/r}(P_{S'}) = \mathcal D_{l/r}(P_{S\cup S'}) + \mathcal D_{l/r}(P_{S\cap S'})
\ee
for all $S,S'\subset\{0,\ldots,d-1\}$. Let us denote by $\mathcal P\subset \KMF(x^d)) \otimes_\Z \C$ the subalgebra generated by all $[P_{S}]$. Then for $\mathcal D_{l/r}|_{\mathcal P}$ to be injective, we also must have $[P_{S}]+[P_{S'}]=[P_{S\cup S'}]+[P_{S\cap S'}]$ in~$\mathcal P$ for compatibility with~\eqref{DSSs}. This is indeed true, as we have distinguished triangles
\be\label{ker-f-triangle}
\xymatrix{%
P_{S} \ar[r]^-{\Phi} & P_{S\cup S'} \ar[r] & \text{C}(\Phi) \ar[r] & T P_{S}
}
\ee
where
\be
\Phi =
\begin{pmatrix}
\hat p_{S'\backslash (S\cap S')} & 0 \\
1&0\\
0&1\\
0& \hat p_{S\backslash (S\cap S')}
\end{pmatrix} ,
\ee
and by row and column manipulations one can show that $\text{C}(\Phi)\cong P_{S'}$. 

We now conclude the argument that $\mathcal D_{l/r}|_{\mathcal P}$ is bijective by simple linear algebra. Let us introduce a vector space~$V$ whose basis is labelled by all non-empty subsets~$S$ of $\{ 0,\ldots, d-1\}$, $V=\C \{ e_{S} \}$, and a linear map $f: V\to\End^0(\End_{\MF(x^d)}(I))$ with $f(e_{S})=\mathcal D_{l/r}(P_{S})$. Since~$f$ is surjective, $\operatorname{Ker}(f)$ is of codimension $d-1$. A convenient basis of $\operatorname{Ker}(f)$ is 
\be
  \{ e_{\{0,\ldots,d-1\}} \} \cup \{ e_{S}-{\textstyle \sum_{i\in S}e_{\{ i \}}} |\, |S|> 2 \} \, ,
\ee  
which is missing $e_{\{i\}}$ with $i=0,\dots,d-2$ to be a basis of $V$ (the element $e_{\{0\}} + \cdots +e_{\{d-1\}}$ is contained in the span of the above vectors) and so has the correct dimension.
Define the linear map $g : V \rightarrow\KMF(x^d)) \otimes_\Z \C$ via $g(e_{S})=[P_{S}]$. Since $P_{\{0,\dots,d-1\}} \cong 0$ in $\MFW$ we have $g(e_{\{0,\ldots,d-1\}}) = 0$, and from the triangle \eqref{ker-f-triangle} we see that $g(e_{S}+e_{S'}-e_{S\cup S'}-e_{S\cap S'})=0$. One checks that every vector in the above basis of $\operatorname{Ker}(f)$ can be written as a linear combination of vectors on which $g$ vanishes. Thus $g$ factors through $V/\operatorname{Ker}(f)$, and we have the commuting diagram
\be
\xymatrix{%
V/(\operatorname{Ker}(f)) \ar@{^{(}->>}[drrrr] \ar@{->>}[ddr] &&&& \\
& V \ar@{->>}[ul] \ar@{->>}[rrr]^-{f} \ar@{->>}[d]^-{g} &&& \End^0(\End_{\MF(x^d)}(I)) \\
& \mathcal P \ar@{->>}[rrru]_-{\mathcal D_{l/r}|_{\mathcal P}} &&&
}
\ee
which implies that $\mathcal D_{l/r}|_{\mathcal P}$ is bijective. 
\end{proof}

\subsection{Comparison with conformal field theory results}\label{CFTcomp}

We will now review the description of topological defects in rational conformal field theories associated to Landau-Ginzburg models with potential~$x^d$; then we shall compare defect actions and pivotal structures in both theories.

\subsubsection[Topological defects in $\mathcal N=2$ minimal models]{Topological defects in $\boldsymbol{\mathcal N=2}$ minimal models}

The vertex operator algebras $\sVir_d$ for the $\mathcal N=2$ minimal models form a discrete series labelled by an integer $d \in \Z_{\ge 3}$ and have central charge $c=3-6/d$. The bosonic part $(\sVir_d)_\text{bos}$ of $\sVir_d$ can be obtained via the coset construction from $\big(\widehat{\mathfrak{su}}(2)_{d-2} \oplus \widehat{\mathfrak{u}}(1)_4 \big) / \widehat{\mathfrak{u}}(1)_{2d}$. Accordingly, the isomorphism classes of irreducible representations of $(\sVir_d)_\text{bos}$ are labelled by the set
\begin{align}
\mathcal I = \big\{ (l,m,s) \;|\; & l \in \{0,1,\dots,d-2\} , m \in \{0,1,\dots,2d-1\},\nonumber \\
 & s \in \{0,1,2,3\}, l+m+s \text{ even} \big\} / \sim
\end{align}
where the equivalence relation $\sim$ identifies $(l,m,s)$ with $(d-2-l,m+d,s+2)$ for all $(l,m,s) \in \mathcal I$. We denote elements of  $\mathcal I$ by $[l,m,s]$, and hence we have $[l,m,s] = [d-2-l,m+d,s+2]$. For each isomorphism class of irreducible representations one may choose a representative $R_{[l,m,s]}$. We denote the category of representations of $(\sVir_d)_\text{bos}$ by 
$\mathcal{C}_d^{\mathcal N=2}$ (it is a $\C$-linear semisimple abelian braided monoidal category, which is in addition ribbon and modular).

The modular $S$-matrix for the characters of $R_{[l,m,s]}$ can also be found from the coset construction, and in the present case it is a simple product of the individual $S$-matrices, up to an overall factor: 
\be
 S_{[l,m,s],[x,y,z]} = 2 S^{\widehat{\mathfrak{su}}(2)_{d-2}}_{l,x} (S^{\widehat{\mathfrak{u}}(1)_{d}}_{m,y})^{*} S^{\widehat{\mathfrak{u}}(1)_2}_{s,z}
\ee
where 
\be
 S^{\widehat{\mathfrak{su}}(2)_{d-2}}_{a,b} = \sqrt{\frac{2}{d}} \sin\Big( \frac{\pi}{d}(a+1)(b+1) \Big)
 \, , \quad 
 S^{\widehat{\mathfrak{u}}(1)_N}_{a,b} = \sqrt{\frac{1}{2N}} \E^{-\pi \I ab/N} \, .
\ee

We consider the A-type $\mathcal N=2$ minimal models. The bosonic part of their space of states is given by
\be
\mathcal{H}_\text{bos} = \bigoplus_{[l,m,s] \in \mathcal I}
R_{[l,m,s]} \otimes \bar R_{[l,m,-s]} \, .
\label{eq:N=2-mm-bulk}
\ee
The chiral primaries are the highest weight states $\phi_{l,l,0}$ in the direct summands 
$R_{[l,l,0]} \otimes \bar R_{[l,l,0]}$ for $l \in \{0,1,\dots,d-2\}$. The fields $\phi_{l,l,0}$ have left/right conformal weight given by $h=\bar h = l/(2d)$, which for chiral primaries is also equal to half the $U(1)$-charge. 

The two-point-correlator of two fields $\psi,\psi' \in \mathcal{H}_\text{bos}$ on the Riemann sphere $\mathbbm P^1$ is given by 
\be
 \langle \psi(z) \psi'(w) \rangle = \kappa_{\psi\psi'} \, (z-w)^{-h_{\psi}-h_{\psi'}} (z^*-w^*)^{-\bar h_{\psi}-\bar h_{\psi'}} \, .
\ee
If $\psi$ and $\psi'$ are quasi-primary, the constant $\kappa_{\psi\psi'}$ can be non-zero only if $h_{\psi} = h_{\psi'}$ and $\bar h_{\psi} = \bar h_{\psi'}$. Let $\tilde\phi_{d-2,d+2,0}$ be a ground state in $R_{[d-2,d+2,0]} \otimes \bar R_{[d-2,d+2,0]}$ such that 
$\langle \tilde\phi_{d-2,d+2,0}(z)\phi_{d-2,d-2,0}(w) \rangle \neq 0$. Note that in order to have a non-zero two-point-correlator, by $U(1)$-charge conservation, $\tilde\phi_{d-2,d+2,0}$ needs to have minus the charge of $\phi_{d-2,d-2,0}$.
  
Consider the three-point-correlator with $\tilde\phi_{d-2,d+2,0}$ placed at infinity (with standard local coordinates around infinity on $\mathbbm P^1$), and insertions of $\phi_{r,r,0}$ and $\phi_{s,s,0}$ at $z$ and $w$. For an appropriate normalisation of the fields we have
\be \label{3pt-with-charge}
 \langle \tilde\phi_{d-2,d+2,0}(\infty) \phi_{r,r,0}(z) \phi_{s,s,0}(w) \rangle = \delta_{r+s,d-2} \, .
\ee
There is no position dependence because by $U(1)$-charge conservation, the correlator can be non-zero only for $r+s=d-2$, and in this case $h_{[d-2,d+2,0]}-h_{[r,r,0]}-h_{[s,s,0]} = 0$. The above three-point-correlator will correspond to the two-point-correlator \eqref{top-bulk-metric} in the topological Landau-Ginzburg model, which in the one-variable case just reads $\langle x^r \, x^s \rangle = \delta_{r+s,d-2}$. 

\medskip

The possible defects (preserving the holomorphic and anti-holomorphic copy of $(\sVir_d)_\text{bos}$) can be computed using the methods of \cite{Petkova:2000ip} (as done in \cite{br0707.0922}) or those of \cite{Frohlich:2006ch} (as done in \cite{cr0909.4381}). One finds that the elementary defects are also labelled by the set $\mathcal{I}$; we denote them as
\be
 X_{[l,m,s]} 
 \, , \quad [l,m,s] \in \mathcal{I} \, .
\ee
From \cite{tft1,Frohlich:2006ch} we know that the topological defects $X_{[l,m,s]}$ are simple objects in a monoidal category $\mathcal{D}_d^{\mathcal N=2}$ (see \cite[sec.~3.1\,\&\,app.~A.2]{cr0909.4381} for more details in the case at hand). The tensor product corresponds to fusion of defects and the morphism spaces
$\text{Hom}_{\mathcal{D}_d^{\mathcal N=2}}(X_1 \otimes \ldots \otimes X_m, Y_1 \otimes \ldots \otimes Y_n)$ are the spaces of defect junction fields of left/right conformal weight $(0,0)$ that
are inserted at a junction point with $m$ incoming defect lines labelled $X_1,\dots,X_m$ and
$n$ outgoing defect lines labelled $Y_1,\dots,Y_n$. As is the case for any rational conformal field theory, the category $\mathcal{D}_d^{\mathcal N=2}$ is (left and right) rigid and has a pivotal structure \cite{Frohlich:2006ch}.

Let us denote the defect operators (acting on bulk fields) of a defect $X$ by $\mathcal{D}^{\text{CFT}}_{l/r}(X)$. According to \cite[app.\,A.2]{cr0909.4381}, braided induction provides a monoidal equivalence $\mathcal{C}_d^{\mathcal N=2} \cong \mathcal{D}_d^{\mathcal N=2}$. 
The description of CFT correlators via three-dimensional topological field theory (3dTFT) \cite{tft1,Frohlich:2006ch} shows that in this case the defect operators are simply given by ratios of $S$-matrix elements.
For a field~$\psi$ in $R_{[l,m,s]} \otimes \bar R_{[l,m,-s]}$ one finds (see also \cite{br0707.0922})
\be
 \mathcal{D}^{\text{CFT}}_{l}(X_{[x,y,z]})(\psi) = \frac{S_{[l,m,s],[x,-y,-z]}}{S_{[l,m,s],[0,0,0]}} \psi
 \, , \quad
 \mathcal{D}^{\text{CFT}}_{r}(X_{[x,y,z]})(\psi) = \frac{S_{[l,m,s],[x,y,z]}}{S_{[l,m,s],[0,0,0]}} \psi \, .
\ee
We will be particularly interested in the action of defects $X_{x,y,0}$ on the chiral primaries $\phi_{l,l,0}$, for which we get explicitly
\begin{align}
 \mathcal{D}^{\text{CFT}}_{l}(X_{[x,y,0]})(\phi_{l,l,0}) &= \frac{\sin( \pi(x+1)(l+1)/d )}{\sin(\pi (l+1)/d)} \E^{+\pi \I y l/d} \phi_{l,l,0} \, ,
\nonumber\\
 \mathcal{D}^{\text{CFT}}_{r}(X_{[x,y,0]})(\phi_{l,l,0}) &= \frac{\sin( \pi(x+1)(l+1)/d )}{\sin(\pi (l+1)/d)} \E^{- \pi \I y l/d} \phi_{l,l,0} \, .
\label{X_xyz-on-phill0}
\end{align}

One can prove in the 3dTFT approach \cite{Frohlich:2006ch} that for all bulk fields $\psi,\psi' \in \mathcal{H}_\text{bos}$ and for all defects $X$ one has
\be
 \left\langle \mathcal{D}^{\text{CFT}}_{l}(X)\big(\psi(z)\big) \psi'(w) \right\rangle = 
 \left\langle \psi(z) \mathcal{D}^{\text{CFT}}_{r}(X)\big(\psi'(w)\big) \right\rangle \, ,
\ee
see the illustration~\eqref{beauty}. 
If $\psi$ and $\psi'$ are the identity field $\one$, this implies that 
\be \label{cft-def-Dl1=Dr1}
  \mathcal{D}^{\text{CFT}}_{l}(X) (\one) = \mathcal{D}^{\text{CFT}}_{r}(X) (\one) \, ,
\ee  
a result which holds for all rational conformal field theories whose left and right chiral symmetries coincide, so that they admit a description via the 3dTFT approach. The equality \eqref{cft-def-Dl1=Dr1} can also be read off from \eqref{X_xyz-on-phill0} upon setting $l=0$ as the multiplicative constant is then the quantum dimension of the representation $R_{[x,y,0]}$.

Again with the help of the 3dTFT approach one verifies that for all fields $\psi$,
\begin{align}
 &\mathcal{D}^{\text{CFT}}_{r}(X_{[x,y,z]})\big(\tilde\phi_{d-2,d+2,0}(w_1)\psi(w_2)\big)
\nonumber \\
 &= \frac{S_{[d-2,d+2,0],[x,y,z]}}{S_{[0,0,0],[x,y,z]}}
 \tilde\phi_{d-2,d+2,0}(w_1) \mathcal{D}^{\text{CFT}}_{r}(X_{[x,y,z]})\big(\psi(w_2)\big)
\nonumber  \\
 &= (-1)^{x+y} \E^{-2\pi \I y/d}
 \tilde\phi_{d-2,d+2,0}(w_1) \mathcal{D}^{\text{CFT}}_{r}(X_{[x,y,z]})\big(\psi(w_2)\big) \, .
\end{align}
Inserting this into \eqref{3pt-with-charge} results in
\begin{align}
&  \left\langle \tilde\phi_{d-2,d+2,0}(\infty) 
 \mathcal{D}^{\text{CFT}}_{l}(X_{[x,y,z]})\big(\phi_{r,r,0}(w_1) \big)
 \phi_{s,s,0}(w_2) 
 \right\rangle 
\nonumber \\
&  = (-1)^{z} \E^{-2\pi \I y/d}
 \left\langle \tilde\phi_{d-2,d+2,0}(\infty) 
 \phi_{r,r,0}(w_1) 
 \mathcal{D}^{\text{CFT}}_{r}(X_{[x,y,z]})\big( \phi_{s,s,0}(w_2) \big)
 \right\rangle  \, .
\label{3pt-with-charge-move-defect}
\end{align}

\subsubsection{Comparison of defect operators}

In \cite{br0707.0922}, $X_{[b,a+2b,0]}$ is identified as the conformal field theory equivalent of the  Landau-Ginzburg defect described by the matrix factorisation $P_{\{a,\dots,a+b\}}$ of \eqref{PS}. We write this as
\be\label{CFT-MF-fun-on-obj}
 F(X_{[b,2a+b,0]} ) = P_{\{a,\dots,a+b\}}
\ee
which will later provide the action of a functor $F$ on objects. The actions \eqref{DlPS} and \eqref{DrPS} of $P_{\{a,\dots,a+b\}}$ on $x^l$ can be rewritten as
\begin{align}
\mathcal{D}^\text{MF}_{l}(P_{\{a,\dots,a+b\}}) (x^l)
&= \frac{\sin( \pi(b+1)(l+1)/d )}{\sin(\pi (l+1)/d)} \E^{+\pi \I (l+1)(2a+b)/d} x^l \, ,
\nonumber \\
\mathcal{D}^\text{MF}_{r}(P_{\{a,\dots,a+b\}}) (x^l)
&= \frac{\sin( \pi(b+1)(l+1)/d )}{\sin(\pi (l+1)/d)} \E^{-\pi \I (l+1)(2a+b)/d} x^l \, .
\label{P_ab-on-xl}
\end{align}
The two actions \eqref{X_xyz-on-phill0} and \eqref{P_ab-on-xl} do not quite agree, for example
\be
  \mathcal{D}^{\text{MF}}_{l}(P_{\{a,\dots,a+b\}}) (1) = \E^{2\pi \I (2a+b)/d} \, \mathcal{D}^{\text{MF}}_{r}(P_{\{a,\dots,a+b\}}) (1) \, ,
\ee
so that \eqref{cft-def-Dl1=Dr1} does not hold for $\mathcal{D}^{\text{MF}}_{l/r}(X)$.
In general, if we define a linear map~$f$ from the space of chiral primaries to $\End_{\MF(x^d)}(I)$ by setting $f(\phi_{l,l,0}) = x^l$, then
\begin{align} 
\mathcal{D}^\text{MF}_{l}( F(X_{[x,y,0]}) ) (f( \phi_{l,l,0} ))
&= \E^{+\pi \I y/d} f( \mathcal{D}^\text{CFT}_{l}( X_{[x,y,0]} )  (\phi_{l,l,0}) )
\nonumber \, , \\
\mathcal{D}^\text{MF}_{r}( F(X_{[x,y,0]}) ) (f( \phi_{l,l,0}) )
&= \E^{-\pi \I y/d} f( \mathcal{D}^\text{CFT}_{r}( X_{[x,y,0]} )  (\phi_{l,l,0}) ) \, .
\label{action-on-chiral-prim-with-phase}
\end{align}
In fact, these prefactors are precisely what is needed in order to make~\eqref{adjoint-defect-in-2pf} and \eqref{3pt-with-charge-move-defect} fit to the observation that the pairing~\eqref{top-bulk-metric} 
is given by the three-point-correlator~\eqref{3pt-with-charge}.

\subsubsection{Comparison of pivotal structures}

Denote by $\mathcal{D}^{\mathcal N=2}_{d,s=0}$ the full subcategory of $\mathcal{D}^{\mathcal N=2}_{d}$ consisting of objects isomorphic to direct sums of objects of the form $X_{[l,m,0]}$; this is a monoidal subcategory (and hence also rigid and pivotal).
Similarly, let $(\mathcal{P}_{d})_{0}$ be the (non-full) subcategory of $\MF(x^d)$ whose objects are isomorphic to direct sums of objects of the form $P_{\{a,\dots,a+b\}}$ and whose morphisms are morphisms of R-charge zero in $\MF(x^d)$; this is again a monoidal (and rigid and pivotal) subcategory.

The assignment \eqref{CFT-MF-fun-on-obj} extends to an equivalence $F : \mathcal{D}^{\mathcal N=2}_{d,s=0} \rightarrow (\mathcal{P}_{d})_{0}$ of $\C$-linear semisimple categories (because it is bijective on representatives of the isomorphism classes of simple objects). It was conjectured in \cite{cr0909.4381} (and already established on the level of objects in  \cite{br0707.0922}) that $F$ can be extended to an equivalence $(F,F^2,F^0)$ of monoidal categories. From remark~\ref{left-rigid-all-equiv} we see that $(F,F^2,F^0)$ together with the right dualities of the source and target categories gives natural isomorphisms $\phi_X : F(X^\vee) \rightarrow F(X)^\vee$. 

We now want to see if $F$ is in addition pivotal in the sense of definition \ref{pivotal-functor-def}. This will turn out to be not the case, and to illustrate this we consider all pivotal structures on $\mathcal{D}^{\mathcal N=2}_{d,s=0}$ simultaneously by defining
\be
  \tilde \delta^{\text{CFT}}_X = \eta_X \circ \delta^{\text{CFT}}_X \, ,
\ee
where $\eta$ is a natural monoidal transformation of the identity functor on $\mathcal{D}^{\mathcal N=2}_{d,s=0}$. The functor $F$ will be pivotal for a unique choice of $\eta$. If this $\eta$ is different from the identity, $\mathcal{D}^{\mathcal N=2}_{d,s=0}$ and $(\mathcal{P}_{d})_{0}$ are not pivotally equivalent with their standard pivotal structures. We can fix $\eta_X$ on simple objects by noting that the requirement 
$\delta^{\text{MF}}_{F(X)} = (\psi_{X}^{-1})^\vee \circ \psi_{X^\vee} \circ F(\tilde \delta^{\text{CFT}}_X)$ from definition \ref{pivotal-functor-def}  implies
\be
 (F^0)^{-1} \circ F( \tev_X \circ (\eta_X \otimes \id_{X^\vee}) \circ \coev_X ) \circ F^0 = \tev_{F(X)} \circ \coev_{F(X)} \, .
\ee
Namely, if we write $\eta_{X_{[x,y,0]}} = \eta_{x,y} \id_{X_{[x,y,0]}}$, then $\eta_{x,y} F( \mathcal{D}^{\text{CFT}}_{r}(X_{[x,y,0]})(\one) )
= \mathcal{D}^\text{MF}_{r}(F(X_{[x,y,0]}))(\id)$, and comparison with~\eqref{action-on-chiral-prim-with-phase} shows that $\eta_{x,y} = \E^{-\pi \I y /d}$. Note that this is compatible with the fusion rules as is required for a natural monoidal transformation. Also, since $\eta_{x,y}$ is different from the identity,
\begin{quote} 
$\mathcal{D}^{\mathcal N=2}_{d,s=0}$ and $(\mathcal{P}_{d})_{0}$ are \textsl{not} pivotally equivalent. 
\end{quote}

\begin{remark}
This result raises the question if the difference of pivotal structures just observed can be avoided by redefining the pivotal structure on $\MF(x^d)$ in~\eqref{delta-for-MFbi}. In answer to this, we first note that it is of course possible to use the equivalence~$F$ to transport the pivotal structure from $ \mathcal{D}^{\mathcal N=2}_{d,s=0}$ to $(\mathcal{P}_{d})_{0}$, but it is not obvious that this pivotal structure then extends to all of $\MF(x^d)$. But rather than pursuing this point, we would like to offer an alternative perspective which we believe to be the correct interpretation of the above discrepancy. 

Our starting assumption is that a fundamental property of the notion of a ``defect operator'' should be that it factors through the relevant Grothendieck group of the category of defect conditions. This is satisfied on the CFT side (since $\mathcal{D}^{\mathcal N=2}_{d,s=0}$ is semi-simple, and so $K_{0}(\mathcal{D}^{\mathcal N=2}_{d,s=0})$ coincides with the free abelian group of isomorphism classes modulo the direct sum relation), and by proposition~\ref{prop:map-on-K0} it is also satisfied on the Landau-Ginzburg side with the pivotal structure~\eqref{delta-for-MFbi}. 

The property to factor through the Grothendieck group is tied to the pivotal structure and will in general fail if the pivotal structure is modified. In our example, this can be seen explicitly as follows. Let us consider the defect operator $\mathcal D_l$, the analysis for $\mathcal D_r$ gives the same result. Consider the identity defect~$I$ and its image $TI$ under the shift functor. The triangulated structure on $\MF(x^d)$ demands $[TI]=-[I]$ in $K_{0}(\MF(x^d))$. From the definition of the shift functor and the form of $P_S$ in~\eqref{PS} we conclude that $TI = TP_{\{0\}} \cong P_{\{1,\ldots,d-1\}}$. Formula~\eqref{P_ab-on-xl} now reproduces the answer we already knew from lemma~\ref{Dprop-general}(i) and proposition~\ref{prop:map-on-K0}: $\mathcal D_l(I)(x^0) = x^0$ and $\mathcal D_l(TI)(x^0) = - x^0$. On the other hand, it is equally easy to verify that the pivotal structure on $(\mathcal P_d)_0$ obtained by transporting that of $\mathcal{D}^{\mathcal N=2}_{d,s=0}$ does \textsl{not} factor through $K_{0}(\MF(x^d))$. Namely, by~\eqref{CFT-MF-fun-on-obj} we have $TI \cong F(X_{[d-2,d,0]})$ and from~\eqref{X_xyz-on-phill0} we see $\mathcal D^{\text{CFT}}_l(X_{[0,0,0]})(\one) = \one$ and $\mathcal D^{\text{CFT}}_l(X_{[d-2,d,0]})(\one) = \one$.

This calculation illustrates that the pivotal structure we chose in~\eqref{delta-for-MFbi} is adapted to the triangulated structure on $\MF(x^d)$ (in the sense that the defect operator factors through $K_{0}(\MF(x^d))$), while the pivotal structure obtained on $(\mathcal P_d)_0$ by transporting the one from $\mathcal{D}^{\mathcal N=2}_{d,s=0}$ (independent of whether it extends to all of $\MF(x^d)$ or not) is not. 
\end{remark}

The above observation shows that if we want to use the rigid structure to aid the comparison between matrix factorisation and conformal field theory data, we should look for quantities \textsl{independent} of the pivotal structure. One such quantity is the following. Let $\mathcal M$ be a $\C$-linear rigid monoidal category, and let~$\delta$ and $\delta' = \delta \circ \eta$ be two pivotal structures on $\mathcal M$, with $\eta$ a monoidal isomorphism of the identity functor on $\mathcal M$. For all $X \in \mathcal M$ we have
\be
  \ev_X \circ (\eta_{X^\vee} \otimes \eta_X) = 
  \ev_X \circ \eta_{X^\vee \otimes X} = 
  \eta_{I} \circ \ev_X = \ev_X \, .
\ee
Suppose now that $X$ is absolutely simple, i.\,e.~$\End(X) = \C \id_X$. Then there are constants $\xi$ and $\tilde\xi$ such that
$\eta_X = \xi \id_X$ and $\eta_{X^\vee} = \tilde\xi \id_{X^\vee}$, and the above equation implies $\xi \tilde\xi =1$. Then, denoting the linear maps \eqref{DlDr-def-general} for $\delta$ and $\delta'$ by $\mathcal D_{l/r}(X)$ and $\mathcal D_{l/r}'(X)$, respectively,
\begin{align}
  (\mathcal D_l'(X) \circ \mathcal D_r'(X))(\id)
  &=
  \mathcal D_r'(X^\vee \otimes X)(\id)
  =
  \xi \tilde\xi \mathcal D_r(X^\vee \otimes X)(\id)
\nonumber\\
  &=
  (\mathcal D_l(X) \circ \mathcal D_r(X))(\id) \, .
\end{align}
Thus, for absolutely simple $X$, $(\mathcal D_l(X)\circ \mathcal D_r(X))(\id)$ is independent of the pivotal structure. 

In the case of matrix factorisations, the relevant condition is that the space of degree preserving endomorphisms of $X$ is $\C \id_X$. In the example treated above, the $P_S$ have this property, and indeed from \eqref{X_xyz-on-phill0} and \eqref{P_ab-on-xl} one checks that
$(\mathcal{D}^\text{MF}_{l}( F(X_{[x,y,0]}) ) \circ \mathcal{D}^\text{MF}_{r}( F(X_{[x,y,0]}) ))(\id)$ gives the same multiple of $\id$ as
$(\mathcal{D}^\text{CFT}_{l}( X_{[x,y,0]} )\circ \mathcal{D}^\text{CFT}_{r}( X_{[x,y,0]} ))(\one)$ gives of $\one$.

\section{Discussion}\label{discuss}

In this paper we have studied dualities in the topological defect category $\MFW$ of Landau-Ginzburg models. More precisely, we have explicitly constructed the rigid and pivotal structure of $\MFW$ in the one-variable case, and then used it to compute the defect action on bulk fields. We also analysed the relation between the Grothendieck ring $\KMF(W))$ and R-charge preserving operators on the bulk algebra. For the case of many variables, we have suggested how to establish rigidity and  constructed a pivotal structure in general under the assumption of rigidity. Our results show that the CFT/LG correspondence can\textsl{not} straightforwardly be extended to the level of rigid and pivotal monoidal categories, yet still the comparison of quantities independent of the pivotal structures yields agreement for the action of defects on bulk fields. 

\medskip

Another way to think of dualities for defects between two Landau-Ginzburg models with the same potential is to embed them into a larger structure. Indeed, it is natural to organise all topological defects between all Landau-Ginzburg models into a bicategory $\mathcal{LG}$: its objects are ``theories'', i.\,e.~pairs $(R,W)$ of polynomial rings~$R$ and potentials $W\in R$ with an isolated singularity at the origin, 1-morphisms between $(R,W)$ and $(R',W')$ are matrix factorisations~$X$ of $W\tc 1 + 1\tc W'$, and 2-morphisms between~$X$ and~$Y$ are elements of $\Hom_{\operatorname{MF}(W\tc 1 - 1\tc W')}(X,Y)$. (Equivalently, one may also use the categories $\operatorname{MF}_{\text{bi}}(W,W')$ of~\cite{cr0909.4381} for 1- and 2-morphisms.) It has been established~\cite{Calinetal} that this bicategory can be naturally endowed with the structure of a monoidal framed bicategory. 

Using the results of~\cite{d0904.4713}, one can view $\mathcal{LG}$ as a subbicategory of the homotopy category of the bicategory $\mathcal{LG}_{\text{DG}}$ that has differential graded categories $\operatorname{DG}(W)$ as objects and the 1- and 2-morphisms are provided by the derived category of differential graded modules over $\operatorname{DG}(W)\otimes \operatorname{DG}(-W')$, see e.\,g.~\cite{k0601185} for the terminology. Then one may expect that $\mathcal{LG}$ is also a symmetric monoidal $(\infty,2)$-category. If this is the case, one can~\cite[sec.~5]{dm1004.0687} use the results of~\cite{d0904.4713}  to find that as an object in $\mathcal{LG}_{\text{DG}}$, the category $\operatorname{DG}(W)$ is fully dualisable in the sense of~\cite[def.~2.3.21]{l0905.0465}, and we expect $\mathcal{LG}=\mathcal{LG}^{\text{fd}}$ (see~\cite[sec.~2.3]{l0905.0465} for the notation). This would in particular imply that every defect~$X$ between Landau-Ginzburg models with potential~$W$ has itself a dual (called an adjoint in~\cite{l0905.0465}), and that defect fields $\ev_{X}$ and $\coev_{X}$ that satisfy the Zorro moves exist.

Let us expand on some of the structure of the bicategory $\mathcal{LG}$. While it remains to be rigorously answered whether it is a symmetric monoidal $(\infty,2)$-category, it is monoidal as a weak double category~\cite{Calinetal}. The unit object is simply $(\C,0)$, and the tensor product on objects is given by $(R,W)\otimes(R',W')=(R\tc R', W\tc 1 + 1\tc W')$. On the level of 1- and 2-morphisms, the tensor product is the external one (i.\,e.~as in~\eqref{XtensorY} and~\eqref{morphtensor} but with ``$\otimes_{R}$'' replaced by ``$\tc$'') while the composition of 1-morphisms is given by fusion. 

The dual of an object $(R,W)$ in $\mathcal{LG}$ is given by $(R,-W)$, and one may now ask for evaluation and coevaluation maps on this higher categorial level. By definition, these are 1-morphisms
\begin{align}
\ev_{(R,W)} & : (R,-W) \otimes (R,W) \longrightarrow (\C,0) \, , \nonumber \\
\coev_{(R,W)} & : (\C,0) \longrightarrow (R,W) \otimes (R,-W)
\end{align}
which are objects in the (1-)categories
\begin{align}
\operatorname{MF}((-W\tc 1+ 1\tc W)-1\tc 0) & \equiv\operatorname{MF}(0\tc 1 - (W\tc 1 - 1\tc W))\nonumber \\
& \equiv\operatorname{MF}(-W\tc 1+ 1\tc W)  \, .
\end{align}
If we denote by~$I_{W}$ the unit object in $\MFW\equiv \operatorname{MF}(W\tc 1- 1\tc W)$ and define $\ev_{(R,W)}=\coev_{(R,W)}=I_{-W}$, then one may verify that the Zorro moves for $\ev_{(R,W)}, \coev_{(R,W)}$ hold up to 2-isomorphism. Furthermore, we can define another duality structure by $\tev_{(R,W)}=\tcoev_{(R,W)}=TI_{W}$. With this one may consider the \textsl{quantum dimension of a Landau-Ginzburg model}: in analogy to the 1-categorial case we set
\be
\dim\big( (R,W) \big) = \tev_{(R,W)} \circ \coev_{(R,W)} \, .
\ee
Then we use the relation~\eqref{HomTensor} to find that $\dim((R,W))$ is given by the bulk algebra,
\be
\dim\big( (R,W) \big) \cong R/(\partial W) \, ,
\ee
which is isomorphic to the Hochschild cohomology of $\operatorname{DG}(W)$~\cite{d0904.4713}. An analogous result is also true for general B-twisted sigma models~\cite{Lurietalk}.

\subsubsection*{Acknowledgements}

We thank 
Nicolas Behr, 
Ilka Brunner, 
Orit Davidovich, 
Alexei Davydov,
Tobias Dyckerhoff, 
Michael Kay, 
Daniel Murfet and
Daniel Roggenkamp 
for discussions. Part of this work was completed during the Oberwolfach workshop ``Geometry, Quantum Fields, and Strings: Categorial Aspects'' and made possible through the support of N.\,C.'s~K\"ahler Fellowship at the University of Hamburg. The work of I.\,R.~is supported by the German Science Foundation (DFG) within the Collaborative Research Center 676 ``Particles, Strings and the Early Universe''.

\appendix

\section{Appendix}

\subsection{Explicit morphisms of the monoidal structure}\label{uglymatrices}

The tensor product of two morphisms $\varphi,\varphi'$ in $\MFW$ is given by
\be\label{morphtensor}
\varphi\otimes \varphi'= \begin{pmatrix}
\varphi_0\otimes_R \varphi'_0 &0&0&0 \\
0&\varphi_1\otimes_R \varphi'_1&0&0 \\
0&0&\varphi_1\otimes_R \varphi'_0&0 \\
0&0&0&\varphi_0\otimes_R \varphi'_1
\end{pmatrix} , 
\ee
and the explicit associator isomorphism $\alpha_{X,Y,Z}:(X\otimes Y)\otimes Z\rightarrow X\otimes(Y\otimes Z)$ reads
\begin{align}
(\alpha_{X,Y,Z})_0 & = \begin{pmatrix}
\id_{X_0\tr Y_0\tr Z_0} &0&0&0\\
0&0&0&\id_{X_0\tr Y_1\tr Z_1}\\
0&\id_{X_1\tr Y_1\tr Z_0}&0&0\\
0&0&\id_{X_1\tr Y_0\tr Z_1}&0
\end{pmatrix}, \\
(\alpha_{X,Y,Z})_1 & = \begin{pmatrix}
\id_{X_1\tr Y_0\tr Z_0} &0&0&0\\
0&0&0&\id_{X_1\tr Y_1\tr Z_1}\\
0&\id_{X_0\tr Y_1\tr Z_0}&0&0\\
0&0&\id_{X_0\tr Y_0\tr Z_1}&0
\end{pmatrix}.
\end{align}

In the case of a potential~$W$ in only one variable, the homotopy inverses of $\lambda_{X}, \rho_{X}$ in~\eqref{lambdarho} are given by
\begin{align}
  \lambda_X^{-1} & = 
  \begin{pmatrix}
    [ 1 \tc  \id_{\check X_0} ]\hat~ & 0 \\
    [ \tfrac{1 \tc  \check d^X_0(a,b)- 1 \tc \check d^X_0(x,b)}{a-x}]\hat~ & 0 \\
    0 & [ \tfrac{1 \tc  \check d^X_1(a,b)-1 \tc \check d^X_1(x,b)}{a-x}]\hat~  \\
    0 & [ 1 \tc  \id_{\check X_1} ]\hat~  
  \end{pmatrix}
: X \longrightarrow I\otimes X\, ,\\
  \rho_X^{-1} & = 
  \begin{pmatrix}
    [ \id_{\check X_0} \tc  1 ]\hat~ & 0 \\
    [ \tfrac{\check d^X_0(a,x)\tc 1 - \check d^X_0(a,b)\tc 1}{x-b}]\hat~ & 0 \\
    0 & [ \id_{\check X_1} \tc 1 ]\hat~ \\
    0 & [ \tfrac{\check d^X_1(a,b)\tc 1-\check d^X_1(a,x)\tc 1}{x-b}]\hat~
  \end{pmatrix}
: X \longrightarrow X\otimes I
\end{align}
where we employ a natural generalisation of the hat-notation introduced in section~\ref{prelim}, see~\cite[app~A.1]{cr0909.4381} for details.

\subsection{The evaluation map is a morphism}\label{evmorph}

To show that $\ev_{X}:X^\vee\otimes X\to I$ is well-defined in $\MFW$ we have to check that $I\circ \ev_{X}=\ev_{X}\circ(X^\vee\otimes X)$. Writing this out in components, the condition becomes
\begin{align}
\iota_{0}\circ A_{X} & = B_{X}\circ ((d^X_{0})^\vee\otimes_{R}\id_{X_{0}}) + C_{X}\circ (\id_{X^\vee_{0}}\otimes_{R} d_{0}^X) \, , \label{co11} \\
0 & = B_{X} \circ (\id_{X^\vee_{1}}\otimes_{R} d^X_{1}) + C_{X}\circ (d_{1}^\vee\otimes_{R} \id_{X_{1}}) \, , \label{co12} \\
\iota_{1} \circ B_{X} & = -A_{X}\circ ((d^X_{1})^\vee \otimes_{R} \id_{X_{1}}) \, , \quad \iota_{1} \circ C_{X} = A_{X} \circ (\id_{X^\vee_{1}}\otimes_{R}d^X_{1}) \, . \label{co2}
\end{align}
We first show that~\eqref{co11} and~\eqref{co12} are satisfied if~\eqref{co2} holds. Since $\iota_{1}=[a-b]\hat~$ is injective, \eqref{co11} is true if $\iota_{1}\circ\iota_{0}\circ A_{X} = \iota_{1}\circ B_{X}\circ ((d^X_{0})^\vee\otimes_{R}\id_{X_{0}}) + \iota_{1}\circ C_{X}(\id_{X^\vee_{0}}\otimes_{R} d_{0}^X)$ which is equivalent to
\begin{align}
[W(a)-W(b)]\hat~ \circ A_{X} & = (-A_{X}\circ ((d^X_{1})^\vee \otimes_{R} \id_{X_{1}}))\circ ((d^X_{0})^\vee\otimes_{R}\id_{X_{0}}) \nonumber \\
& \qquad + (A_{X} \circ (\id_{X^\vee_{1}}\otimes_{R}d^X_{1}))\circ (\id_{X^\vee_{0}}\otimes_{R} d_{0}^X)  \nonumber \\
&= -A_{X} \circ [W(x)-W(a)]\hat~ + A_{X} \circ [W(x)-W(b)]\hat~ \nonumber \\
&= [W(a)-W(b)]\hat~ \circ A_{X} \, .
\end{align}
The identity~\eqref{co12} is checked similarly. Thus it remains to show that~\eqref{co2} holds for $A_{X}, B_{X}, C_{X}$ given by~\eqref{AX}--\eqref{CX}. Let us write $X_{1}=R\tc\check X_{1}\tc R$ where $\check X_{1}$ is a vector space with basis $\{ e_{i} \}$. To see that the second equation in~\eqref{co2} is true it is sufficient to prove that this is so when both sides are applied to elements of the form $1\tc e_{i}^*\tc x^k \tc e_{j}\tc 1$. But we have
\be
(\iota_{1}\circ C_{X})(1\tc e_{i}^*\tc x^k \tc e_{j}\tc 1) = -\delta_{i,j}\delta_{k,0}[a-b]\hat~(1\tc 1\tc 1)
\ee
and
\begin{align}
&(A_{X} \circ (\id_{X^\vee_{1}}\otimes_{R}d^X_{1}) )(1\tc e_{i}^*\tc x^k \tc e_{j}\tc 1) \nonumber \\
=& -\left[ \ev_{\check X_{1}} \left(e^*_{i} \tc \oint \frac{(a-b-x) x^k \D x}{x(W(x)-W(b))} \check d^X_{0}(x,b) \check d^X_{1}(x,b) (e_{j}) \right) \right]^{\!\wedge} \nonumber \\
=& -\left[ \ev_{\check X_{1}} \left( \oint (a-b-x) x^{k-1} \D x  \, e^*_{i} \tc e_{j} \right) \right]^{\!\wedge} \nonumber \\
=& -\delta_{i,j}\delta_{k,0}[a-b]\hat~(1\tc 1\tc 1) \, ,
\end{align}
and the first equation in~\eqref{co2} follows analogously.

\subsection{Zorro moves}\label{appZorro}

We want to show that the Zorro move
\be\label{Z11}
\rho_{X} \circ (\id_{X} \otimes \ev_{X})\circ \alpha_{X,X^\vee,X} \circ (\coev_X \otimes \id_{X}) \circ \lambda^{-1}_{X}= \id_{X}
\ee
holds true for all $X\in\MFW$ whose entries have polynomial degrees lower than $\text{deg}(W)$. By straightforward matrix multiplication we find that the left-hand side is a $(2\times 2)$-matrix whose $(1,1)$-entry is given by
\begin{align}\label{T}
F & = (\id_{X_{0}} \otimes_{R} \mu)\circ (\id_{X_{0}}\otimes_R A_{X}) \nonumber \\
& \quad \circ \left\{ \left( \left[\tfrac{(\check d^X_{f})_{1}(a,x) \tc \id_R \tc \id_{\check X_1^*}  - (\check d^X_{f})_{1}(b,x)\tc \id_R \tc \id_{\check X_1^*}}{a-b}\right]^{\!\wedge} \circ c_{X_{1}} \right) \otimes_{R} \id_{X_{0}} \right\} 
 \nonumber \\
& \quad  \circ [1\tc \id_{\check X_{0}}]\hat~ \, .
\end{align}
Since the left-hand side of \eqref{Z11} is a morphism in $\MFW$ it suffices to prove that
\be
(\id_{R}\tc (e_{r}^0)^*\tc \id_{R})(F(1\tc e_{s}^0\tc 1)) = \delta_{r,s}\, 1\tc 1
\ee
in order to check that~\eqref{Z11} is true. Here and below we denote by $\{ e_{r}^i \}$ a basis of the vector space~$\check X_{i}$. 

Substituting the expression~\eqref{AX} into~\eqref{T} we find that
\begin{align}\label{puh}
& (\id_{R}\tc (e_{r}^0)^*\tc \id_{R})(T(1\tc e_{s}^0\tc 1)) \nonumber \\
=& \sum_{l=1}^{\dim \check X_{0}} \left[ \oint \frac{(2\pi\I)^{-1}\D x}{W(x)-W(b)} \frac{(e^0_{r})^*((\check d^X_{1}(a,b) - \check d^X_{1}(x,b))(e^1_{l}))}{a-x} \, (e^1_{l})^* (\check d^X_{0}(x,b)(e^0_{s})) \right]^{\!\wedge} \nonumber \\
=& \left[ \oint \frac{(2\pi\I)^{-1}\D x}{W(x)-W(b)} \frac{(e^0_{r})^*(\check d^X_{1}(a,b) \check d^X_{0}(x,b)(e^0_{s}))}{a-x} \right]^{\!\wedge} 
- \left[ \oint \frac{\D x}{2\pi\I} \frac{\delta_{r,s}}{a-x} \right]^{\!\wedge} \, .
\end{align}
where we used the matrix bi-factorisation condition $\check d^X_{1}(x,b) \check d^X_{0}(x,b)=(W(x)-W(b))\id_{\check X_{0}}$. This is indeed equal to $\delta_{r,s}\, 1\tc 1$ as there are no entries of degree $\text{deg}(W)$ or higher in $\check d_{0}^X$. 
 
The other Zorro move~\eqref{lambdaZorro} is proved analogously.

\subsection{R-charge and duals}\label{R-charge-app}

In this appendix we formulate duals for graded matrix bi-factorisations and show that $\ev_X$ and $\coev_X$ have R-charge zero. 

\subsubsection{Group action on bimodules}

Let $R$ and $S$ be $\C$-algebras. Given $\mu\in\operatorname{Aut}(R)$ and $\nu\in\operatorname{Aut}(S)$ we obtain a functor $\Ga{\mu,\nu}$ from $R$-mod-$S$ to itself by twisting the action of $R$ and $S$,
\be
  (X , \rho^l, \rho^r) \longmapsto (X, \rho^l \circ (\mu \otimes \id_X) , \rho^r \circ( \id_X \otimes \nu))
  \, , \quad
  f \longmapsto f \, .
\ee
$\Gamma$ defines a strict action of $\operatorname{Aut}(R)^\text{op} \times \operatorname{Aut}(S)^\text{op}$ on $R$-mod-$S$, i.\,e.~$\Ga{\mu',\nu'} \circ \Ga{\mu,\nu} = \Ga{\mu\mu',\nu\nu'}$. The group action commutes with taking duals in the sense that there is a natural isomorphism
\be\label{group-dual-natiso}
  \chi^{\mu,\nu} : (\,\cdot\,)^\vee \circ \Ga{\mu,\nu} \Longrightarrow \Ga{\nu,\mu} \circ (\,\cdot\,)^\vee \, ,
\ee
which takes an element $\psi \in X^\vee = \Hom_{R\text{-mod-}S}(X,R \otimes_\C S)$ to $\chi_X^{\mu,\nu}(\psi) = (\mu \otimes_\C \nu) \circ \psi$ (this defines a map $\chi_X^{\mu,\nu} : (\Ga{\mu,\nu}(X))^\vee \rightarrow \Ga{\nu,\mu}(X^\vee)$ of $S$-$R$-bimodules) and which satisfies
\begin{align}
  &\big( 
  \xymatrix{%
  (\Ga{\mu\mu',\nu\nu'}X)^\vee
  \ar[rr]^-{\chi^{\mu',\nu'}_{\Ga{\mu,\nu}X}} &&
  \Ga{\nu',\mu'}( (\Ga{\mu,\nu} X)^\vee )
  \ar[rr]^-{\Ga{\nu',\mu'} \chi^{\mu,\nu}_X} &&
  \Ga{\nu\nu',\mu\mu'}(X^\vee) 
  }
  \big)
\nonumber  \\
  &=
  \big( 
    \xymatrix{%
  (\Ga{\mu\mu',\nu\nu'}X)^\vee
  \ar[rr]^-{\chi^{\mu\mu',\nu\nu'}_X} &&
  \Ga{\nu\nu',\mu\mu'}(X^\vee) 
  }
  \big) \, .
\label{chi-composition-property}
\end{align}

We now specialise to the case that $R = \C[x_1,\dots,x_M]$ and $S = \C[y_1,\dots,y_N]$ and consider group homomorphisms $\sigma_R : \C \rightarrow R$ and $\sigma_S : \C \rightarrow S$ given by 
\be
  \sigma_R(\alpha)(x_i) = \E^{\I q_x \alpha} x_i \, , \quad 
  \sigma_S(\alpha)(y_j) = \E^{\I q_y \alpha} y_j \, ,
\ee
where $q_x,q_y \in \C$ are constants. Denote by $\Ga\alpha$ the diagonal action $\alpha \mapsto \Ga{\sigma_R(\alpha),\sigma_R(\alpha)}$ of $(\C,+)$ on $R$-mod-$S$. We denote the natural isomorphism~\eqref{group-dual-natiso} as $\chi^\alpha : (\,\cdot\,)^\vee \circ \Ga{\alpha} \Rightarrow \Ga{\alpha} \circ (\,\cdot\,)^\vee$. 

\begin{definition}
An $R$-$S$-\textsl{bimodule with} $\mathfrak u(1)$-\textsl{action} is a pair $(X,\varphi^X)$ where $X$ is an $R$-$S$-bimodule and $\varphi^X_\alpha : X \rightarrow \Ga\alpha(X)$ is a family of isomorphisms such that
\be \label{u1-action-composition-varphi}
  \big( 
  \xymatrix{%
  X \ar[r]^-{\varphi^X_{\alpha+\beta}} & \Ga{\alpha+\beta}X 
  }
  \big)
  =
  \big( 
  \xymatrix{%
  X \ar[r]^-{\varphi^X_{\alpha}} & \Ga{\alpha}X \ar[r]^-{\Ga\alpha(\varphi^X_{\beta})} & \Ga{\alpha+\beta}X 
  }
  \big) \, .
\ee
\end{definition}

In other words, $(X,\varphi^X)$ is a $\C$-invariant object in the category with $\C$-action $R$-mod-$S$. We say a bimodule map $f : X \rightarrow Y$ has \textsl{R-charge} $p$ iff the diagram
\be\label{bimod-R-charge-def}
\xymatrix{%
\Ga\alpha X \ar[rr]^-{\E^{\I p \alpha} \Ga\alpha(f)}  && \Ga\alpha Y  \\
X\ar[u]_{\varphi^X_\alpha} \ar[rr]^-{f} && Y \ar[u]^{\varphi^Y_\alpha}
}%
\ee
commutes for all $\alpha \in \C$. Given a bimodule with $\mathfrak u(1)$-charge $(X,\varphi^X)$, we define its dual as $(X,\varphi)^\vee = (X^\vee, \tilde\varphi)$ with
\be \label{dual-u1-action-def}
  \big( 
  \xymatrix{%
  X^\vee \ar[r]^-{\tilde\varphi_\alpha} & \Ga\alpha(X^\vee) 
  }
  \big)
  =
  \big( 
    \xymatrix{%
    X^\vee \ar[r]^-{(\varphi_\alpha^{-1})^\vee} & (\Ga\alpha X)^\vee
  \ar[r]^-{\chi^\alpha_X} & \Ga\alpha(X^\vee) 
  }
   \big) \, .
\ee
We need to verify the composition rule \eqref{u1-action-composition-varphi}:
\begin{align}
\Ga\alpha(\tilde\varphi_\beta) \circ \tilde\varphi_\alpha 
&\overset{(1)}{=} 
\Ga\alpha(\chi_X^\beta) \circ \Ga\alpha(\varphi_\beta^{-1\,\vee}) \circ \chi^\alpha_X \circ \varphi_\alpha^{-1\,\vee} \nonumber\\
&\overset{(2)}{=} 
\Ga\alpha(\chi_X^\beta) \circ \chi^\alpha_{\Ga\beta X} \circ (\Ga\alpha(\varphi_\beta^{-1}))^\vee \circ  \varphi_\alpha^{-1\,\vee}\nonumber\\
&\overset{(3)}{=} 
\chi^{\alpha+\beta}_X \circ \big( \varphi^{-1}_\alpha \circ \Ga\alpha(\varphi_\beta)^{-1} \big)^\vee = \tilde\varphi_{\alpha+\beta} \, ,
\end{align}
where step (1) is the definition of $\tilde\varphi$, step (2) is naturality of $\chi^\alpha$, and step (3) is \eqref{chi-composition-property}.

\subsubsection[$\mathfrak u(1)$-action and duals for the bimodules $R \otimes_\C W \otimes_\C S$]{$\boldsymbol{\mathfrak u(1)}$-action and duals for the bimodules $\boldsymbol{R \otimes_\C W \otimes_\C S}$}\label{u1-action-duals-bim}

Given a $\C$-vector space $W$, we obtain a free $R$-$S$-bimodule $R \otimes_\C W \otimes_\C S$. For such bimodules we can give a more direct formulation of the $\mathfrak u(1)$-action and their duals. Define the bimodule map
\begin{align}
  s_\alpha^W : R \otimes_\C W \otimes_\C S & \longrightarrow \Ga\alpha(R \otimes_\C W \otimes_\C S) \\
  r \otimes_\C w \otimes_\C s & \longmapsto  \sigma_R(\alpha)(r) \otimes_\C w \otimes_\C \sigma_S(\alpha)(s) \, .
\end{align}
Let $a \equiv (a_1,\dots,a_M)$ and $b \equiv (b_1,\dots,b_N)$ be formal variables.
For a map $f(a,b) : V \rightarrow W[a,b]$ we obtain the commuting diagram
\be\label{s-alpha-defining-diag}
\xymatrix{%
\Ga\alpha(R \otimes_\C V \otimes_\C S) \ar[rrr]^-{\Ga\alpha(\,[f(a,b)]\hat~\,)}  &&& \Ga\alpha(R \otimes_\C W \otimes_\C S)  \\
R \otimes_\C V \otimes_\C S \ar[u]^{s^V_\alpha} \ar[rrr]^-{[f(\sigma_R(-\alpha)(a),\sigma_S(-\alpha)(b)]\hat~} &&& R \otimes_\C W \otimes_\C S \ar[u]_{s^W_\alpha} \, .
}%
\ee
Given a $\mathfrak u(1)$-action $\varphi_\alpha^W$ on $R \otimes_\C W \otimes_\C S$, we define the bimodule map
\be \label{u-alpha-def-via-phi}
  U^W(\alpha) = (\varphi^W_{-\alpha})^{-1} \circ s_{-\alpha}^W : 
  R \otimes_\C W \otimes_\C S \longrightarrow R \otimes_\C W \otimes_\C S \, ,
\ee
i.\,e.~$\varphi^W_{\alpha} = s_\alpha^W \circ U^W(-\alpha)^{-1}$.
With this choice of signs, comparing \eqref{bimod-R-charge-def} and \eqref{s-alpha-defining-diag} shows that $[f(a,b)]\hat~$ has R-charge $p$ iff
\be
  U^W(\alpha) \circ [f(\sigma_R(\alpha)(a),\sigma_S(\alpha)(b)]\hat~ \circ U^V(\alpha)^{-1} = \E^{\I p \alpha} \, [f(a,b)]\hat~ \, .
\ee
which is the standard R-charge condition, see e.\,g.~\cite{hw0404196}. 

For bimodules of the form $R \otimes_\C W \otimes_\C S$ we have a natural contravariant functor $(\,\cdot\,)^+$, given by
\be
  (R \otimes_\C W \otimes_\C S)^+ = S \otimes_\C W^* \otimes_\C R 
  \, , \quad
  ([f(a,b)]\hat~)^+ = [f^*(b,a)]\hat~
\ee
where for $f(a_1,\dots,a_M,b_1,\dots,b_N) = 
\sum_{k_1,\dots,k_M,l_1,\dots,l_N} f_{k_1,\dots,l_N} a_1^{k_1} \cdots a_M^{k_M} b_1^{l_1} \cdots b_{N}^{l_N}$ we set
\be
  f^*(b,a)
  = 
  \sum_{k_1,\dots,k_M,l_1,\dots,l_N} f_{k_1,\dots,l_N} 
  b_1^{k_1} \cdots b_M^{k_M} a_1^{l_1} \cdots a_{N}^{l_N} \, .
\ee
Note that $f(a,b) : V \rightarrow W[a_1,\dots,a_M,b_1,\dots,b_N]$ while $f^*(b,a) : W^* \rightarrow V^*[a_1,\dots,a_N,b_1,\dots,b_M]$, as by convention the $a_i$ act on the left algebra, which is 
$R$ for $R \otimes_\C W \otimes_\C S$ while it is $S$ for $S \otimes_\C W^* \otimes_\C R$. The number of formal $b$-variables changes for the same reason.

We can define a natural isomorphism $\kappa: (\,\cdot\,)^+ \Rightarrow (\,\cdot\,)^\vee$ via
\begin{align}
  \kappa_W : (R \otimes_\C W \otimes_\C S)^+ &\longrightarrow  (R \otimes_\C W \otimes_\C S)^\vee \, , \\
   s \otimes_\C \phi \otimes_\C r & \longmapsto  \big( e \otimes_\C w \otimes_\C f \mapsto  \phi(w) \cdot (re) \otimes_\C (fs) \big) \, .  
\end{align}  
Indeed one checks that $\kappa_W$ provides a natural family of $S$-$R$-bimodule isomorphisms. In addition, it makes the following diagram commute (we omit the $\otimes_\C$):
\be\label{kappa-s-interrel}
\xymatrix{%
(R  \,W  S)^\vee
  \ar[rr]^-{((s_\alpha^{W})^{-1})^\vee} &&
\big(\Ga\alpha(R  \,W  S)\big)^\vee
  \ar[rr]^-{\chi^\alpha_{R  \,W  S}} &&
\Ga\alpha((R  \,W  S)^\vee)
\\
(R \, W  S)^+
  \ar[rrrr]^-{s_\alpha^{W^*}}
  \ar[u]^{\kappa_W} &&&&
\Ga\alpha((R \, W  S)^+)
  \ar[u]_{\Ga\alpha(\kappa_W)} \, .
}%
\ee
In the main text the natural isomorphism $\kappa$ is used implicitly, but for the purpose of this appendix we find it clearer to distinguish the two duals.

Given a bimodule $R \otimes_\C W \otimes_\C S$ with $\mathfrak u(1)$-action described by $U^W(\alpha)$, we assign to $(R \otimes_\C W \otimes_\C S)^+$ the $\mathfrak u(1)$-action $(U^W(\alpha)^{-1})^+$. This is the unique choice compatible with \eqref{dual-u1-action-def} in the sense that the diagram
\be \label{+-dual-relation-to-^-dual}
\xymatrix{%
(R  \,W  S)^\vee
  \ar[rrr]^-{\tilde\varphi^{R\,WS}_\alpha} &&&
\Ga\alpha((R  \,W  S)^\vee)
\\
(R \, W  S)^+
  \ar[rr]^-{(U^W(-\alpha))^+}
  \ar[u]^{\kappa_W} &&
(R \, W  S)^+
  \ar[r]^-{s_\alpha^{W^*}} &
\Ga\alpha((R \, W  S)^+)
  \ar[u]_{\Ga\alpha(\kappa_W)}
}%
\ee
commutes (that the map $(U^W(-\alpha))^+$ appears instead of $(U^W(\alpha)^{-1})^+$ is due to definition \eqref{u-alpha-def-via-phi}). This follows when inserting definitions \eqref{dual-u1-action-def} and \eqref{u-alpha-def-via-phi} and using commutativity of \eqref{kappa-s-interrel}.

\subsubsection{Graded matrix bi-factorisations}

\begin{definition}
A \textsl{graded matrix bi-factorisation} is a pair $(X,\varphi^X)$ where $X = (X_{0},X_{1},d_{0}^X,d_{1}^X)$ is a matrix bi-factorisation and $X_0 \oplus X_1$ is a bimodule with $\mathfrak u(1)$-action $\varphi^X(\alpha) = (\begin{smallmatrix} \varphi^{X_0}(\alpha) & 0  \\ 0 & \varphi^{X_1}(\alpha) \end{smallmatrix})$ such that $d^X = (\begin{smallmatrix} 0 & d_1^X \\ d_0^X & 0 \end{smallmatrix})$ has R-charge 1.
\end{definition}

We now restrict ourselves to the one-variable case $R=S=\C[x]$ with potential $W(x)=x^d$. In this case the constant $q_x$ is given by $2/d$. If $(X , \varphi^X)$ is a graded matrix bi-factorisation, we define its dual graded matrix bi-factorisation to be $(X^\vee, \varphi^{(X^\vee)})$, where we take
\be
  \varphi^{(X^\vee)} = \E^{\I \alpha (q_x-1)} 
  \begin{pmatrix} \tilde\varphi^{X_1}(\alpha) & 0  \\ 0 & \tilde\varphi^{X_0}(\alpha) \end{pmatrix} 
\ee
and $\tilde\varphi^X$ was defined in \eqref{dual-u1-action-def}. The reason to include the phase shift is to ensure that $I^\vee \cong I$ via an isomorphism of R-charge zero; we will come to this in a moment. Independent of the phase shift one checks that if $d^X$ has R-charge 1 with respect to the $\mathfrak u(1)$-action $\varphi^X$, then $d^{(X^\vee)}$ has R-charge 1 with respect to the $\mathfrak u(1)$-action $\varphi^{(X^\vee)}$. 

If $X_i = R \otimes_\C \check X_i \otimes_\C R$ and the $\mathfrak u(1)$-action is described by $U^{X_i}(\alpha)$, then we define the matrix bi-factorisation $X^+ = (X_1^+, X_0^+, (d_0^X)^+, -(d_1^X)^+)$ with $\mathfrak u(1)$-action described by
\be \label{+-dual-of-U}
  U^{(X^+)}(\alpha) = \E^{\I \alpha (q_x-1)}  \begin{pmatrix} (U^{X_1}(\alpha)^{-1})^+ & 0 \\ 0 & (U^{X_0}(\alpha)^{-1})^+ \end{pmatrix} \, .
\ee
This is isomorphic to $(X^\vee, \varphi^{(X^\vee)})$ via the isomorphism $\kappa_{X_1\oplus X_0}$, which one can verify to have R-charge zero. Recall the definition of the graded matrix bi-factorisation $I$ in section~\ref{R-charge-short}; plugging $U^I$ into \eqref{+-dual-of-U} we see that $I = I^+ \cong I^\vee$ as graded matrix bi-factorisations, and the isomorphism is of R-charge zero.

Let $X$ be a matrix bi-factorisation with $X_i = R \otimes_\C \check X_i \otimes_\C R$. Note that the maps $\ev_X$ and $\coev_X$ given in section~\ref{left-dual-MFbi} are actually maps $X^+ \otimes X \rightarrow I$ and $I \rightarrow X \otimes X^+$, respectively. Similarly, the identity verified in lemma~\ref{phihop} is actually that for $f : X \rightarrow Y$ we have $\coev_X \circ (f \otimes \id_{X^+}) = \coev_Y \circ (\id_Y \otimes f^+)$. Analogous statements hold for $\ev_X$. The map $\coev_X$ satisfies
\be
\xymatrix{%
I  \ar[rr]^-{\Ga\alpha(\coev_X)} && \Ga\alpha(X) \otimes \Ga\alpha(X^+)  \\
I \ar[u]^{s_\alpha \circ U^I(-\alpha)^{-1}} \ar[rr]_-{~\E^{\I \alpha (1-q_x)} \coev_X} 
&& X \otimes X^+ \ar[u]_{s_\alpha^W \otimes s_\alpha^{W^*}} .
}%
\ee
Combining this observation with definition \eqref{+-dual-of-U} and the fact that the $\mathfrak u(1)$-action on $X^+$ is given by the bottom line of \eqref{+-dual-relation-to-^-dual}, it is straightforward to check that $\coev_X$ has R-charge zero: naturality of $\kappa$ implies compatibility with the differential and commutativity of \eqref{+-dual-relation-to-^-dual} gives the R-charge to be zero. 
The argument for $\ev_X$ is analogous.

\subsection{Proof of lemma \ref{isos-nat-monoid}}\label{left-from-right-lemma}

To show part (i) of lemma \ref{isos-nat-monoid} we use the relation~\eqref{dualphi} twice to find that the naturality condition $\varphi^{\vee\vee}\circ t_{X}=t_{Y}\circ\varphi$ is equivalent to
\be\label{Rolle}
\begin{tikzpicture}[very thick,scale=1.0,color=blue!50!black, baseline]
\draw[line width=1pt] 
(-1,-0.75) node[inner sep=2pt,draw] (dx) {{\small$t_{X}$}}
(1,0) node[inner sep=2pt,draw] (f) {{\small$\varphi$}}; 
\draw[line width=1pt] 
(-1,-1.95) node[line width=0pt] (X) {{\small $X$}}
(3,1.95) node[line width=0pt] (Y) {{\small $Y^{\vee\vee}$}};
\draw[oodirected] (f) .. controls +(0,1) and +(0,1) .. (0,0);
\draw[uudirected] (2,0) .. controls +(0,-1) and +(0,-1) .. (f);
\draw[directed] (3,0) .. controls +(0,-1.5) and +(0,-1.5) .. (0,0);
\draw[directed] (2,0) .. controls +(0,1.5) and +(0,1.5) .. (-1,0);
\draw (-1,0) -- (dx);
\draw (X) -- (dx);
\draw (Y) -- (3,0);
\end{tikzpicture}
\,=
\begin{tikzpicture}[very thick,scale=1.0,color=blue!50!black, baseline]
\draw[line width=1pt] 
(-1,-0.75) node[inner sep=2pt,draw] (f) {{\small$\varphi\vphantom{t_{Y}}$}}
(-1,0.75) node[inner sep=2pt,draw] (dy) {{\small$t_{Y}$}}; 
\draw[line width=1pt] 
(-1,-1.95) node[line width=0pt] (X) {{\small $X$}}
(-1,1.95) node[line width=0pt] (Y) {{\small $Y^{\vee\vee}$}};
\draw (f) -- (dy);
\draw (f) -- (X);
\draw (dy) -- (Y);
\end{tikzpicture}
\quad\Leftrightarrow\quad
\begin{tikzpicture}[very thick,scale=1.0,color=blue!50!black, baseline]
\draw[line width=1pt] 
(0,0) node[inner sep=2pt,draw] (f) {{\small$\varphi$}}; 
\draw[oodirected] (f) .. controls +(0,1) and +(0,1) .. (-1,0);
\draw[uudirected] (1,0) .. controls +(0,-1) and +(0,-1) .. (f);
\draw[line width=1pt] 
(-1,-1.95) node[line width=0pt] (X) {{\small $Y^\vee$}}
(1,1.95) node[line width=0pt] (Y) {{\small $X^\vee$}};
\draw (-1,0) -- (X);
\draw (1,0) -- (Y);
\end{tikzpicture}
\,=\,
\begin{tikzpicture}[very thick,scale=1.0,color=blue!50!black, baseline]
\draw[line width=1pt] 
(-1,-0.8) node[inner sep=2pt,draw] (dx) {{\small$t_{X}^{-1}$}}
(-1,0) node[inner sep=2pt,draw] (f) {{\small$\varphi\vphantom{t_{Y}}$}}
(-1,0.75) node[inner sep=2pt,draw] (dy) {{\small$t_{Y}$}}; 
\draw[uudirected] (0,0.75) .. controls +(0,1) and +(0,1) .. (dy);
\draw[oodirected] (dx) .. controls +(0,-1) and +(0,-1) .. (-2,-0.8);
\draw[line width=1pt] 
(0,-1.95) node[line width=0pt] (X) {{\small $Y^\vee$}}
(-2,1.95) node[line width=0pt] (Y) {{\small $X^\vee$}};
\draw (f) -- (dy);
\draw (f) -- (dx);
\draw (-2,-0.8) -- (Y);
\draw (0,0.75) -- (X);
\end{tikzpicture} \, .
\ee
In the second step we have composed both sides with $t_{X}^{-1}$ ``from below'' and applied two Zorro moves. But the last expression in~\eqref{Rolle} is precisely the left-hand side of~\eqref{sovereign} by definition of $\tcoev_{X}$ and $\tev_{Y}$. 

To prove part (ii) let us write out~\eqref{vv-monoidal} in pictorial language. Using~\eqref{gamma} and~\eqref{dualphi} we find
\be
((\nu^2_{Y,X})^\vee)^{-1} = 
\begin{tikzpicture}[very thick,scale=0.8,color=blue!50!black, baseline]
\draw[line width=1pt] 
(-3,-3) node[inner sep=0pt] (YX) {{\small$(Y^\vee\otimes X^\vee)^\vee$}}
(4.5,3) node[inner sep=0pt] (XY) {{\small$(X\otimes Y)^{\vee\vee}$}}; 
\draw[directed] (2,0) .. controls +(0,-1) and +(0,-1) .. (1,0);
\draw[directed] (3,0) .. controls +(0,-2) and +(0,-2) .. (0,0);
\draw[directed] (2.5,0) .. controls +(0,3) and +(0,3) .. (-3,0);
\draw[directed] (0.5,0) .. controls +(0,1.5) and +(0,1.5) .. (-1.5,0);
\draw[directed] (4.5,0) .. controls +(0,-3) and +(0,-3) .. (-1.5,0);
\draw (-3,0) -- (YX);
\draw (4.5,0) -- (XY);
\draw[dotted] (0,0) -- (1,0);
\draw[dotted] (2,0) -- (3,0);
\end{tikzpicture} \, .
\ee
That this is indeed the inverse of $(\nu^2_{Y,X})^\vee$ can be shown by concatenating the above expression with $(\nu^2_{Y,X})^\vee$ and using repeated Zorro moves to obtain the identity. Thus $t_{X\otimes Y}^{-1} \circ ((\nu^2_{Y,X})^\vee)^{-1} \circ \nu^2_{X^\vee,Y^\vee}\circ (t_{X}\otimes t_{Y}) = \id_{X\otimes Y}$ is equivalent to
\be
\begin{tikzpicture}[very thick,scale=0.8,color=blue!50!black, baseline]
\draw[line width=1pt] 
(10.5,2) node[inner sep=2pt,draw] (XY) {{\small$t_{X\otimes Y}^{-1}$}}
(1,-2) node[inner sep=2pt,draw] (Y) {{\small$t_{Y}$}}
(0,-2) node[inner sep=2pt,draw] (X) {{\small$t_{X}$}}; 
\draw[directed] (2,0) .. controls +(0,1) and +(0,1) .. (1,0);
\draw[directed] (3,0) .. controls +(0,2) and +(0,2) .. (0,0);
\draw[directed] (4.5,0) .. controls +(0,-1.5) and +(0,-1.5) .. (2.5,0);
\draw[directed] (9,0) .. controls +(0,3) and +(0,3) .. (4.5,0);
\draw[directed] (7,0) .. controls +(0,1.5) and +(0,1.5) .. (5.5,0);
\draw[directed] (8.5,0) .. controls +(0,-1) and +(0,-1) .. (7.5,0);
\draw[directed] (9.5,0) .. controls +(0,-2) and +(0,-2) .. (6.5,0);
\draw[directed] (10.5,0) .. controls +(0,-3) and +(0,-3) .. (5.5,0);
\draw[dotted] (2,0) -- (3,0);
\draw[dotted] (6.5,0) -- (7.5,0);
\draw[dotted] (8.5,0) -- (9.5,0);
\draw[line width=1pt] 
(0,-3.2) node[line width=0pt] (XX) {{\small $X$}}
(1,-3.2) node[line width=0pt] (YY) {{\small $Y$}};
\draw[line width=1pt] 
(10.5,3.2) node[line width=0pt] (XXYY) {{\small $X\otimes Y$}};
\draw (10.5,0) -- (XY);
\draw (0,0) -- (X);
\draw (1,0) -- (Y);
\draw (XXYY) -- (XY);
\draw (XX) -- (X);
\draw (YY) -- (Y);
\end{tikzpicture}
=\;
\begin{tikzpicture}[very thick,scale=0.8,color=blue!50!black, baseline]
\draw[line width=1pt] 
(0,-3.2) node[line width=0pt] (XX) {{\small $X\otimes Y$}}
(0,3.2) node[line width=0pt] (YY) {{\small $X\otimes Y$}};
\draw (XX) -- (YY);
\end{tikzpicture} \, . 
\ee

Now we apply one Zorro move to the left-hand side, compose with $t_{X\otimes Y}$ ``from above'' and with $(t_{X\otimes Y})^{-1}$ ``from below'', and append curved lines to the left and right to obtain
\be
\begin{tikzpicture}[very thick,scale=0.8,color=blue!50!black, baseline]
\draw[line width=1pt] 
(0.5,-2) node[inner sep=2pt,draw] (XY) {{\small$t_{X\otimes Y}^{-1}$}}
(1,0) node[inner sep=2pt,draw] (Y) {{\small$t_{Y}$}}
(0,0) node[inner sep=2pt,draw] (X) {{\small$t_{X}$}}; 
\draw[directed] (7,0) .. controls +(0,4) and +(0,4) .. (X);
\draw[directed] (6,0) .. controls +(0,3) and +(0,3) .. (Y);
\draw[directed] (4.5,0) .. controls +(0,1.5) and +(0,1.5) .. (3,0);
\draw[directed] (6,0) .. controls +(0,-1) and +(0,-1) .. (5,0);
\draw[directed] (7,0) .. controls +(0,-2) and +(0,-2) .. (4,0);
\draw[directed] (8,0) .. controls +(0,-3) and +(0,-3) .. (3,0);
\draw[directed] (9,2) .. controls +(0,1) and +(0,1) .. (8,2);
\draw[directed] (0.5,-2.5) .. controls +(0,-1) and +(0,-1) .. (-1,-2.5);
\draw[dotted] (4,0) -- (5,0);
\draw[line width=1pt] 
(-1,3.8) node[line width=0pt] (XXYY) {{\small $(X\otimes Y)^{\vee}$}}
(9,-3.8) node[line width=0pt] (XXYYu) {{\small $(X\otimes Y)^{\vee}$}};
\draw (X) -- (XY);
\draw (Y) -- (XY);
\draw (8,0) -- (8,2);
\draw (XXYYu) -- (9,2);
\draw (0.5,-2.5) -- (XY);
\draw (-1,-2.5) -- (XXYY);
\end{tikzpicture}
\;=\;
\begin{tikzpicture}[very thick,scale=0.8,color=blue!50!black, baseline]
\draw[line width=1pt] 
(0,3.8) node[line width=0pt] (XXYY) {{\small $(X\otimes Y)^{\vee}$}}
(0,-3.8) node[line width=0pt] (XXYYu) {{\small $(X\otimes Y)^{\vee}$}};
\draw (XXYY) -- (XXYYu);
\end{tikzpicture} \, . 
\ee
Composing both sides with
\be
\begin{tikzpicture}[very thick,scale=0.8,color=blue!50!black, baseline]
\draw[line width=1pt] 
(-2,-1) node[line width=0pt] (X) {{\small $Y^\vee$}}
(-1,-1) node[line width=0pt] (Y) {{\small $X^\vee$}}
(2,2) node[line width=0pt] (XY) {{\small $(X\otimes Y)^\vee$}}; 
\draw[directed] (0,0) .. controls +(0,1) and +(0,1) .. (-1,0);
\draw[directed] (1,0) .. controls +(0,2) and +(0,2) .. (-2,0);
\draw[directed] (2,0) .. controls +(0,-1) and +(0,-1) .. (0.5,0);
\draw (-1,0) -- (Y);
\draw (-2,0) -- (X);
\draw[dotted] (0,0) -- (1,0);
\draw (2,0) -- (XY);
\end{tikzpicture} 
\ee
``from below'', applying two more Zorro moves on the left-hand side and using the definitions of $\tcoev_{X\otimes Y}, \tev_{X}, \tev_{Y}$, we finally arrive at~\eqref{bends}. 

\subsection{Trace formula for defect action}\label{defacform}

We want to prove the explicit expression~\eqref{DlDr} for the action of a defect~$X$ on a bulk field~$\varphi$. As $\varphi\in\End_{\MFW}(I)\cong R/(\partial W)$, it suffices to compute the $(1,1)$-entry of the $(2\times 2)$-matrix
\be
\mathcal D_{l}(X)(\varphi) = \ev_{X} \circ (\id_{X^\vee} \otimes (\lambda_{X} \circ (\varphi\otimes\id_{X}) \circ  \lambda_{X}^{-1}))\circ \tcoev_{X} \, ,
\ee
because the other non-zero entry must be the same. Substituting the explicit expressions for $\ev_{X}, \lambda_{X}, \lambda^{-1}_{X}, \tcoev_{X}$, we find that the $(1,1)$-entry of~\eqref{DlDr} is equal to

\begin{align}
 & A_{X} \circ \mu(\varphi_{0}) \circ \left( \id \otimes [1\tc \id]\hat~ \right) \circ \left[ \frac{\id_{\check X_{1}} \tc \id_R\tc \big(\check d^X_{1}(x,a) - \check d^X_{1}(x,b)\big)}{a-b} \right]^{\!\wedge} \circ c_{X_{1}} \nonumber \\
=& \left[ \frac{1}{2\pi\I} \oint \frac{(a-b-x)}{x(W(x)-W(b))} \frac{\operatorname{tr}\left( \check\varphi_{0}(x) \check d^X_{0}(x,b) (\check d^X_{1}(x,a) - \check d^X_{1}(x,b)) \right)}{a-b} \D x \right]^{\!\wedge} \nonumber \\
=& -\left[ \frac{1}{2\pi\I} \oint \frac{\operatorname{tr}\left( \check\varphi_{0}(x) \check d^X_{0}(x,b) (\check d^X_{1}(x,a) - \check d^X_{1}(x,b)) \right) \D x}{x(W(x)-W(b))}  \right]^{\!\wedge} \label{ab0} \\
&\qquad +\left[ \frac{1}{2\pi\I} \oint \frac{\operatorname{tr}\left( \check\varphi_{0}(x) \check d^X_{0}(x,b) (\check d^X_{1}(x,a) - \check d^X_{1}(x,b)) \right) \D x}{(W(x)-W(b))(a-b)}  \right]^{\!\wedge} \label{ab0} \nonumber \\
=& - \left[ \frac{1}{2\pi\I} \oint \frac{\operatorname{tr}\left(\check d^X_{0}(x,b) \check d^X_{1}(x,a) \check\varphi_{0}(x)\right)\D x}{(W(x)-W(b))(a-b)}  \right]^{\!\wedge} . \nonumber
\end{align}
Here it was used that since $\End_{\MFW}(I)\cong R/(\partial W)$ we have $\hat a=\hat b$, the term $\check d^X_{0}(x,b) (\check d^X_{1}(x,a) - \check d^X_{1}(x,b))$ in line~\eqref{ab0} is zero. 

The expression for $\mathcal D_{r}(X)(\varphi)$ is proved analogously.


\begin{thebibliography}{BHLS}

\bibitem[ADD]{add0401}
S.~K. Ashok, E.~Dell'Aquila, and D.-E. Diaconescu, \textsl{Fractional {B}ranes in
  {L}andau-{G}inzburg {O}rbifolds}, \httpurl{projecteuclid.org/DPubS?service=UI&version=1.0&verb=Display&handle=euclid.atmp/1098389089}{Adv. Theor. Math. Phys. \textbf{8} (2004),
  461--513}, \href{http://www.arxiv.org/abs/hep-th/0401135}{[hep-th/0401135]}.

\bibitem[BG]{bg0503}
I.~Brunner and M.~R. Gaberdiel, \textsl{Matrix factorisations and permutation
  branes}, \doi{10.1088/1126-6708/2005/07/012}{JHEP \textbf{0507} (2005), 012},
  \href{http://www.arxiv.org/abs/hep-th/0503207}{[hep-th/0503207]}. 

\bibitem[BHLS]{bhls0305}
I.~Brunner, M.~Herbst, W.~Lerche, and B.~Scheuner, \textsl{Landau-{G}inzburg
  {R}ealization of {O}pen {S}tring {TFT}}, \doi{10.1088/1126-6708/2006/11/043}{JHEP \textbf{0611} (2003), 043},
  \href{http://www.arxiv.org/abs/hep-th/0305133}{[hep-th/0305133]}.

\bibitem[BR]{br0707.0922}
I.~Brunner and D.~Roggenkamp, \textsl{B-type defects in {L}andau-{G}inzburg
  models}, \doi{10.1088/1126-6708/2007/08/093}{JHEP \textbf{0708} (2007), 093},
  \href{http://arxiv.org/abs/0707.0922}{[0707.0922 [hep-th]]}.

\bibitem[BRR]{brr0909.0696}
I.~Brunner, D.~Roggenkamp, and S.~Rossi, \textsl{Defect Perturbations in Landau-Ginzburg Models}, \doi{10.1007/JHEP03(2010)015}{JHEP \textbf{1003} (2010), 015},
  \href{http://arxiv.org/abs/0909.0696}{[0909.0696 [hep-th]]}.

\bibitem[CR]{cr0909.4381}
N.~Carqueville and I.~Runkel,
{\it On the monoidal structure of matrix bi-factorisations}, \doi{10.1088/1751-8113/43/27/275401}{J. Phys. A: Math. Theor. \textbf{43} (2010), 275401},
\arxiv{0909.4381}{[0909.4381 [math-ph]]}.

\bibitem[Dy]{d0904.4713}
T.~Dyckerhoff, \textsl{Compact generators in categories of matrix factorizations},
  Duke Math. J. \textbf{159} (2011), 223--274,
  \href{http://arxiv.org/abs/0904.4713}{[arXiv:0904.4713]}.

\bibitem[DM]{dm1004.0687}
T.~Dyckerhoff and D.~Murfet,
{\it The Kapustin-Li formula revisited},
\arxiv{1004.0687}{[1004.0687 [math.AG]]}.

\bibitem[Ei]{eisen1980}
D.~Eisenbud, \textsl{Homological algebra with an application to group
  representations}, Trans. Amer. Math. Soc. \textbf{260} (1989), 35--64.

\bibitem[ERR]{err0508}
H.~Enger, A.~Recknagel, and D.~Roggenkamp, \textsl{Permutation branes and linear
  matrix factorisations}, \doi{10.1088/1126-6708/2006/01/087}{JHEP \textbf{0601} (2006), 087},
  \href{http://www.arxiv.org/abs/hep-th/0508053}{[hep-th/0508053]}.

\bibitem[FY]{freyd-yetter}
P.~J.~Freyd and D.~N.~Yetter,
{\it Braided compact closed categories with applications to low dimensional topology},
\doi{10.1016/0001-8708(89)90018-2}{Adv.\ Math.\ {\bf 77} (1989), 156--182}.

\bibitem[Fr3]{Frohlich:2006ch}
J.~Fr\"ohlich, J.~Fuchs, I.~Runkel and C.~Schweigert,
{\it Duality and defects in rational conformal field theory},
\doi{10.1016/j.nuclphysb.2006.11.017}{Nucl.\ Phys.\ B {\bf 763} (2007), 354--430}, \arxiv{hep-th/0607247}{[hep-th/0607247]}.

\bibitem[FRS1]{tft1}
J.~Fuchs, I.~Runkel and C.~Schweigert,
{\it TFT construction of RCFT correlators. I: Partition functions},
\doi{10.1016/S0550-3213(02)00744-7}{Nucl.\ Phys.\ B {\bf 646} (2002), 353--497}, \arxiv{hep-th/0204148}{[hep-th/0204148]}.

\bibitem[FRS2]{tft3}
J.~Fuchs, I.~Runkel and C.~Schweigert,
{\it TFT construction of RCFT correlators. III: Simple currents},
\doi{10.1016/j.nuclphysb.2004.05.014}{Nucl.\ Phys.\  B {\bf 694} (2004), 277--353}, \arxiv{hep-th/0403157}{[hep-th/0403157]}.

\bibitem[FRS3]{Fuchs:2007vk}
J.~Fuchs, I.~Runkel and C.~Schweigert,
{\it The fusion algebra of bimodule categories},
\doi{10.1007/s10485-007-9102-7}{Appl.\ Cat.\ Str.\ {\bf 16} (2008), 123--140},
\arxiv{math.CT/0701223}{[math.CT/0701223]}.

\bibitem[HL]{hl0404184}
M.~Herbst and C.~I.~Lazaroiu, \textsl{Localization and traces in open-closed topological Landau-Ginzburg models}, \doi{10.1088/1126-6708/2005/05/044}{JHEP \textbf{0505} (2005), 044},
     \arxiv{hep-th/0404184}{[hep-th/0404184]}.

\bibitem[Ho]{h0401}
K.~Hori, \textsl{Boundary {RG} {F}lows of {$N=2$} {M}inimal {M}odels}, Banff
  2003, {M}irror symmetry {V} (N.~Yui, S.-T. Yau, and J.D. Lewis, eds.), Am.
  Math. Soc., 2006,
  pp.~381--405,
     \arxiv{hep-th/0401139}{[hep-th/0401139]}.
     
 \bibitem[HWa]{hw0404196}
K.~Hori and J.~Walcher, \textsl{F-term equations near Gepner points}, \doi{10.1088/1126-6708/2005/01/008}{JHEP \textbf{0501} (2005), 008}, \arxiv{hep-th/0404196}{[hep-th/0404196]}.

\bibitem[HWe]{howewest}
P.~Howe and P.~West, 
  \textsl{{$N=2$} superconformal models, {L}andau-{G}insburg hamiltonians
  and the $\epsilon$ expansion}, \doi{10.1016/0370-2693(89)91619-5}{Phys. Lett. B \textbf{223} (1989), 377--385};
  \textsl{Chiral {C}orrelators {I}n {L}andau-{G}inzburg
  {T}heories {A}nd {$N=2$} {S}uperconformal {M}odels},
      \doi{10.1016/0370-2693(89)90950-7}{
   Phys. Lett. B
  \textbf{227} (1989), 397--405}; 
  \textsl{Fixed points in multifield {L}andau-{G}insburg models}, \doi{10.1016/0370-2693(90)90068-H}{Phys. Lett. B \textbf{244} (1990), 270--274}.

\bibitem[Hu]{Huang2005}
Y.-Z.~Huang,
{\it Rigidity and modularity of vertex tensor categories},
\doi{10.1142/S0219199708003083}{Commun.\ Contemp.\ Math.\ {\bf 10} (2008), 871--911}, \arxiv{math/0502533}{[math.QA/0502533]}.

\bibitem[JS]{JoyalStreet2}
A.~Joyal and R.~Street,
{\it The geometry of tensor calculus {II}}, unpublished, available on R.~Street's \href{http://www.math.mq.edu.au/~street/GTCII.pdf}{website}.

\bibitem[KST]{kst0511155}
H.~Kajiura, K.~Saito, and A.~Takahashi, \textsl{Matrix Factorizations and Representations of Quivers II: type ADE case}, \doi{10.1016/j.aim.2006.08.005}{Adv. in Math. \textbf{211} (2007), 327--362}, \arxiv{math.AG/0511155}{[math.AG/0511155]}. 

\bibitem[KL1]{kl0210}
A.~Kapustin and Y.~Li, \textsl{D-branes in {L}andau-{G}inzburg {M}odels and
  {A}lgebraic {G}eometry}, \doi{10.1088/1126-6708/2003/12/005}{JHEP \textbf{0312} (2003), 005},
  \href{http://www.arxiv.org/abs/hep-th/0210296}{[hep-th/0210296]}.

\bibitem[KL2]{kl0305}
A.~Kapustin and Y.~Li, \textsl{Topological {C}orrelators in {L}andau-{G}inzburg {M}odels with {B}oundaries}, \href{http://projecteuclid.org/DPubS?service=UI&version=1.0&verb=Display&handle=euclid.atmp/1112627039}{Adv. Theor. Math. Phys. \textbf{7} (2004), 727--749},
  \href{http://www.arxiv.org/abs/hep-th/0305136}{[hep-th/0305136]}.

\bibitem[KaR]{kr0405232}
A.~Kapustin and L.~Rozansky, \textsl{On the relation between open and closed
  topological strings}, \doi{10.1007/s00220-004-1227-z}{Commun. Math. Phys. \textbf{252} (2004), 393--414},
  \href{http://arxiv.org/abs/hep-th/0405232}{[hep-th/0405232]}.
  
\bibitem[KRS]{krs0810.5415}
A.~Kapustin, L.~Rozansky, and N.~Saulina \textsl{Three-dimensional topological field theory and symplectic algebraic geometry I}, 
  \href{http://arxiv.org/abs/0810.5415}{[0810.5415 [hep-th]]}.  

\bibitem[KMS]{kms1989}
D.~A. Kastor, E.~J. Martinec, and S.~H. Shenker, \textsl{{RG} {F}low in {$N=1$}
  {D}iscrete {S}eries}, \doi{10.1016/0550-3213(89)90060-6}{Nucl. Phys. B \textbf{316} (1989), 590--608}.

\bibitem[Ke]{k0601185}
B.~Keller,
{\it On differential graded categories},
Eur. Math. Soc., International Congress of Mathematicians {\bf Vol. II} (2006), 151--190, \arxiv{math/0601185}{[math/0601185]}.

\bibitem[KhR]{Khovanov:2004}
M.~Khovanov and L.~Rozansky
{\it Matrix factorizations and link homology},
\doi{doi:10.4064/fm199-1-1}{Fundamenta Mathematicae {\bf 199} (2008), 1--91}, \arxiv{math/0401268}{[math.QA/0401268]}.

\bibitem[Kt]{KontsevichU}
M.~Kontsevich, unpublished.

\bibitem[La]{l0312}
C.~I.~Lazaroiu, \textsl{On the boundary coupling of topological
  {L}andau-{G}inzburg models}, \doi{10.1088/1126-6708/2005/05/037}{JHEP \textbf{0505} (2005), 037},
  \arxiv{hep-th/0312286}{[hep-th/0312286]}.

\bibitem[LMZ]{Calinetal}
C.~I.~Lazaroiu, D.~McNamee, and A.~Zejak, work in progress; D.~McNamee, \textsl{On the mathematical structure of topological defects in Landau-Ginzburg models}, Master thesis. 

\bibitem[Lu1]{Lurietalk}
J.~Lurie,
{\it TQFT and the Cobordism Hypothesis}, \href{http://lab54.ma.utexas.edu:8080/video/lurie.html}{lectures at the University of Texas at Austin}, January 2009.

\bibitem[Lu2]{l0905.0465}
J.~Lurie,
{\it On the Classification of Topological Field Theories},
\arxiv{0905.0465}{[0905.0465 [math.CT]]}.

\bibitem[MaL]{catforwormath}
S.~Mac~Lane, \textsl{Categories for the working mathematician}, 2nd ed.,
  Springer, 1998.

\bibitem[Ml]{maltsiniotis}
G.~Maltsiniotis,
{\it Traces dans les cat\'egories mono\"\i dales, dualit\'e et cat\'egories mono\"\i dales fibr\'ees},
\httpurl{www.numdam.org:80/numdam-bin/fitem?id=CTGDC_1995__36_3_195_0}{Cahiers de Topologie {\bf 36} (1995), 195--288}. 

\bibitem[Mr]{m1989}
E.~J.~Martinec, \textsl{Algebraic {G}eometry and {E}ffective {L}agrangians},
  \doi{10.1016/0370-2693(89)90074-9}{Phys. Lett.~B \textbf{217} (1989), 431}.

\bibitem[Mg]{mu0804.3587}
M.~M\"uger,
{\it Tensor categories: A selective guided tour},
\arxiv{0804.3587}{[0804.3587 [math.CT]]}.

\bibitem[Mf]{m0912.1629}
D.~Murfet,
{\it Residues and duality for singularity categories of isolated Gorenstein singularities},
\arxiv{0912.1629}{[0912.1629 [math.AC]]}.

\bibitem[Ne]{Neeman}
A.~Neeman, \textsl{Triangulated {C}ategories}, Annals of Mathematics Studies,
  Princeton University Press, 2001.

\bibitem[Or]{o0302304}
D.~Orlov,
{\it Triangulated categories of singularities and {D}-branes in {L}andau-{G}inzburg models},
\doi{10.1070/SM2006v197n12ABEH003824}{Tr. Mat. Inst. Steklova {\bf 246} (2004), 240--262}, \arxiv{math.AG/0302304}{[math.AG/0302304]}.

\bibitem[PZ]{Petkova:2000ip}
V.~B.~Petkova and J.-B.~Zuber,
{\it Generalised twisted partition functions},
\doi{10.1016/S0370-2693(01)00276-3}{Phys.\ Lett.\ B {\bf 504} (2001), 157--164}, \arxiv{hep-th/0011021}{[hep-th/0011021]}.

\bibitem[PV]{pv1002.2116}
A.~Polishchuk and A.~Vaintrob,
{\it Chern characters and Hirzebruch-Riemann-Roch formula for matrix factorizations},
\arxiv{1002.2116}{[1002.2116 [math.AG]]}.

\bibitem[RS]{Runkel:2008gr}
 I.~Runkel and R.~R.~Suszek,
 {\it Gerbe-holonomy for surfaces with defect networks},
 \arxiv{0808.1419}{[0808.1419 [hep-th]]}.

\bibitem[Sch]{s0309131}
P.~Schauenburg, 
{\it On the Frobenius-Schur indicators for quasi-Hopf algebras},
\doi{10.1016/j.jalgebra.2004.08.015}{J.\ Algebra {\bf 282} (2004), 129--139},
\arxiv{math.qa/0309131}{[math.QA/0309131]}.

\bibitem[Se]{s0904.1339}
E.~Segal,
{\it The closed state space of affine Landau-Ginzburg B-models},
\arxiv{0904.1339}{[0904.1339 [math.AG]]}.

\bibitem[Sh]{s0710.1937}
D.~Shklyarov,
{\it Hirzebruch-Riemann-Roch theorem for DG algebras},
\arxiv{0710.1937}{[0710.1937 [math.KT]]}.

\bibitem[To1]{t0408337}
B.~To\"en, 
{\it The homotopy theory of dg-categories and derived {M}orita theory},
\doi{10.1007/s00222-006-0025-y}{Invent. Math. {\bf 167(3)} (2007), 615--667}, \arxiv{math/0408337}{[math/0408337]}.

\bibitem[To2]{Toendglectures}
B.~To\"en, 
{\it Lectures on {DG}-categories},
available on author's \href{http://ens.math.univ-montp2.fr/~toen/swisk.pdf}{website}.
  
\bibitem[Va]{v1991}
C.~Vafa, \textsl{Topological {L}andau-{G}inzburg {M}odels}, \doi{10.1142/S0217732391000324}{Mod. Phys. Lett.~A
  \textbf{6} (1991), 337--346}.

\bibitem[VW]{vw1989}
C.~Vafa and N.~Warner, \textsl{Catastrophes and the classification of conformal
  theories}, \doi{10.1016/0370-2693(89)90473-5}{Phys. Lett.~B \textbf{218} (1989), 51}.

\bibitem[WA]{Wong:1994pa}
 E.~Wong and I.~Affleck,
 {\it Tunneling in quantum wires: A Boundary conformal field theory
   approach},
 \doi{10.1016/0550-3213(94)90479-0}{Nucl.\ Phys.\ B {\bf 417} (1994), 403--438},  \arxiv{cond-mat/9311040}{[cond-mat/9311040].}

\bibitem[Yo]{yoshinoTP}
Y.~Yoshino, \textsl{Tensor products of matrix factorizations}, \href{http://projecteuclid.org/DPubS?service=UI&version=1.0&verb=Display&handle=euclid.nmj/1118766411}{Nagoya Math. J.
  \textbf{152} (1998), 39--56}.

\end{thebibliography}
\end{document}